\documentclass[final]{siamltex}

\usepackage{amssymb}
\usepackage{amsmath}
\usepackage{graphicx}
\usepackage{subfig}
\usepackage{algorithm}
\usepackage{algpseudocode}
\usepackage{color}
\usepackage{psfrag}
\usepackage{url}
\usepackage{array}
\usepackage{slashbox}
\usepackage{wrapfig}
\usepackage{float}
\newcolumntype{S}{>{\centering\arraybackslash} m{0.8cm}}
\newcolumntype{F}{>{\centering\arraybackslash} m{.2\linewidth}}

\providecommand{\tabularnewline}{\\}

\begin{document}

\title{Exact and Stable Recovery of Rotations for Robust Synchronization}

\author{Lanhui Wang%
\thanks{The Program in Applied and Computational Mathematics (PACM), Princeton University, Fine Hall, Washington Road, Princeton, NJ 08544-1000, \texttt{lanhuiw@math.princeton.edu}, Corresponding author. Tel.: +1 609 258 5785; fax: +1 609 258 1735. }
\and Amit Singer%
\thanks{Department of Mathematics and PACM, Princeton University, Fine Hall, Washington Road, Princeton, NJ 08544-1000,\texttt{amits@math.princeton.edu} } }

\date{}
\maketitle

\maketitle

\begin{abstract}
The synchronization problem over the special orthogonal group $SO(d)$ consists of estimating a set
of unknown rotations $R_1,R_2,\ldots,R_n$ from noisy measurements
of a subset of their pairwise ratios $R_{i}^{-1}R_{j}$. The problem has found applications in computer vision, computer graphics, and sensor network localization, among others. Its least squares solution can be approximated by either spectral relaxation or semidefinite programming followed by a rounding procedure, analogous to the approximation algorithms of \textsc{Max-Cut}. The contribution of this paper is three-fold: First, we introduce a robust penalty function involving the sum of unsquared deviations and derive a relaxation that leads to a convex optimization problem; Second, we apply the alternating direction method to minimize the penalty function; Finally, under a specific model of the measurement noise and for both complete and random measurement graphs, we prove that the rotations are exactly and stably recovered, exhibiting a phase transition behavior in terms of the proportion of noisy measurements. Numerical simulations confirm the phase transition behavior for our method as well as its improved accuracy compared to existing methods.
%%%% If classification number provided then
%\\
%2000 Math Subject Classification: 34K30, 35K57, 35Q80,  92D25
\end{abstract}
\section{Introduction}
\addcontentsline{toc}{section}{Introduction}
The synchronization problem over the special orthogonal group $SO(d)$ of rotations in $\mathbb{R}^d$
\begin{equation}
\label{SOd}
SO(d) = \{R \in \mathbb{R}^{d\times d} \,:\, R^T R = RR^T = I_d,\;\; \det R = 1 \}
\end{equation}
consists of estimating a set of $n$ rotations $R_1,\ldots,R_n \in SO(d)$ from a subset of (perhaps noisy) measurements $R_{ij}$ of their ratios $R_i^{-1}R_j$. The subset of available ratio measurements is viewed as the edge set of an undirected graph $\mathcal{G}=(\mathcal{V},\mathcal{E})$, with $|\mathcal{V}|=n$. The goal is to find $R_1,\ldots,R_n$ that satisfy \begin{equation}
R_i^{-1}R_j \approx R_{ij},\quad \mbox{for } (i,j)\in \mathcal{E}.
\end{equation}
Synchronization over the rotation group $SO(d)$ has many applications. Synchronization over $SO(2)$
plays a major role in the framework of angular embedding for ranking and for image reconstruction from pairwise intensity differences \cite{Yu2009,Yu2012} and for a certain algorithm for sensor network localization
\cite{sensor_network}. Synchronization over $SO(3)$ is invoked by many algorithms for structure from motion in computer vision \cite{Martinec2007,Tron2009,Hartley2011,Basri2012}, by algorithms for global alignment of 3-D scans in computer graphics \cite{3d_scan}, and by algorithms for finding 3-D structures of molecules using NMR spectroscopy \cite{Cucuringu2012} and cryo-electron microscopy \cite{Shkolnisky2012,angle_cla}. A closely related problem in terms of applications and methods is the synchronization over the orthogonal group $O(d)$, where the requirement of positive determinant in (\ref{SOd}) is alleviated. We remark that the algorithms and analysis presented in this paper follow seamlessly to the case of $O(d)$. We choose to focus on $SO(d)$ only because this group is encountered more often in applications.

If the measurements $R_{ij}$ are noiseless and the graph $\mathcal{G}$ is connected then
the synchronization problem can be easily solved by considering a spanning tree in
$\mathcal{G}$, setting the rotation of the root node arbitrarily, and determining all other rotations by traversing the tree
while sequentially multiplying the rotation ratios. The rotations obtained
in this manner are uniquely determined up to a global rotation, which is the intrinsic degree of freedom of the synchronization problem. However,
when the ratio measurements are corrupted by noise, the spanning
tree method suffers from accumulation of errors. Estimation methods that use all available ratio measurements and exploit the redundancy of graph cycles are expected to perform better.

If some (though possibly not all) pairwise measurements are noise-free, then a cycle-based algorithm can be used to find the noise-free ratio measurements. Specifically, in order to determine if a ratio measurement is ``good'' (noise-free) or ``bad'' (corrupted by random noise), one can examine cycles in the graph that include that edge and check their consistency. A consistent cycle is a cycle for which sequentially multiplying the rotation ratios along the cycle results in the identity rotation. Under the random noise assumption, ratio measurements along consistent cycles are almost surely ``good'', and if the subgraph associated with the ``good" measurements is connected, then the spanning tree method can be used to determine the rotations. However, cycle-based algorithms have two main weaknesses. First, the computational complexity of cycle-based algorithms increases exponentially with the cycle length. Second, the cycle-based algorithms are unstable to small perturbations on the ``good'' ratio measurements.

Methods that are based on least squares have been proposed and analyzed in
the literature. While the resulting problem is non-convex, the solution to the least
squares problem is approximated by either a spectral relaxation (i.e., using leading
eigenvectors) or by semidefinite programming (SDP) (see, e.g., \cite{ang_sync,Howard2010,Yu2009,Bandeira2012}).
Typically in applications, the ratio measurements generally consist of noisy
inliers, which are explained well by the rotations $R_1,\ldots,R_n$, along
with outliers, that have no structure. The least squares method is
sensitive to these outliers. 

In this paper we propose to estimate the rotations
by minimizing a different, more robust self consistency error, which
is the sum of unsquared residuals \cite{L1_Nyquist,L1_spath_watson,irls_t},
rather than the sum of squared residuals. The minimization problem
is semidefinite relaxed and solved by the alternating direction method.
Moreover, we prove that under some conditions the rotations $R_1,\ldots,R_n$
can be exactly and stably recovered (up to a global rotation), see Theorems \ref{thm:p_c}, \ref{thm:stability_weak} and \ref{thm:stability_strong} for the complete graph, and Theorems \ref{thm:p_c_missing}-\ref{thm:stability_strong_missing} for random graphs. Our numerical experiments demonstrate that
the new method significantly improves the estimation of rotations, and
in particular, achieving state-of-the-art results.

The paper is organized as follows: In Section \ref{sec:LSQR} we review existing SDP and spectral relaxation methods for approximating the least squares solution. In Section \ref{sec:LUD} we derive the more robust least unsquared deviation (LUD) cost function and its convex relaxation. In Section \ref{sec:exact} we introduce the noise model and prove conditions for exact recovery by the LUD method. In Section \ref{sec:stability} we prove that the recovery of the rotations is stable to noise. In Section \ref{sec:missing_entries}, we generalize the results to the case of random (incomplete) measurement graphs. In Section \ref{sec:ADM}, we discuss the application of the alternating direction method for solving the LUD optimization problem. The results of numerical experiments on both synthetic data as well as for global alignment of 3D scans are reported in Section \ref{sec:exp}. Finally, Section \ref{sec:summary} is a summary.

\section{Approximating the Least Squares Solution}
\addcontentsline{toc}{section}{Approximating the Least Squares Solution}
\label{sec:LSQR}
In this section we overview existing relaxation methods that attempt to approximate the least squares solution. The least squares solution to synchronization is the set of rotations $R_1,\ldots,R_n$ in $SO(d)$ that minimize the sum of squared deviations
\begin{equation}
\label{LS}
\min_{R_1,\ldots,R_n \in SO(d)} \sum_{(i,j) \in \mathcal{E} } w_{ij} \left\|{R}_i^{-1}{R}_j - R_{ij} \right\|^2,
\end{equation}
where $\| \cdot \|$ denotes the Frobenius norm\footnote{The Frobenius norm of an $m\times n$ matrix $A$ is defined as $\|A\| = \sqrt{\sum_{i=1}^m \sum_{j=1}^n  A_{ij}^2}$.}, $w_{ij}$ are non-negative weights that reflect the measurement precisions\footnote{If all measurements have the same precision, then the weights take the simple form $w_{ij}=1$ for $(i,j)\in \mathcal{E}$ and $w_{ij}=0$ otherwise.}, and $R_{ij}$ are the noisy measurements. The feasible set $SO(d)^n = SO(d)\times \cdots \times SO(d)$ of the minimization problem (\ref{LS}) is however non-convex. Convex relaxations of (\ref{LS}) involving SDP and spectral methods have been previously proposed and analyzed.

\subsection{Semidefinite Programming Relaxation}
\label{sec:SDR}
\addcontentsline{toc}{subsection}{Semidefinite Programming Relaxation}
Convex relaxation using SDP was introduced in \cite{ang_sync} for $SO(2)$ and in \cite{Basri2012} for $SO(3)$ and is easily generalized for any $d$. The method draws similarities with the Goemans and Williamson approximation algorithm to \textsc{Max-Cut} that uses SDP \cite{Goemans1995}.

The first step of the relaxation involves the observation that the least squares problem (\ref{LS}) is equivalent to the maximization problem
\begin{equation}
\label{max}
\max_{{R}_1,\ldots,{R}_n \in SO(d)} \sum_{(i,j) \in \mathcal{E} } w_{ij} \operatorname{Tr} ({R}_i^{-1}{R}_j R_{ij}^T),
\end{equation}
due to the fact that $\|R_i^{-1}R_j\|^2=\|R_{ij}\|^2=d$.

The second step of the relaxation introduces the matrix $G$ of size
$n\times n$ whose entries
\begin{equation}
\label{G}
G_{ij} = R_i^T R_j
\end{equation}
are themselves matrices of size $d\times d$, so that the overall size of $G$ is $nd \times nd$. The matrix $G$ admits the decomposition
\begin{equation}
G = R^T R,
\end{equation}
where $R$ is a matrix of size $d\times nd$ given by
\begin{equation}
R = \left[
      \begin{array}{cccc}
        R_1 & R_2 & \cdots & R_n \\
      \end{array}
    \right].
\end{equation}
The objective function in (\ref{max}) can be written as $\operatorname{Tr}(GC)$, where the entries of $C$ are given by $C_{ij} = w_{ij}R_{ij}^T$ (notice that $C$ is symmetric, since $R_{ij}^T = R_{ji}$ and $w_{ij}=w_{ji}$). The matrix $G$ has the following properties:
\begin{enumerate}
\item $G \succeq 0$, i.e., it is positive semidefinite (PSD).
\item $G_{ii} = I_d$ for $i=1,\ldots, n$.
\item $\operatorname{rank}(G)=d$.\label{rank}
\item $\det(G_{ij})=1$ for $i,j=1,\ldots,n$.\label{det}
\end{enumerate}
Relaxing properties (\ref{rank}) and (\ref{det}) leads to the following SDP:
\begin{align}
& \max_{G} \operatorname{Tr}(GC) \label{maxG}\\
\text{s.t. } & G \succeq 0 \nonumber \\
& G_{ii}=I_d, \quad \mbox{for } i=1,2,\ldots,n. \nonumber
\end{align}
Notice that for $d=1$ (\ref{maxG}) reduces to the SDP for \textsc{Max-Cut} \cite{Goemans1995}. Indeed, synchronization over $O(1) \cong \mathbb{Z}_2$ is equivalent to \textsc{Max-Cut}. \footnote{The case of $SO(2)$ is special in the sense that it is possible to represent group elements as complex-valued numbers, rendering $G$ a complex-valued Hermitian positive semidefinite matrix of size $n\times n$ (instead of a $2n\times 2n$ real valued positive semidefinite matrix).}

The third and last step of the relaxation involves the rounding procedure. The rotations $R_1,\ldots,R_n$ need to be obtained from the solution $G$ to (\ref{maxG}). The rounding procedure can be either random or deterministic. In the random procedure, a matrix $Q$ of size $nd \times d$ is sampled from the uniform distribution over matrices with $d$ orthonormal columns in $\mathbb{R}^{nd}$ (see \cite{Mezzadri2007} for description of the sampling procedure). The Cholesky decomposition of $G = LL^T$ is computed, and the product $LQ$ of size $nd\times d$ is formed, viewed as $n$ matrices of size $d\times d$, denoted by $(LQ)_1, \ldots, (LQ)_n$. The estimate for the inverse of the $i$'th rotation $\hat{R}_i^T$ is then obtained via the singular value decomposition (SVD) (equivalently, via the polar decomposition) of $(LQ)_i$ as (see, e.g., \cite{Higham1986,Moakher2003,Sarlette2009})
\begin{equation}
\label{SVD}
(LQ)_i = U_i \Sigma_i V_i^T, \quad J = \left[
                                         \begin{array}{cc}
                                           I_{d-1} & 0 \\
                                           0 & \det{U_i V_i^T} \\
                                         \end{array}
                                       \right],\quad
\hat{R}_i^T = U_i J V_i^T.
\end{equation}
(the only difference for synchronization over $O(d)$ is that $\hat{R}_i^T = U_i V_i^T$, excluding any usage of $J$).
In the deterministic procedure, the top $d$ eigenvectors $v_1,\ldots,v_d \in \mathbb{R}^{nd}$ of $G$ corresponding to its $d$ largest eigenvalues are computed. A matrix $T$ of size $nd\times d$ whose columns are the eigenvectors is then formed, i.e., $T = \left[
                                                                                         \begin{array}{cccc}
                                                                                           v_1 & \cdots & v_d \\
                                                                                         \end{array}
                                                                                       \right]$.
As before, the matrix $T$ is viewed as $n$ matrices of size $d\times d$, denoted $T_1,\ldots,T_n$ and the SVD procedure detailed in (\ref{SVD}) is applied to each $T_i$ (instead of $(GQ)_i$) to obtain $R_i$.

We remark that the least squares formulation (\ref{LS}) for synchronization of rotations is an instance of quadratic optimization problems under orthogonality constraints (\textsc{Qp-Oc}) \cite{Nemirovski2007, So09}. Other applications of \textsc{Qp-Oc} are the generalized orthogonal Procrustes problem and the quadratic assignment problem. Different semidefinite relaxations for the Procrustes problem have been suggested in the literature \cite{Nemirovski2007, NRV12, So09}. In \cite{Nemirovski2007}, for the problem (\ref{max}), the orthogonal constraints $R_i^T R_i = I_d$ can be relaxed to  $\left\|R_i\right\| \leq 1$, and the resulting problem can be converted to a semidefinite program with one semidefinite constraint for a matrix of size $nd^2\times nd^2$, $2n$ semidefinite constraints from matrices of size $d \times d$ (see (55) in \cite{Nemirovski2007}) and some linear constraints. Thus, compared with that relaxation, the one we use in (\ref{maxG}) has a lower complexity, since the problem (\ref{maxG}) has size $nd\times nd$. As for the approximation ratio of the relaxation, if $C$ is positive semidefinite (as in the case of the Procrustes problem), then there is a constant approximation ratio for the relaxation (\ref{maxG}) for the groups $O(1)$ and $SO(2)$ \cite{SDP}. When the matrix  $C$ is not positive semidefinite, an approximation algorithm with ratio $\Omega (1/\log n)$ is given in \cite{SDP} for the cases over $O(1)$ and $SO(2)$. 

\subsection{Spectral Relaxations}
\addcontentsline{toc}{subsection}{Spectral Relaxations}
Spectral relaxations for approximating the least squares solution have been previously considered in \cite{ang_sync,Howard2010,Yu2009,Yu2012,Martinec2007,sensor_network,Bandeira2012,3d_scan}. All methods are based on eigenvectors of the graph connection Laplacian (a notion that was introduced in \cite{SingerWu2012}) or one of its normalizations. The graph connection Laplacian, denoted $L_1$ is constructed as follows: define the symmetric matrix $W_1 \in \mathbb{R}^{nd\times nd}$, such that the $(i,j)$'th $d\times d$ block is given by $(W_1)_{ij} = w_{ij}R_{ij}$. Also, let $D_1\in \mathbb{R}^{nd\times nd}$, be a diagonal matrix such that $(D_1)_{ii} = d_iI_d$, where $d_i = \sum_j w_{ij}$. The graph
connection Laplacian $L_1$ is defined as
\begin{equation}
L_1 = D_1 - W_1.
\end{equation}
It can be verified that $L_1$ is PSD, and that in the noiseless case $L_1 R^T = 0$. The least $d$ eigenvectors of $L_1$ corresponding to its smallest eigenvalues, followed by the SVD procedure for rounding (\ref{SVD}) can be used to recover the rotations. A slightly modified procedure that uses the eigenvectors of the normalized graph connection Laplacian
\begin{equation}
\mathcal{L}_1 = D_1^{-1/2} L_1 D_1^{-1/2} = I_{nd} - D_1^{-1/2}W_1 D_1^{-1/2}
\end{equation}
is analyzed in \cite{Bandeira2012}. Specifically, Theorem 10 in \cite{Bandeira2012} bounds the least squares cost (\ref{LS}) incurred by this approximate solution from above and below in terms of the eigenvalues of the normalized graph connection Laplacian and the second eigenvalue of the normalized Laplacian (the latter reflecting the fact that synchronization is easier on well-connected graphs, or equivalently, more difficult on graphs with bottlenecks). This generalizes a previous result obtained by \cite{Trevisan2009} that considers a spectral relaxation algorithm for \textsc{Max-Cut} that achieves a non-trivial approximation ratio.

\section{Least Unsquared Deviation (LUD) and Semidefinite Relaxation}
\label{sec:LUD}
\addcontentsline{toc}{section}{Least Unsquared Deviation (LUD) and Semidefinite Relaxation}
As mentioned earlier, the least squares approach may not be optimal when
a large proportion of the measurements are outliers  \cite{L1_Nyquist,L1_spath_watson,irls_t}. To guard the orientation
estimation from outliers, we replace the sum of squared residuals
in (\ref{LS}) with the more robust sum of unsquared residuals
\begin{equation}
\min_{R_{1},\ldots,R_{n}\in SO\left(d\right)}\sum_{\left(i,j\right)\in \mathcal{E}}\left\| R_i^{-1}R_{j}-R_{ij}\right\|, \label{eq:unsquared_res}
\end{equation}
to which we refer as LUD.\footnote{For simplicity we consider the case where $w_{ij}=1$ for $(i,j)\in \mathcal{E}$. In general, one may consider the minimization of $\sum_{\left(i,j\right)\in \mathcal{E}} w_{ij} \left\| R_i^{-1}R_{j}-R_{ij}\right\|$.}
The self consistency error given in (\ref{eq:unsquared_res}) mitigates
the contribution from large residuals that may result from outliers.
However, the problem (\ref{eq:unsquared_res}) is non-convex and therefore
extremely difficult to solve if one requires the matrices $R_{i}$
to be rotations, that is, when adding the orthogonality and determinant constraints of $SO(d)$ given in (\ref{SOd}).

Notice that the cost function (\ref{eq:unsquared_res}) can be rewritten using the Gram matrix $G$ that was defined earlier in (\ref{G}) for the SDP relaxation. Indeed, the optimization (\ref{eq:unsquared_res}) is equivalent to
\begin{equation}
\min_{R_{1},\ldots,R_{n}\in SO\left(d\right)}\sum_{\left(i,j\right)\in \mathcal{E}}\left\| G_{ij}-R_{ij}\right\|.
\end{equation}
Relaxing the non-convex rank and determinant constraints as in the SDP relaxation leads to the following natural convex relaxation of the optimization problem (\ref{eq:unsquared_res}):
\begin{equation}
\min_{G}\sum_{\left(i,j\right)\in \mathcal{E}}\left\Vert G_{ij}-R_{ij}\right\Vert \text{ s.t. }G_{ii}=I_{d},\text{ and }G\succcurlyeq0.\label{eq:min_G}
\end{equation}
This type of relaxation is often referred to as semidefinite relaxation (SDR) \cite{SDP,SDR}. Once $G$ is found, either the deterministic or random procedures for rounding can be used to determine the rotations $R_1,\ldots,R_n$.

\section{Exact Recovery of the Gram matrix $G$}
\label{sec:exact}
\addcontentsline{toc}{section}{Exact Recovery of the Gram matrix $G$}
\subsection{Main Theorem}
Consider the Gram matrix $G$ as $G=\left(G_{ij}\right)_{i,j=1,\ldots,n}$ where $G_{ij}=R_{i}^{T}R_{j}$
is obtained by solving the minimization problem (\ref{eq:min_G}). We will show that for a certain probabilistic measurement model, the Gram matrix $G$ is exactly recovered with high probability (w.h.p.).
Specifically, in our model, the measurements $R_{ij}$ are given by
\begin{equation}
R_{ij}=\delta_{ij}R_i^TR_j+\left(1-\delta_{ij}\right)\tilde{R}_{ij}\label{eq:noise}
\end{equation}
for $\left(i,j\right)\in \mathcal{E}$, where $\delta_{ij}$ are i.i.d indicator Bernoulli random variables with probability $p$ (i.e. $\delta_{ij}=1$ with probability $p$ and $\delta_{ij}=0$ with probability $1-p$). The incorrect measurements $\tilde{R}_{ij}$ are i.i.d samples from the uniform (Haar) distribution over the group $SO\left(d\right)$. Assume all pairwise
ratios are measured, that is, $\mathcal{E}=\left\{ \left(i,j\right)|i,j=1\ldots,n\right\} $\footnote{The measurements in (\ref{eq:noise}) satisfy $R_{ij}=R_{ji}^T$, $R_{ii} = I_d$.}.
Let $\mathcal{E}_{c}$ denote the index set of correctly measured rotation
ratios $R_{ij}$, that is, $\mathcal{E}_{c}=\left\{ \left(i,j\right)|R_{ij}=R_i^TR_j\right\} $. In fact, $\mathcal{E}_c$ is the edge set of a realization of a random graph drawn from the Erd\H{o}s-R\'enyi model $\mathcal{G}(n,p)$.
In the remainder of this section we shall prove the existence of a critical probability, denoted $p^*_c(d)$, such that for all $p> p^*_c(d)$ the program (\ref{eq:min_G}) recovers $G$ from the measurements (\ref{eq:noise}) w.h.p. (that tends to 1 as $n\rightarrow\infty$). In addition, we give an explicit upper bound $p_c(d)$ for $p^*_c(d)$.

\begin{theorem}
\label{thm:p_c}
Assume that all pairwise ratio measurements $R_{ij}$ are generated according to (\ref{eq:noise}). Then there exists a critical probability $p^*_c(d)$ such that when $p > p^*_c(d)$, the Gram matrix $G$ is exactly recovered by the solution to the optimization problem (\ref{eq:min_G}) w.h.p. (as $n\to \infty$). Moreover, an upper bound $p_c(d)$ for $p^*_c(d)$ is
\begin{equation}
p_c(d)=1-\left(\frac{-c_{1}\left(d\right)+\sqrt{c_{1}\left(d\right)^2+8\left(c\left(d\right)+2/\sqrt{d}\right)/\sqrt{d}}}{2\left(c\left(d\right)+2/\sqrt{d}\right)}\right)^{2}\label{eq:p_c},
\end{equation}
where  $c\left(d\right)$ and $c_1(d)$ are constants defined as % is a constant
%such that
\begin{equation}
c\left(d\right)=\frac{1}{d}\mathbb{E}\left(\text{Tr}\left(\frac{I_{d}-R}{\left\Vert I_{d}-R\right\Vert }\right)\right)\label{eq:def_cd},
\end{equation}
%and the constant $c_1\left(d\right)$ is defined as
\begin{equation}
c_{1}\left(d\right)=\sqrt{\frac{1-c\left(d\right)^{2}d}{2}},\label{eq:c_1}
\end{equation}
where the expectation is w.r.t. the rotation $R$ that is distributed uniformly at random.
%Particularly, (using results in Appendix A) we have $c\left(2\right)=\frac{\sqrt{2}}{\pi}$,
%$c\left(3\right)=\frac{8\sqrt{2}}{9\pi}$ and $c\left(4\right)\approx0.3505$,
%thus we obtain
In particular,
\[
p_{c}\left(2\right)\approx0.4570,\: p_{c}\left(3\right)\approx0.4912,\text{ and }p_{c}\left(4\right)\approx0.5186
\]
for $SO\left(2\right)$, $SO\left(3\right)$ and $SO\left(4\right)$
respectively.
\end{theorem}

\paragraph{Remark} The case $d=1$ is trivial since the special orthogonal group contains only one element in this case, but is not trivial for the orthogonal group $O(1)$. In fact, the latter is equivalent to \textsc{Max-Cut}. Our proof below implies that $p_c(1)=0$, that is, if the proportion of correct measurements is strictly greater than $1/2$, then all signs are recovered exactly w.h.p as $n\to \infty$. In fact, the proof below shows that w.h.p the signs are recovered exactly if the bias of the proportion of good measurements is as small as $\mathcal{O}\left(\sqrt{\frac{\log n}{n}}\right)$.

\subsection{Proof of Theorem \ref{thm:p_c}}
%Since the optimization problem (\ref{eq:min_G}) is convex, it suffices to show that w.h.p. (over the measurements), the correct solution is a local minimum of the objective function.

We shall show that the correct solution is a global minimum of the objective function, and analyze the perturbation of the objective function directly. The proof proceeds in two steps. First, we show that without loss of generality we can assume that the correct solution is $R_i = I_d$ and $G_{ij} = I_d$ for all $1\le i,j\le n$. Then, we consider the projection of the perturbation into three different subspaces.\footnote{The method of decomposing the perturbation was used in \cite{irls_t} and \cite{m_estimator} for analyzing the performance of a convex relaxation method for robust principal component analysis.} Using the fact that the diagonal blocks of the perturbation matrix must be $0$, it is possible to show that when the perturbation reduces the objective function for indices in $\mathcal{E}\backslash\mathcal{E}_c$ (that is, the incorrect measurements), it has to increase the objective function for indices in $\mathcal{E}_c$. If the success probability $p$ is large enough, then w.h.p. the amount of increase (to which we later refer as the ``loss'') is greater than the amount of decrease (to which we later refer as the ``gain"), therefore the correct solution must be the solution of the convex optimization problem (\ref{eq:min_G}).

\subsubsection{Fixing the correct solution}

\begin{lemma}\label{lemma:Id}
If the optimal solution of (\ref{eq:min_G}) is w.h.p. $G_{ij} = I_d$ when $R_i = I_d$, then the optimal solution of (\ref{eq:min_G}) is w.h.p.
$G_{ij} = R_i^TR_j$ for arbitrary rotations $\{R_i\}$.
%Without loss of generality, we can assume
%$G_{ij}=I_{d}$.
\end{lemma}

\begin{proof}
We give a bijection that preserves feasibility and objective value between the feasible solutions to (\ref{eq:min_G}) when $R_i = I_d$ and the solution for general $R_i$.

 In fact, given any feasible solution $G$ for general $R_i$, let
\[
\hat{G}=\text{diag}\left(R_{1},\ldots,R_{n}\right)\text{ }G\text{ diag}\left(R_{1},\ldots,R_{n}\right)^{T}.
\]
that is, $\hat{G}_{ij} = R_i G_{ij} R_j^T$. If we rotate $R_{ij}$ similarly to get $\hat{R}_{ij}$ = $R_i R_{ij} R_j^T$, then $\hat{R}_{i,j} = I$. Since $G$ is the solution of (\ref{eq:min_G}), we know $\hat{G}$ must be a solution to the following convex program with the same objective value
\begin{equation}
\min_{\hat{G}} \sum_{(i,j)\in \mathcal{E}} \left \Vert \hat{G}_{ij}-\hat{R}_{ij}\right\Vert \text{ s.t. }\hat{G}_{ii}=I_{d},\text{ and }\hat{G}\succcurlyeq0.
\label{eq:min_canonical}
\end{equation}

However, observe that for edges in $\mathcal{E}_c$, $R_iR_{ij}R_j^T = I_d$; for edges not in $\mathcal{E}_c$, $R_iR_{ij}R_j^T$ is still a uniformly distributed random rotation in $SO(d)$ (due to left and right invariance of the Haar measure). Therefore, (\ref{eq:min_canonical}) is equivalent to (\ref{eq:min_G}) when $R_i = I_d$.

The other direction can be proved identically.
\end{proof}

Using the Lemma \ref{lemma:Id}, we can assume without loss of generality that $R_i = I_d$ for all $1\le i\le n$. Now the correct solution should be $G_{ij} = I_d$. We denote this solution by $G$, and consider a perturbation $G+\Delta$.

\subsubsection{Decomposing the perturbation}

Let $G+\Delta$ be a perturbation of $G$ such that $\Delta_{ii}=0$
and $G+\Delta\succcurlyeq0$. Let $\mathcal{S}$ be the $d$-dimensional linear subspace of $\mathbb{R}^{dn}$ spanned by the vectors ${\bf s}_1,\ldots,{\bf s}_d$, given by
\begin{gather}
{\bf s}_{1}=\left({\bf e}_{1},{\bf e}_{1},\ldots,{\bf e}_{1}\right)^{T},\nonumber\\
{\bf s}_{2}=\left({\bf e}_{2},{\bf e}_{2},\ldots,{\bf e}_{2}\right)^{T},\nonumber\\
\vdots\nonumber\\
{\bf s}_{d}=\left({\bf e}_{d},{\bf e}_{d},\ldots,{\bf e}_{d}\right)^{T},\label{eq:s}
\end{gather}
where ${\bf e}_{m}$ ($m=1,2,\ldots,d$) is the $d$-dimensional row
vector
\[
{\bf e}_{m}(l)=\begin{cases}
1 & \text{if }l=m,\\
0 & \text{otherwise,}
\end{cases}
\]
that is, ${\bf e}_{1}=\left(1,0,\ldots,0\right)$, ${\bf e}_{2}=\left(0,1,0,\ldots,0\right)$,
$\ldots$ and ${\bf e}_{d}=\left(0,\ldots,0,1\right)$.

Intuitively, if a vector $\bf v$ is in the space $\mathcal S$, then the restrictions of $\bf v$ on blocks of size $d$ are all the same (i.e. ${\bf v}_{1,\ldots,d} = {\bf v}_{d+1,\ldots,2d}=\cdots ={\bf v}_{(n-1)d+1,\ldots,nd}$). On the other hand, if a vector $\bf v$ is in $\bar{\mathcal S}$, the orthogonal complement of $\mathcal S$, then it satisfies $\sum_{i=1}^n {\bf v}_{(i-1)d+1,\ldots,id} = 0$.

For two linear subspaces $\mathcal{A}$, $\mathcal{B}$, we use $\mathcal{A}\otimes \mathcal{B}$ to denote the space spanned by vectors $\{{\bf a}\otimes {\bf b}|{\bf a}\in \mathcal{A}\mbox{ and } {\bf b}\in \mathcal{B}\}$. By properties of tensor operation, the dimension of $\mathcal{A}\otimes \mathcal{B}$ equals the dimension of $\mathcal{A}$ times the dimension of $\mathcal{B}$. We also identify the tensor product of two vectors ${\bf a}\otimes {\bf b}$ with the (rank 1) matrix ${\bf a b}^T$.

Now let $P$,
$Q$ and $T$ be $\Delta$'s projections to $\mathcal{S}\otimes\mathcal{S}$,
$\left(\mathcal{S}\otimes\bar{\mathcal{S}}\right)\cup\left(\mathcal{\bar{S}}\otimes\mathcal{S}\right)$,
and $\mathcal{\bar{S}}\otimes\bar{\mathcal{S}}$ respectively, that
is,
\[
\Delta=P+Q+T,\text{ where }P\in\mathcal{S}\otimes\mathcal{S}\text{, }Q\in\left(\mathcal{S}\otimes\bar{\mathcal{S}}\right)\cup\left(\mathcal{\bar{S}}\otimes\mathcal{S}\right)\text{, and }T\in\mathcal{\bar{S}}\otimes\bar{\mathcal{S}}.
\]
Using the definition of $\mathcal{S}$, it is easy to verify that
\begin{equation}
P_{ij} = P_{11} \quad\forall 1\le i,j\le n,\label{eq:P}
\end{equation}
\begin{equation}
Q=Q^{1}+Q^{2},\label{eq:Q}
\end{equation}
where $Q^{1}\in\mathcal{S}\otimes\bar{\mathcal{S}}$, $Q^{2}\in\mathcal{\bar{S}}\otimes\mathcal{S}$,
$Q^{1}=\left(Q^{2}\right)^{T},$
\begin{eqnarray}
Q_{1j}^{1} = Q_{2j}^{1}=\ldots=Q_{nj}^{1}, \forall j\label{eq:Q_1}
\end{eqnarray}
 and
\begin{equation}
Q_{i1}^{2}=Q_{i2}^{2}=\ldots=Q_{in}^{2}, \forall i.\label{eq:Q_2}
\end{equation}
 Moreover, we have
\begin{equation}
\sum_{j}Q_{1j}^{1}=0\text{, }\sum_{i}Q_{i1}^{2}=0\text{ and }\sum_{i,j}T_{ij}=0.\label{eq:Q_T_0}
\end{equation}
For the matrix $T$, the following notation is used to refer to its submatrices
\[
T_{ij}=\left(T_{ij}^{pq}\right)_{p,q=1,\ldots,d},\text{ and }T^{pq}=\left(T_{ij}^{pq}\right)_{i,j=1,\ldots,n},
\]
where $T_{ij}$ are $T$'s $d\times d$ sub-blocks and $T^{pq}$ are $T$'s $n\times n$ sub-matrices whose entries are the $(p,q)$'th elements of the sub-blocks $T_{ij}$.

Recall that $G_{ij}=I_{d}$. Denote the objective function by $F$, that is,
\begin{equation}
F\left(G\right) = \sum _{\left(i,j\right)\in \mathcal{E}} \left\Vert G_{ij}-R_{ij}\right\Vert.
\end{equation}
Then,
\begin{eqnarray}
F\left(G+\Delta\right)-F(G) & = & \sum_{\left(i,j\right)\in \mathcal{E}}\left(\left\Vert I_{d}-R_{ij}+\Delta_{ij}\right\Vert -\left\Vert I_{d}-R_{ij}\right\Vert \right)\nonumber \\
 & = & \sum_{\left(i,j\right)\in \mathcal{E}\backslash \mathcal{E}_{c}}\left(\left\Vert I_{d}-R_{ij}+\Delta_{ij}\right\Vert -\left\Vert I_{d}-R_{ij}\right\Vert \right)+\sum_{\left(i,j\right)\in \mathcal{E}_{c}}\left\Vert \Delta_{ij}\right\Vert \nonumber \\
 & =: & f_{1}+f_{2}.\label{eq:f1_plus_f2}
\end{eqnarray}
Intuitively, if the objective value of $G+\Delta$ is smaller than the objective value of $G$, then $\Delta_{ij}$ should be close to $0$ for the correct measurements $(i,j)\in \mathcal{E}_c$, and is large on $(i,j)$ such that $(i,j)\not\in \mathcal{E}_c$.

We shall later show that the ``gain''  $f_1$ from $(i,j)\not\in \mathcal{E}_c$ can be upperbounded by the trace of $T$ and the
off-diagonal entries of $Q^1_{ij}$. Then we shall show that when the trace of $T$ is large, the diagonal entries of $\Delta_{ij}$ for $(i,j)\in \mathcal{E}_c$ are large, therefore the ``gain'' is smaller than the ``loss'' generated by these diagonal entries. On the other hand, when the off-diagonal entries of $Q^1_{ij}$ are large, then the off-diagonal entries of $\Delta_{ij}$ for $(i,j)\in \mathcal{E}_c$ are large, once again the ``gain'' will be smaller than the ``loss'' generated by the off-diagonal entries.

\subsubsection{Observations on $P$, $Q$ and $T$}

To bound the ``gain'' and ``loss'', we need the following set of observations.

\begin{lemma}
\label{lem:sum_trace_p}
$
nP_{11}=-\sum_{i=1}^n T_{ii}.
$
\end{lemma}
\begin{proof}
Using  (\ref{eq:P})-(\ref{eq:Q_T_0}) and the fact that $\Delta_{ii}=0$, we have
\begin{eqnarray*}
0&=&\sum_{i}\Delta_{ii}=
\sum_{i}\left(P_{ii}+Q_{ii}^{1}+Q_{ii}^{2}+T_{ii}\right)\\
&=&nP_{11}+\sum_{i}Q_{1i}^{1}+\sum_{i}Q_{i1}^{2}+\sum_{i}T_{ii}
=nP_{11}+\sum_{i}T_{ii}.
\end{eqnarray*}
\end{proof}

\begin{lemma}
\label{lem:PQT}
$\left(P_{ij}+Q_{ij}\right)+\left(P_{ji}+Q_{ji}\right)= -\left(T_{ii}+T_{jj}\right).$
\end{lemma}
\begin{proof}
We use the symmetry of the matrices $P$, $Q$:
\begin{align*}
 & \left(P_{ij}+Q_{ij}\right)+\left(P_{ji}+Q_{ji}\right)\\
= & P_{ii}+Q_{ij}^{1}+Q_{ij}^{2}+P_{jj}+Q_{ji}^{1}+Q_{ji}^{2}\\
= & P_{ii}+Q_{jj}^{1}+Q_{ii}^{2}+P_{jj}+Q_{ii}^{1}+Q_{jj}^{2}\\
= & \left(P_{ii}+Q_{ii}^{1}+Q_{ii}^{2}\right)+\left(P_{jj}+Q_{jj}^{1}+Q_{jj}^{2}\right)\\
= & \left(\Delta_{ii}-T_{ii}\right)+\left(\Delta_{jj}-T_{jj}\right)\\
= & -\left(T_{ii}+T_{jj}\right),
\end{align*}
where the second equality follows (\ref{eq:Q_1})-(\ref{eq:Q_2}).
\end{proof}

\begin{lemma}
\label{lem:T_PSD}
$T\succcurlyeq0$.
\end{lemma}
\begin{proof}
Since $G+\Delta=G+P+Q+T\succcurlyeq0$, for any vector ${\bf v} \in \mathcal{\bar{S}}$, ${\bf v}^T(G+\Delta){\bf v} \ge 0$. However, ${\bf v}^TG{\bf v} = {\bf v}^T P {\bf v} = {\bf v}^TQ{\bf v} = 0$ according to the definition of $G$, $P$ and $Q$. Therefore, for all ${\bf v}\in \mathcal{\bar{S}}$ we have ${\bf v}^T T {\bf v} \ge 0$. Also, $T {\bf w} = 0$ for ${\bf w}  \in \mathcal{S}$. Therefore, if ${\bf v}\in \mathcal{\bar{S}}$ and ${\bf w} \in \mathcal{S}$ then $({\bf v}+{\bf w} )^T T ({\bf v}+{\bf w} ) = {\bf v}^T T {\bf v} \geq 0$. Hence, $T$ is positive semidefinite.
\end{proof}

\begin{lemma}
\label{lem:T_trace_bound}
Let $T_{ij}^\text{d}$ be a diagonal matrix whose diagonal entries are those of $T_{ij}$, then
$\left\Vert \sum_{i}T_{ii}^\text{d}\right\Vert \geq\text{Tr}\left(T\right)/\sqrt{d}$.
\end{lemma}

\begin{proof}
This is just a straight forward application of the Cauchy-Schwarz inequality.
Denote $X=\sum_{i}T_{ii}^\text{d}$. Clearly, $X\succcurlyeq0$ , since
it is a sum of positive semidefinite matrices. Let $x_{1},\ldots,x_{d}$
be the diagonal entries of $X$. Then, from  Cauchy-Schwarz inequality,
we have
\[
\text{Tr}\left(X\right)=\sum_{j}x_{j}\leq\left(\sum_{j}x_{j}^{2}\right)^{\frac{1}{2}}\cdot\sqrt{d}=\sqrt{d}\left\Vert X\right\Vert ,
\]
that is, $\left\Vert \sum_{i}T_{ii}^\text{d}\right\Vert =\left\Vert X\right\Vert \geq\text{Tr}\left(X\right)/\sqrt{d}=\text{Tr}\left(T\right)/\sqrt{d}$.
\end{proof}

\begin{lemma}
\label{lem:lambda}
Let $A$ be a $n\times n$ adjacency matrix such that $A_{ij}=1$
if and only if $\left(i,j\right)\in \mathcal{E}_{c}$; otherwise, $A_{ij}=0$. Let
$B=A-\frac{1}{n^{2}}\left({\bf 1}^TA{\bf 1}\right){\bf 1}{\bf 1}^T$, where ${\bf 1}$ is the all-ones (column) vector.
Let $\lambda = \|B\|_2$, where $\|\cdot\|_2$ denotes the spectral norm of a matrix. Then,
\[
\lambda=\mathcal{O}_P\left(\sqrt{n}\right).
\]

Here the notation $f = \mathcal{O}_P(g)$ means $f$ is upper bounded by $c_p \cdot g$ with high probability, where $c_p$ is a constant that may depend on $p$.
\end{lemma}
\begin{proof}
Let $B_{1}=A-p{\bf 1}{\bf 1}^T$. Observe that $B_1$ is a random matrix
where each off-diagonal entry is either $\left(1-p\right)$ with probability
$p$ or $-p$ with probability $\left(1-p\right)$. Therefore, by Wigner's semi-circle law and the concentration of the eigenvalues of random symmetric matrices with i.i.d entries of absolute value at most 1 \cite{Alon2002}, the
largest eigenvalue (in absolute value) of $B_{1}$ is $\left\Vert B_{1}\right\Vert _{2}=\mathcal{O}_P\left(\sqrt{n}\right)$.
Then we have
\begin{eqnarray*}
\left\Vert B\right\Vert _{2}
 & = & \left\Vert B_{1}+\left(p-\frac{1}{n^{2}}\left({\bf 1}^TA{\bf 1}\right)\right){\bf 1}{\bf 1}^T\right\Vert _{2}\\
 & \leq & \left\Vert B_{1}\right\Vert _{2}+\left|p-\frac{1}{n^{2}}\left({\bf 1}^TA{\bf 1}\right)\right|\left\Vert {\bf 1}{\bf 1}^T\right\Vert _{2}\\
 & \leq & \left\Vert B_{1}\right\Vert _{2}+n\mathcal{O}_P\left(n^{-1}\sqrt{\log n}\right)\\
 & = & \mathcal{O}_P\left(\sqrt{n}\right),
\end{eqnarray*}
where the second inequality uses the Chernoff bound.
\end{proof}

\paragraph{Remark} The matrix $A$ is the adjacency matrix of the Erd\H{o}s\textendash{}R\'{e}nyi
(ER) random graph model $\mathcal{G}(n,p)$ where the expectation of every node's
degree is $(n-1)p$. Denote by $\lambda_i$ the $i$'th largest eigenvalue of a matrix. It is intuitive to see that $\lambda_{1}\left(A\right)\approx\left(n-1\right)p$
and the first eigenvector is approximately the all-ones vector ${\bf 1}$,
and $\lambda_{2}\left(A\right)\approx\lambda_{1}\left(B\right)=\mathcal{O}_P\left(\sqrt{n}\right)$.

Using these observations, we can bound the sum of norm of all $T_{ij}$'s by the trace of $T$.

\begin{lemma}
\label{lem:T_bounds}
We have the following three inequalities for the matrix $T$:
\begin{enumerate}
\item $\left|\sum_{\left(i,j\right)\in \mathcal{E}_{c}}T_{ij}^{pq}\right|\leq\lambda\text{Tr}\left(T^{pp}\right)^{\frac{1}{2}}\text{Tr}\left(T^{qq}\right)^{\frac{1}{2}}$
for $p,q=1,2,\ldots,d$,
\item $\left|\left\langle I_{d},\sum_{\left(i,j\right)\in \mathcal{E}_{c}}T_{ij}\right\rangle \right|\leq\lambda\text{Tr}\left(T\right),$
\item $\left\Vert \sum_{\left(i,j\right)\in \mathcal{E}_{c}}T_{ij}\right\Vert \leq\lambda\text{Tr}\left(T\right).$
\end{enumerate}
\end{lemma}
Here the notation for the inner product of two matrices $\left\langle X, Y \right\rangle$ means $\text{Tr}(X^T Y)$.

\begin{proof}
1. Since $T\succcurlyeq0$, we can assume $T_{ij}^{pq}=\left\langle {\bf u}_{i}^{p},{\bf u}_{j}^{q}\right\rangle $,
where ${\bf u}_{i}^{p}$($i=1,\ldots,n$, $p=1,\ldots,d$) is an $n$
dimensional row vector, then $\text{Tr}\left(T^{pp}\right)=\sum_{i=1}^{n}\left\langle {\bf u}_{i}^{p},{\bf u}_{i}^{p}\right\rangle $,
and

\[
\sum_{\left(i,j\right)\in \mathcal{E}_{c}}T_{ij}^{pq}=\left({\bf u}_{1}^{p},\ldots,{\bf u}_{n}^{p}\right)\left(A\otimes I_{n}\right)\left({\bf u}_{1}^{q},\ldots,{\bf u}_{n}^{q}\right)^{T}.
\]
We claim that $\sum_{i=1}^{n}{\bf u}_{i}^{p}={\bf 0}$. In fact, since
$T\in\mathcal{\bar{S}}\otimes\bar{\mathcal{S}}$, for any $p,q=1,\ldots,d$,
we have $\sum_{i,j=1}^{n}T_{ij}^{pq}={\bf s}_{p}^{T}T{\bf s}_{q}=0$
. Therefore, we obtain
\[
0  = \sum_{i,j=1}^{n}T_{ij}^{pp}  =  \sum_{i,j=1}^{n}\left\langle {\bf u}_{i}^{p},{\bf u}_{j}^{p}\right\rangle  =  \left\langle \sum_{i=1}^{n}{\bf u}_{i}^{p},\sum_{j=1}^{n}{\bf u}_{j}^{p}\right\rangle
  =  \left\Vert \sum_{i=1}^{n}{\bf u}_{i}^{p}\right\Vert ^{2}
\]
and thus we have
\begin{equation}
\sum_{i=1}^{n}{\bf u}_{i}^{p}={\bf 0}\text{ for }p=1,\ldots,d.\label{eq:sum_u}
\end{equation}
Let ${\bf v}_{m}^{p}$ ($m=1,\ldots,n$, $p=1,\ldots,d$) be a $n$
dimensional row vector such that
\[
{\bf v}_{m}^{p}=\left({\bf u}_{1}^{p}\left(m\right),{\bf u}_{2}^{p}\left(m\right),\ldots,{\bf u}_{n}^{p}\left(m\right)\right),
\]
then we have $\sum_{m=1}^{n}\left\Vert {\bf v}_{m}^{p}\right\Vert ^{2}=\text{Tr}\left(T^{pp}\right)$,
and $\left\langle {\bf v}_{m}^{p},{\bf 1}\right\rangle =0$ due to
(\ref{eq:sum_u}). Therefore
\begin{eqnarray*}
\left|\sum_{\left(i,j\right)\in \mathcal{E}_{c}}T_{ij}^{pq}\right| & = & \left|\left({\bf u}_{1}^{p},\ldots,{\bf u}_{n}^{p}\right)\left(A\otimes I_{n}\right)\left({\bf u}_{1}^{q},\ldots,{\bf u}_{n}^{q}\right)^{T}\right|\\
 & = & \left|\sum_{m=1}^{n}{\bf v}_{m}^{p}A\left({\bf v}_{m}^{q}\right)^{T}\right|\\
 & \leq & \sum_{m}\lambda\left\Vert {\bf v}_{m}^{p}\right\Vert \left\Vert {\bf v}_{m}^{q}\right\Vert \\
 & \leq & \lambda\left(\sum_{m}\left\Vert {\bf v}_{m}^{p}\right\Vert ^{2}\right)^{\frac{1}{2}}\left(\sum_{m}\left\Vert {\bf v}_{m}^{q}\right\Vert ^{2}\right)^{\frac{1}{2}},\\
 & = & \lambda\text{Tr}\left(T^{pp}\right)^{\frac{1}{2}}\text{Tr}\left(T^{qq}\right)^{\frac{1}{2}}
\end{eqnarray*}
where the first inequality uses Lemma \ref{lem:lambda} and the fact that ${\bf v}_{m}^{p}$
($m=1,\ldots,n$, $p=1,\ldots,d$) is orthogonal to the all-ones vector
${\bf 1}$, and the second inequality uses Cauchy-Schwarz inequality.

2. $\left|\left\langle I_{d},\sum_{\left(i,j\right)\in \mathcal{E}_{c}}T_{ij}\right\rangle \right|\leq\lambda\text{Tr}\left(T\right)$
is clear from 1. when $p=q$.

3. From 1. we have
\begin{eqnarray*}
\left\Vert \sum_{\left(i,j\right)\in \mathcal{E}_{c}}T_{ij}\right\Vert  & = & \sqrt{\sum_{p,q}\left(\sum_{\left(i,j\right)\in \mathcal{E}_{c}}T_{ij}^{pq}\right)^{2}}\\
 & \leq & \sqrt{\sum_{p,q}\left(\lambda\text{Tr}\left(T^{pp}\right)^{\frac{1}{2}}\text{Tr}\left(T^{qq}\right)^{\frac{1}{2}}\right)^{2}}\\
 & = & \lambda\sqrt{\sum_{p,q}\text{Tr}\left(T^{pp}\right)\text{Tr}\left(T^{qq}\right)}\\
 & = & \lambda\sum_{p}\text{Tr}\left(T^{pp}\right)\\
 & = & \lambda\text{Tr}\left(T\right).
\end{eqnarray*}
\end{proof}

\subsubsection{Bounding the ``gain'' from incorrect measurements}
\label{sec:gain}
To make the intuition that ``gain'' from incorrect measurements is always upper bounded by the ``loss'' from correct measurements formal, we shall first focus on $(i,j)\not\in\mathcal{E}_c$, and bound the ``gain'' by trace of $T$ and norms of $Q^1_{ij}$.

Recall that $f_1$ in (\ref{eq:f1_plus_f2}) is the ``gain'', we shall bound it using the following lemma.
\begin{lemma}
\label{lem:gain_global_lemma}
For any pair of non-zero matrices $M_{1}$ ,$M_{2}\in\mathbb{R}^{m\times m}$,
we have
\begin{equation}
\left\Vert M_{1}+M_{2}\right\Vert -\left\Vert M_{1}\right\Vert \geq\left\langle M_{1}/\left\Vert M_{1}\right\Vert ,M_{2}\right\rangle .\label{eq:M12}
\end{equation}
\end{lemma}
\begin{proof}
Using Cauchy-Schwarz inequality, we obtain
\[
\left\Vert M_{1}\right\Vert ^{2}+\left\langle M_{1},M_{2}\right\rangle =\left\langle M_{1},M_{1}+M_{2}\right\rangle \leq\left\Vert M_{1}\right\Vert \left\Vert M_{1}+M_{2}\right\Vert ,
\]
which is equivalent to (\ref{eq:M12}).
\end{proof}

We apply Lemma \ref{lem:gain_global_lemma} to the ``gain'' $f_1$ and obtain
\begin{eqnarray}
f_{1} & \geq & \sum_{\left(i,j\right)\in \mathcal{E}\backslash \mathcal{E}_{c}}\left\langle \frac{I_{d}-R_{ij}}{\left\Vert I_{d}-R_{ij}\right\Vert },\Delta_{ij}\right\rangle \nonumber \\
 & = & \sum_{\left(i,j\right)\in \mathcal{E}\backslash \mathcal{E}_{c}}\left\langle \frac{I_{d}-R_{ij}}{\left\Vert I_{d}-R_{ij}\right\Vert },P_{ij}+Q_{ij}+T_{ij}\right\rangle .\label{eq:f1}
\end{eqnarray}
First we shall bound the ``gain'' from the matrix $P$. Since blocks of $P$ are the same, the average
\[
\frac{1}{| \mathcal{E}\backslash \mathcal{E}_{c}|}\sum_{\left(i,j\right)\in \mathcal{E}\backslash \mathcal{E}_{c}}\left\langle \frac{I_{d}-R_{ij}}{\left\Vert I_{d}-R_{ij}\right\Vert },P_{ij}\right\rangle
\]
 should be concentrated around the expectation of $ \left\langle \frac{I_{d}-R_{ij}}{\left\Vert I_{d}-R_{ij}\right\Vert },P_{ij}\right\rangle $. The expectation is analyzed in Appendix \ref{sec:cd}. By (\ref{eq:P})-(\ref{eq:Q_2}) and the law of large numbers, we
obtain that
\begin{eqnarray}
&&\sum_{\left(i,j\right)\in \mathcal{E}\backslash \mathcal{E}_{c}}\left\langle \frac{I_{d}-R_{ij}}{\left\Vert I_{d}-R_{ij}\right\Vert },P_{ij}\right\rangle  \nonumber\\
& = & \left(c\left(d\right)\left(1-p\right)n\left(n-1\right)+\mathcal{O}_P\left(n\sqrt{\log n}\right)\right)\text{Tr}\left(P_{11}\right)\nonumber \\
 & = &- \left(c\left(d\right)\left(1-p\right)\left(n-1\right)+\mathcal{O}_P\left(\sqrt{\log n}\right)\right)\text{Tr}\left(T\right),\label{eq:f1_P}
\end{eqnarray}
where the last equality uses Lemma \ref{lem:sum_trace_p}, $c\left(d\right)$ is defined in (\ref{eq:def_cd})
and the rotation matrix $R$ is uniformly sampled from the group $SO\left(d\right)$
(see Appendix \ref{sec:cd} for detailed discussion).

For matrix $Q$ we use similar concentration bounds

\begin{align}
 & \sum_{\left(i,j\right)\in \mathcal{E}\backslash \mathcal{E}_{c}}\left\langle \frac{I_{d}-R_{ij}}{\left\Vert I_{d}-R_{ij}\right\Vert },Q_{ij}\right\rangle \nonumber \\
= & \sum_{\left(i,j\right)\in \mathcal{E}\backslash \mathcal{E}_{c}}\left\langle \frac{I_{d}-R_{ij}}{\left\Vert I_{d}-R_{ij}\right\Vert },Q_{ij}^{1}+Q_{ij}^{2}\right\rangle \nonumber \\
= & \sum_{\left(i,j\right)\in \mathcal{E}\backslash \mathcal{E}_{c}}\left\langle \frac{I_{d}-R_{ij}}{\left\Vert I_{d}-R_{ij}\right\Vert },Q_{jj}^{1}+Q_{ii}^{2}\right\rangle \nonumber \\
\geq & \sum_{j}\left(n-p\left(n-1\right)\right)\left\langle c\left(d\right)I_{d},Q_{jj}^{1}\right\rangle +\sum_{i}\left(n-p\left(n-1\right)\right)\left\langle c\left(d\right)I_{d},Q_{ii}^{2}\right\rangle \nonumber\\
&-\mathcal{O}_P\left(\sqrt{n\log n}\right)\sum_i\left\Vert Q_{jj}^1\right\Vert-\mathcal{O}_P\left(\sqrt{n\log n}\right)\sum_i\left\Vert Q_{ii}^2\right\Vert\nonumber \\
= & \left(n-p\left(n-1\right)\right)\left\langle cI_{d},\sum_{j}Q_{jj}^{1}\right\rangle +\left(n-p\left(n-1\right)\right)\left\langle cI_{d},\sum_{i}Q_{ii}^{2}\right\rangle\nonumber\\ &-\mathcal{O}_P\left(\sqrt{n\log n}\right)\sum_i\left\Vert Q_{ii}^1\right\Vert\nonumber \\
= & -\mathcal{O}_P\left(\sqrt{n\log n}\right)\sum_i\left\Vert Q_{ii}^1\right\Vert,\label{eq:f1_Q}
\end{align}
where the third equality uses the fact that $Q^1=\left(Q^2\right)^T$, and the last equality follows (\ref{eq:Q_T_0}).

Finally we shall bound the ``gain'' from matrix $T$, which is
% To find a lower
%bound for $f_{1}$ in (\ref{eq:f1}), we need lower bounding
\begin{equation}
\sum_{\left(i,j\right)\in \mathcal{E}\backslash \mathcal{E}_{c}}\left\langle \frac{I_{d}-R_{ij}}{\left\Vert I_{d}-R_{ij}\right\Vert },T_{ij}\right\rangle .\label{eq:f1_T}
\end{equation}
Before continuing, we need the following results in Lemma \ref{lem:D2parts} for a matrix $D$, where
$D$ is defined as
\begin{equation}
D_{ij}=\begin{cases}
\frac{I_{d}-R_{ij}}{\left\Vert I_{d}-R_{ij}\right\Vert }-c\left(d\right)I_{d} & \text{for }\left(i,j\right)\in \mathcal{E}\backslash \mathcal{E}_{c},\\
0 & \text{otherwise.}
\end{cases}\label{eq:D}
\end{equation}
The proof of Lemma \ref{lem:D2parts} is in Appendix \ref{sec:rm}.
\begin{lemma}
\label{lem:D2parts}
The limiting spectral density of $\frac{1}{\sqrt{n-1}}D$ is Wigner's
semi-circle. In addition, the matrix $D$ can be decomposed as $
D=D_{+}+D_{-},
$
where $D_{+}\succcurlyeq0$ and $D_{-}\prec0$, and we have
\begin{equation}
\left\Vert D_{+}\right\Vert ^{2}\approx\left\Vert D_{-}\right\Vert ^{2}\approx\frac{1}{2}\left(1-p\right)n\left(n-1\right)\left(1-c\left(d\right)^{2}d\right).\label{eq:D_decomposition_norm}
\end{equation}
\end{lemma}

We return to lower bounding (\ref{eq:f1_T}). Since $\sum_{ij}T_{ij}=0$,
we have
\begin{align}
 & \sum_{\left(i,j\right)\in \mathcal{E}\backslash \mathcal{E}_{c}}\left\langle \frac{I_{d}-R_{ij}}{\left\Vert I_{d}-R_{ij}\right\Vert },T_{ij}\right\rangle \nonumber \\
= & \sum_{\left(i,j\right)\in \mathcal{E}\backslash \mathcal{E}_{c}}\left\langle c\left(d\right)I_{d},T_{ij}\right\rangle +\sum_{\left(i,j\right)\in \mathcal{E}\backslash \mathcal{E}_{c}}\left\langle \frac{I_{d}-R_{ij}}{\left\Vert I_{d}-R_{ij}\right\Vert }-c\left(d\right)I_{d},T_{ij}\right\rangle \nonumber \\
= & -c\left(d\right)\sum_{\left(i,j\right)\in \mathcal{E}_{c}}\left\langle I_{d},T_{ij}\right\rangle -c\left(d\right)\sum_{i}\left\langle I_{d},T_{ii}\right\rangle +\sum_{i,j}\left\langle D_{ij},T_{ij}\right\rangle \nonumber \\
= & -c\left(d\right)\sum_{\left(i,j\right)\in \mathcal{E}_{c}}\left\langle I_{d},T_{ij}\right\rangle -c\left(d\right)\text{Tr}\left(T\right)+\text{Tr}\left(DT\right)\nonumber \\
= & -c\left(d\right)\sum_{\left(i,j\right)\in \mathcal{E}_{c}}\left\langle I_{d},T_{ij}\right\rangle -c\left(d\right)\text{Tr}\left(T\right)+\text{Tr}\left(D_{+}T\right)+\text{Tr}\left(D_{-}T\right)\nonumber \\
\geq & -c\left(d\right)\sum_{\left(i,j\right)\in \mathcal{E}_{c}}\left\langle I_{d},T_{ij}\right\rangle -c\left(d\right)\text{Tr}\left(T\right)+\text{Tr}\left(D_{-}T\right),\label{eq:f1_T_1}
\end{align}
where the last inequality follows from the fact that $\text{Tr}\left(D_{+}T\right)\geq0$
since both $T$ and $D_{+}$ are positive semidefinite matrices. Recall
that for any positive semidefinite matrix $X$ the following inequality
holds: $\left\Vert X\right\Vert \leq\text{Tr}\left(X\right)$. Since
$T\succcurlyeq0$, using (\ref{eq:D_decomposition_norm}) we obtain
\begin{eqnarray}
\hspace{-1cm}\left|\text{Tr}\left(D_{-}T\right)\right|\leq\left\Vert D_{-}\right\Vert \left\Vert T\right\Vert \leq\left\Vert D_{-}\right\Vert \text{Tr}\left(T\right)\approx\frac{1}{\sqrt{2}}\sqrt{\left(1-p\right)}\sqrt{1-c\left(d\right)^{2}d}n\text{Tr}\left(T\right).\label{eq:f1_T_2}
\end{eqnarray}
Also, Lemma \ref{lem:T_bounds} reads
\begin{equation}
\left|\left\langle I_{d},\sum_{\left(i,j\right)\in \mathcal{E}_{c}}T_{ij}\right\rangle \right|\leq\lambda\text{Tr}\left(T\right).\label{eq:f1_T_3}
\end{equation}
Combining (\ref{eq:f1_T_1}), (\ref{eq:f1_T_2}), and (\ref{eq:f1_T_3})
together gives
\begin{eqnarray}
&&\sum_{\left(i,j\right)\in \mathcal{E}\backslash \mathcal{E}_{c}}\left\langle \frac{I_{d}-R_{ij}}{\left\Vert I_{d}-R_{ij}\right\Vert },T_{ij}\right\rangle\nonumber \\ &\geq&-\left(c\left(d\right)\lambda+c\left(d\right)+\frac{1}{\sqrt{2}}\sqrt{\left(1-p\right)}\sqrt{1-c\left(d\right)^{2}d}n\right)\text{Tr}\left(T\right).
\end{eqnarray}
Since from Lemma \ref{lem:lambda} we have $\lambda=\mathcal{O}_P\left(\sqrt{n}\right)$,
the bound is given by
\begin{eqnarray}
&&\sum_{\left(i,j\right)\in \mathcal{E}\backslash \mathcal{E}_{c}}\left\langle \frac{I_{d}-R_{ij}}{\left\Vert I_{d}-R_{ij}\right\Vert },T_{ij}\right\rangle \nonumber\\
&\geq&-\left(\frac{1}{\sqrt{2}}\sqrt{\left(1-p\right)}\sqrt{1-c\left(d\right)^{2}d}n+\mathcal{O}_P\left(\sqrt{n}\right)\right)\text{Tr}\left(T\right).\label{eq:f1_T_bound}
\end{eqnarray}

Combining (\ref{eq:f1}), (\ref{eq:f1_P}), (\ref{eq:f1_Q}) and (\ref{eq:f1_T_bound})
together we obtain a lower bound for the gain from $\Delta$ as following:
\begin{eqnarray}
f_{1}&\geq&-\left(\frac{1}{\sqrt{2}}\sqrt{\left(1-p\right)}\sqrt{1-c\left(d\right)^{2}d}n+c\left(d\right)\left(1-p\right)n+\mathcal{O}_P\left(\sqrt{n}\right)\right)\text{Tr}\left(T\right)\nonumber\\
&&-\mathcal{O}_P\left(\sqrt{n\log n}\right)\sum_i\left\Vert Q_{ii}^1\right\Vert.\label{eq:f1_bound}
\end{eqnarray}

\subsubsection{Bounding the ``loss'' from correct measurements}

Now we consider the part $f_2=\sum_{\left(i,j\right)\in \mathcal{E}_{c}}\left\Vert \Delta_{ij}\right\Vert $, which is the loss from the correct entries. We use the notations $\Delta_{ij}^\text{d}$ and $\Delta_{ij}^\text{off}$ to represent the restrictions of the sub-matrices $\Delta_{ij}$ on the diagonal entries and off-diagonal entries respectively. We will analyze $\sum_{\left(i,j\right)\in \mathcal{E}_{c}}\left\Vert \Delta_{ij}^\text{d}\right\Vert $ and $\sum_{\left(i,j\right)\in \mathcal{E}_{c}}\left\Vert \Delta_{ij}^\text{off}\right\Vert $ separately.

For the diagonal entries we have

\begin{eqnarray}
\sum_{\left(i,j\right)\in \mathcal{E}_{c}}\left\Vert \Delta_{ij}^{\text{d}}\right\Vert  & = & \sum_{\left(i,j\right)\in \mathcal{E}_{c}}\left\Vert P_{ij}^{\text{d}}+Q_{ij}^{\text{d}}+T_{ij}^{\text{d}}\right\Vert \nonumber\\
&\geq& \left\Vert \sum_{\left(i,j\right)\in \mathcal{E}_{c}}( P_{ij}^{\text{d}}+Q_{ij}^{\text{d}}+T_{ij}^{\text{d}})\right\Vert \nonumber\\
 & \geq & \left\Vert \sum_{\left(i,j\right)\in \mathcal{E}_{c}}P_{ij}^{\text{d}}+Q_{ij}^{\text{d}}\right\Vert -\left\Vert \sum_{\left(i,j\right)\in \mathcal{E}_{c}}T_{ij}^{\text{d}}\right\Vert \nonumber\\
 & = & \left\Vert \sum_{\left(i,j\right)\in \mathcal{E}_{c}}T_{ii}^{\text{d}}+T_{jj}^{\text{d}}\right\Vert -\left\Vert \sum_{\left(i,j\right)\in \mathcal{E}_{c}}T_{ij}^{\text{d}}\right\Vert \nonumber\\
 & = & pn\left\Vert \sum_{i}T_{ii}^{\text{d}}\right\Vert -\mathcal{O}_{P}\left(\sqrt{n\log n}\right)\text{Tr}\left(T\right)-\left\Vert \sum_{\left(i,j\right)\in \mathcal{E}_{c}}T_{ij}^{\text{d}}\right\Vert\nonumber \\
 & \geq & \frac{2pn}{\sqrt{d}}\text{Tr}\left(T\right)-\mathcal{O}_{P}\left(\sqrt{n\log n}\right)\text{Tr}\left(T\right)-\lambda\text{Tr}\left(T\right)\nonumber\\
 & = & \left(\frac{pn}{\sqrt{d}}-\mathcal{O}_{P}\left(\sqrt{n\log n}\right)\right)\text{Tr}\left(T\right),\label{eq:f2_d_bound}
\end{eqnarray}
where the second equality uses Lemma \ref{lem:PQT},
the third equality uses the law of large numbers and Chernoff  bound, the last inequality
follows Lemma \ref{lem:T_trace_bound} and Lemma \ref{lem:T_bounds}, and the last equality uses the fact that $\lambda=\mathcal{O}_P(\sqrt{n})$ from Lemma \ref{lem:lambda}.

For the off-diagonal entries, we have the following lemma, whose proof is deferred to
Appendix \ref{sec:contribution_Q}.
\begin{lemma}
\label{lem:f2_off_bound}
\begin{equation}
\sum_{\left(i,j\right)\in \mathcal{E}_{c}}\left\Vert \Delta_{ij}^{\text{off}}\right\Vert\geq\frac{p}{d^2}\left(\mathcal{O}_P\left(n\right)\sum_{i=1}^{n}\|Q_{ii}^{1}\|-\mathcal{O}_P\left(n\right)\text{Tr}\left(T\right)\right).\label{eq:f2_off_bound}
\end{equation}
\end{lemma}

Finally, we can merge the loss in two parts by a simple inequality:

\begin{lemma}
\label{lem:d_and_off}
If $\sum_{i}\left|a_{i}\right|\geq a$ and $\sum_{i}\left|b_{i}\right|\geq b$,
then for any $\alpha$, $\beta>0$ such that $\alpha^{2}+\beta^{2}=1$,
we have
\[
\sum_{i}\sqrt{a_{i}^{2}+b_{i}^{2}}\geq\alpha a+\beta b.
\]
\end{lemma}

Applying Lemma \ref{lem:d_and_off} to (\ref{eq:f2_d_bound}) and (\ref{eq:f2_off_bound}), and setting $\alpha=1-\frac{1}{\sqrt{n}}$ and $\beta=\sqrt{\frac{2}{\sqrt{n}}-\frac{1}{n}}$,
we obtain a lower bound for the loss from $\Delta$ when $p>\frac{1}{2}$:
\begin{eqnarray}
f_{2} & = & \sum_{\left(i,j\right)\in \mathcal{E}_{c}}\left\Vert \Delta_{ij}\right\Vert\nonumber\\
 & \geq & \left(1-\frac{1}{\sqrt{n}}\right)\sum_{\left(i,j\right)\in \mathcal{E}_{c}}\left\Vert \Delta_{ij}^{\text{d}}\right\Vert +\sqrt{\frac{2}{\sqrt{n}}-\frac{1}{n}}\sum_{\left(i,j\right)\in \mathcal{E}_{c}}\left\Vert \Delta_{ij}^{\text{off}}\right\Vert\nonumber \\
 & = & \left(\frac{2pn}{\sqrt{d}}-\mathcal{O}_{P}\left(n^{\frac{3}{4}}\right)\right)\text{Tr}\left(T\right)+\mathcal{O}_{P}\left(n^{\frac{3}{4}}\right)\sum_{i=1}^{n}\|Q_{ii}^{1}\|.\label{eq:f2_bound}
\end{eqnarray}

\subsubsection{Finishing the proof}
Now that we have bounded the ``gain'' and the ``loss'', we just need a condition on $p$ such that the ``loss'' is greater than the ``gain'' (w.h.p.). Combining (\ref{eq:f1_plus_f2}), (\ref{eq:f1_bound})
and (\ref{eq:f2_bound}) together, when $p>1/2$ we obtain
\begin{align}
\hspace{-1.5cm}
&F\left(G+\Delta\right)-F(G)  =  f_{1}+f_{2}\nonumber \\
  \geq & \left(\left(2p/\sqrt{d}-\frac{1}{\sqrt{2}}\sqrt{\left(1-p\right)}\sqrt{1-c\left(d\right)^{2}d}-c\left(d\right)\left(1-p\right)\right)n-\mathcal{O}_P\left(n^{\frac{3}{4}}\right)\right)\text{Tr}\left(T\right)
\label{eq:F_bound}
\end{align}
that is, when the number of rotations $n$ is large enough, we just need
\[
p\gtrsim1-\left(\frac{-c_{1}\left(d\right)+\sqrt{c_{1}\left(d\right)^2+8\left(c\left(d\right)+2/\sqrt{d}\right)/\sqrt{d}}}{2\left(c\left(d\right)+2/\sqrt{d}\right)}\right)^{2}:=p_{c}\left(d\right),
\]
where
$c_{1}(d)$ is defined as (\ref{eq:c_1})

\section{Stability of LUD}
\label{sec:stability}
In this section we will analyze the behavior of LUD when the measurements $R_{ij}$ on the ``good'' edge set $\mathcal{E}_c$ are no longer the true rotation ratios $R_i^TR_j$, instead, $R_{ij}$ are small perturbations of  $R_i^TR_j$. Similar to the noise model (\ref{eq:noise}),  we assume in our model that the measurements $R_{ij}$ are given by
\begin{equation}
R_{ij}=\delta_{ij}\bar{R}_{ij}+\left(1-\delta_{ij}\right)\tilde{R}_{ij},\label{eq:noise_purb}
\end{equation}
where the rotation $\bar{R}_{ij}$ is sampled from a probability distribution (e.g. the von Mises-Fisher distribution \cite{chiuso2008}; c.f. Section \ref{sec:exp2}) such that
\begin{equation}
\mathbb{E}\left(\left\Vert \bar{R}_{ij}-R_{i}^{T}R_{j}\right\Vert \right)=\epsilon,\;\text{ and }\;\text{Var}\left(\left\Vert \bar{R}_{ij}-R_{i}^{T}R_{j}\right\Vert \right)=\mathcal{O}\left(\epsilon^2\right).\label{eq:condition_purb}
\end{equation}
Note that the stability result is not limited to the random noise, and the analysis can also be applied to bounded deterministic perturbations.

We can generalize the analysis for exact recovery to this new noise model (\ref{eq:noise_purb}) with small perturbations on the ``good'' edge set and prove the following theorem.
\begin{theorem}
\label{thm:stability_weak}
\emph{(Weak stability)}
Assume that all pairwise ratio measurements $R_{ij}$ are generated according to (\ref{eq:noise_purb}) such that the condition (\ref{eq:condition_purb}) holds for a fixed small $\epsilon>0$. Then there exists a critical probability $p^*_c(d)$ such that when $p > p^*_c(d)$, the solution $\hat G$ to the optimization problem (\ref{eq:min_G}) is close to the true Gram matrix $G$ in the sense that
\[
\left\Vert G-\hat{G}\right\Vert^2\leq \mathcal{O}\left(n^2\right)\epsilon^2
\]
 w.h.p. (as $n\to \infty$). Moreover, an upper bound $p_c(d)$ for $p^*_c(d)$ is given by (\ref{eq:p_c}).
\end{theorem}

\begin{proof}
First, the ``gain'' $f_1$ from the incorrect measurements remains the same, since the noise model for the ``bad'' edge set $\mathcal{E} \backslash\mathcal{E}_c$ is not changed. Thus, the lower bound for $f_1$ is given in (\ref{eq:f1_bound}). For the ``loss" from the good measurements, we
have
\begin{eqnarray}
f_{2} & = & \sum_{\left(i,j\right)\in\mathcal{E}_{c}}\left(\left\Vert I_{d}-R_{ij}+\Delta_{ij}\right\Vert -\left\Vert I_{d}-R_{ij}\right\Vert \right)\nonumber\\
 & \geq & \sum_{\left(i,j\right)\in\mathcal{E}_{c}}\left(\left\Vert \Delta_{ij}\right\Vert -2\left\Vert I_{d}-R_{ij}\right\Vert \right)\nonumber\\
 & \geq & \sum_{\left(i,j\right)\in\mathcal{E}_{c}}\left\Vert \Delta_{ij}\right\Vert -2p\left(n^{2}+\mathcal{O}_{P}\left(n\sqrt{\log n}\right)\right)\epsilon.\label{eq:f2_new_bound}
\end{eqnarray}
Applying Lemma \ref{lem:d_and_off} to (\ref{eq:f2_d_bound}) and (\ref{eq:f2_off_bound}), and setting $\alpha=\sqrt{1-\epsilon_{0}}$
and $\beta=\sqrt{\epsilon_{0}}$ , where $\epsilon_{0}\ll p-p_{c}$,
we obtain
\begin{eqnarray}
\sum_{\left(i,j\right)\in\mathcal{E}_{c}}\left\Vert \Delta_{ij}\right\Vert  & \geq & \sqrt{1-\epsilon_{0}}\sum_{\left(i,j\right)\in\mathcal{E}_{c}}\left\Vert \Delta_{ij}^{\text{d}}\right\Vert +\sqrt{\epsilon_{0}}\sum_{\left(i,j\right)\in\mathcal{E}_{c}}\left\Vert \Delta_{ij}^{\text{off}}\right\Vert \nonumber\\
 & \geq & c_{2}n\text{Tr}\left(T\right)+c_{3}n\sum_{i=1}^{n}\left\Vert Q_{ii}^{1}\right\Vert -\mathcal{O}_{P}\left(\sqrt{n\log n}\right)\text{Tr}\left(T\right),\label{eq:f2_D_new_bound}
\end{eqnarray}
where $c_{2}$ and $c_{3}$ are some constants. Thus, combining (\ref{eq:f1_bound}), (\ref{eq:f2_new_bound}) and
(\ref{eq:f2_D_new_bound}) together, we get
\begin{eqnarray}
&&F(G+\Delta)-F(G) \nonumber\\
& = & f_{1}+f_{2}\nonumber \\
 & \geq & c_{4}n\text{Tr}\left(T\right)+c_{3}n\sum_{i=1}^{n}\left\Vert Q_{ii}^{1}\right\Vert -\mathcal{O}_{P}\left(\sqrt{n\log n}\right)\text{Tr}\left(T\right)\nonumber \\
 &  & -\mathcal{O}_{P}\left(\sqrt{n\log n}\right)\sum_{i=1}^{n}\left\Vert Q_{ii}^{1}\right\Vert -2p\left(n^{2}+\mathcal{O}_{P}\left(n\sqrt{\log n}\right)\right)\epsilon,\label{eq:stability}
\end{eqnarray}
where $c_{4}$ is some constant. Thus, if the RHS of (\ref{eq:stability}) is greater than zero, then $G+\Delta$ is not the minimizer of $F$. In other words, if $G+\Delta$ is the minimizer of $F$, then the RHS of (\ref{eq:stability}) is not greater than zero. Let $n\rightarrow\infty$ in (\ref{eq:stability}), we obtain the necessary condition
for the minimizer $G+\Delta$ of $F$:
\begin{equation}
c_{4}\text{Tr}\left(T\right)+c_{3}\sum_{i=1}^{n}\left\Vert Q_{ii}^{1}\right\Vert \leq2pn\epsilon.\label{eq:condition_epsilon}
\end{equation}
To show that condition (\ref{eq:condition_epsilon}) leads to the conclusion that the amount of perturbation
\begin{equation}
\left\Vert \Delta\right\Vert ^{2}=\left\Vert P\right\Vert ^{2}+\left\Vert Q\right\Vert ^{2}+\left\Vert T\right\Vert ^{2}.\label{eq:Delta2_decompostion}
\end{equation}
is less than $\mathcal{O}\left(n^2\right)\epsilon^2$, we need the following lemmas to upper bound $\left\Vert \Delta\right\Vert^2$ by parts.
\begin{lemma}
\label{lem:P_T_norm2}
\begin{equation}
\left\Vert P\right\Vert ^{2}\leq\text{Tr}\left(T\right)^{2}\leq \mathcal{O}\left(n^2\right)\epsilon^2,\;\text{ and }\;
\left\Vert T\right\Vert ^{2}\leq\text{Tr}\left(T\right)^{2}\leq \mathcal{O}\left(n^2\right)\epsilon^2.\label{eq:P_T_norm2}
\end{equation}
\end{lemma}
\begin{proof}
Using the fact that for any positive semidefinite matrix $M$, $\left\Vert M\right\Vert\leq \text{Tr}\left(M\right)$, we have
\begin{equation}
\left\Vert P\right\Vert ^{2}=n^{2}\left\Vert P_{11}\right\Vert ^{2}=\left\Vert \sum_{i=1}^{n}T_{ii}\right\Vert ^{2}\leq\text{Tr}\left(\sum_{i=1}^{n}T_{ii}\right)^{2}=\text{Tr}\left(T\right)^{2}.\label{eq:P_trT}
\end{equation}
\begin{equation}
\left\Vert T\right\Vert ^{2}\leq\text{Tr}\left(T\right)^{2}.\label{eq:T_trT}
\end{equation}
And following (\ref{eq:condition_epsilon}) we obtain (\ref{eq:P_T_norm2}).
\end{proof}

\begin{lemma}
\label{lem:Q_norm2}
\begin{equation}
\left\Vert Q\right\Vert ^{2}\leq2\text{Tr}\left(T\right)^{2}-2\text{Tr}\left(T\right)+nd.\label{eq:Qnorm_trT}
\end{equation}
\end{lemma}
\begin{proof}
Since $G+\Delta\succcurlyeq0$, we can decompose $G+\Delta$
as $G+\Delta=\sum_{i=1}^{nd}\lambda_{i}{\bf v}_{i}{\bf v}_{i}^{T}$,
where the eigenvectors ${\bf v}_{i}\in\mathbb{R}^{nd}$ and the associated eigenvalues $\lambda_{i}\geq0$ for
all $i$. Then we further decompose each ${\bf v}_{i}$ as ${\bf v}_{i}={\bf v}_{i}^{\mathcal{S}}+{\bf v}_{i}^{\bar{\mathcal{S}}}$,
where ${\bf v}_{i}^{\mathcal{S}}\in\mathcal{S}$ and ${\bf v}_{i}^{\bar{\mathcal{S}}}\in\bar{\mathcal{S}}$.
Thus we obtain
\begin{eqnarray}
&&G+\Delta  =  \sum_{i=1}^{nd}\lambda_{i}{\bf v}_{i}{\bf v}_{i}^{T}=\sum_{i=1}^{nd}\lambda_{i}\left({\bf v}_{i}^{\mathcal{S}}+{\bf v}_{i}^{\bar{\mathcal{S}}}\right)\left({\bf v}_{i}^{\mathcal{S}}+{\bf v}_{i}^{\bar{\mathcal{S}}}\right)^{T} \label{eq:G_D_1} \\
 & = & \sum_{i=1}^{nd}\lambda_{i}{\bf v}_{i}^{\mathcal{S}}\left({\bf v}_{i}^{\mathcal{S}}\right)^{T}+\sum_{i=1}^{nd}\lambda_{i}{\bf v}_{i}^{\bar{\mathcal{S}}}\left({\bf v}_{i}^{\bar{\mathcal{S}}}\right)^{T}+\sum_{i=1}^{nd}\lambda_{i}{\bf v}_{i}^{\mathcal{S}}\left({\bf v}_{i}^{\bar{\mathcal{S}}}\right)^{T}+\sum_{i=1}^{nd}\lambda_{i}{\bf v}_{i}^{\bar{\mathcal{S}}}\left({\bf v}_{i}^{\mathcal{S}}\right)^{T}. \nonumber
\end{eqnarray}
On the other hand, we can decompose $G+\Delta$ as
\begin{equation}
G+\Delta=\left(G+P\right)+Q^{1}+Q^{2}+T,\label{eq:G_D_2}
\end{equation}
where $G+P\in\mathcal{S}\otimes\mathcal{S}$, $Q^{1}\in\mathcal{S}\otimes\bar{\mathcal{S}}$,
$Q^{2}\in\bar{\mathcal{S}}\otimes\mathcal{S}$ and $T\in\bar{\mathcal{S}}\otimes\bar{\mathcal{S}}$.
By comparing the right hand sides of (\ref{eq:G_D_1}) and (\ref{eq:G_D_2}),
we conclude that
\begin{align*}
G+P=\sum_{i=1}^{nd}\lambda_{i}{\bf v}_{i}^{\mathcal{S}}\left({\bf v}_{i}^{\mathcal{S}}\right)^{T},\quad & T=\sum_{i=1}^{nd}\lambda_{i}{\bf v}_{i}^{\bar{\mathcal{S}}}\left({\bf v}_{i}^{\bar{\mathcal{S}}}\right)^{T},\\
Q^{1}=\sum_{i=1}^{nd}\lambda_{i}{\bf v}_{i}^{\mathcal{S}}\left({\bf v}_{i}^{\bar{\mathcal{S}}}\right)^{T},\quad & Q^{2}=\sum_{i=1}^{nd}\lambda_{i}{\bf v}_{i}^{\bar{\mathcal{S}}}\left({\bf v}_{i}^{\mathcal{S}}\right)^{T}=\left(Q^{1}\right)^{T}.
\end{align*}
Therefore, we have
\begin{eqnarray}
&&\left\Vert Q\right\Vert ^{2}  = 2\left\Vert Q^{1}\right\Vert ^{2} =  2\left\Vert \sum_{i=1}^{nd}\lambda_{i}{\bf v}_{i}^{\mathcal{S}}\left({\bf v}_{i}^{\bar{\mathcal{S}}}\right)^{T}\right\Vert ^{2}\nonumber \\
 & = & 2\text{Tr}\left(\sum_{i=1}^{nd}\lambda_{i}{\bf v}_{i}^{\bar{\mathcal{S}}}\left({\bf v}_{i}^{\mathcal{S}}\right)^{T}\sum_{j=1}^{nd}\lambda_{j}{\bf v}_{j}^{\mathcal{S}}\left({\bf v}_{j}^{\bar{\mathcal{S}}}\right)^{T}\right)\nonumber \\
 & = & 2\sum_{i,j=1}^{nd}\lambda_{i}\lambda_{j}\left({\bf v}_{i}^{\mathcal{S}}\right)^{T}{\bf v}_{j}^{\mathcal{S}}\left({\bf v}_{j}^{\bar{\mathcal{S}}}\right)^{T}{\bf v}_{i}^{\bar{\mathcal{S}}}\nonumber \\
 & \leq & \sum_{i,j=1}^{nd}\lambda_{i}\lambda_{j}\left(\left({\bf v}_{i}^{\mathcal{S}}\right)^{T}{\bf v}_{j}^{\mathcal{S}}\left({\bf v}_{i}^{\mathcal{S}}\right)^{T}{\bf v}_{j}^{\mathcal{S}}+\left({\bf v}_{j}^{\bar{\mathcal{S}}}\right)^{T}{\bf v}_{i}^{\mathcal{\bar{\mathcal{S}}}}\left({\bf v}_{j}^{\bar{\mathcal{S}}}\right)^{T}{\bf v}_{i}^{\mathcal{\bar{\mathcal{S}}}}\right)\nonumber \\
 & = & \text{Tr}\left(\sum_{i=1}^{nd}\lambda_{i}{\bf v}_{i}^{\mathcal{S}}\left({\bf v}_{i}^{\mathcal{S}}\right)^{T}\sum_{j=1}^{nd}\lambda_{j}{\bf v}_{j}^{\mathcal{S}}\left({\bf v}_{j}^{\mathcal{S}}\right)^{T}\right)+\text{Tr}\left(\sum_{i=1}^{nd}\lambda_{i}{\bf v}_{i}^{\mathcal{\bar{\mathcal{S}}}}\left({\bf v}_{i}^{\mathcal{\bar{\mathcal{S}}}}\right)^{T}\sum_{j=1}^{nd}\lambda_{j}{\bf v}_{j}^{\mathcal{\bar{\mathcal{S}}}}\left({\bf v}_{j}^{\mathcal{\bar{\mathcal{S}}}}\right)^{T}\right)\nonumber \\
 & = & \left\Vert G+P\right\Vert ^{2}+\left\Vert T\right\Vert ^{2}.\label{eq:Qnorm_bound}
\end{eqnarray}
Using the fact that $G=I_{nd}$ and $nP_{11}=-\sum_{i}T_{ii}$, we
obtain
\begin{eqnarray}
\left\Vert G+P\right\Vert ^{2} & = & \left\Vert G\right\Vert ^{2}+\left\Vert P\right\Vert ^{2}+2\text{Tr}\left\langle G,P\right\rangle \nonumber \\
 & = & nd+\left\Vert P\right\Vert ^{2}-2\text{Tr}\left(T\right).\label{eq:G_P_trT}
\end{eqnarray}
Combining (\ref{eq:P_trT}), (\ref{eq:T_trT}), (\ref{eq:G_P_trT})
and (\ref{eq:Qnorm_bound}), we get (\ref{eq:Qnorm_trT}).
\end{proof}

Using Lemma \ref{lem:P_T_norm2} and Lemma \ref{lem:Q_norm2}, we reach the conclusion in the theorem.
\end{proof}

\paragraph{Remark}Using Lemma \ref{lem:Q_norm2}, $\left\Vert Q\right\Vert^2\leq \mathcal{O}\left(n^2\right)\epsilon^2$ holds when $\epsilon\gg1/\sqrt{n}$, which leads to the weak stability result of LUD stated in Theorem \ref{thm:stability_weak} that requires $\epsilon$ to be fixed. In fact, a stronger stability result of LUD that allows $\epsilon\rightarrow 0$ as $n \to \infty$ can also be proven using ideas similar to the proof of Theorem \ref{thm:p_c}.
The proof of strong stability can be found in Appendix \ref{sec:strong_stability}.

\begin{theorem}
\label{thm:stability_strong}
\emph{(Strong stability)}
Assume that all pairwise ratio measurements $R_{ij}$ are generated according to (\ref{eq:noise_purb}) such that the condition (\ref{eq:condition_purb}) holds for arbitrary small $\epsilon>0$. Then there exists a critical probability $p^*_c(d)$ such that when $p > p^*_c(d)$, the solution $\hat G$ to the optimization problem (\ref{eq:min_G}) is close to the true Gram matrix $G$ in the sense that
\[
\left\Vert G-\hat{G}\right\Vert^2\leq \mathcal{O}\left(n^2\right)\epsilon^2
\]
 w.h.p. (as $n\to \infty$). Moreover, an upper bound $p_c(d)$ for $p^*_c(d)$ is given by (\ref{eq:p_c}).
\end{theorem}

\section{A Generalization of LUD to random incomplete measurement graphs}
\label{sec:missing_entries}
The analysis of exact and stable recovery of rotations from full measurements (Section \ref{sec:exact} and \ref{sec:stability}) can be straightforwardly generalized to the case of random incomplete measurement graphs. Here we assume that the edge set $\mathcal{E}$, which is the index set of measured rotation ratios $R_{ij}$, is a realization of a random graph drawn from the Erd\H{o}s-R\'enyi model $\mathcal{G}(n,p_1), p_1\geq 2\log(n)/n$, and the rotation measurements $R_{ij}$ in the edge set $\mathcal{E}$ are generated according to (\ref{eq:noise}) or (\ref{eq:noise_purb}). The reason why we have the restriction that $p_1\geq 2\log(n)/n$ is that as $n$ tends to infinity, the probability that a graph on $n$ vertices with edge probability $2\log(n)/n$ is connected, tends to 1. The ``good'' edge set $\mathcal{E}_c$ and the ``bad'' edge set $\mathcal{E}\backslash \mathcal{E}_{c}$ can be seen as realizations of random graphs drawn from the Erd\H{o}s-R\'enyi models $\mathcal{G}(n,p_1(1-p))$ and $\mathcal{G}(n,p_1p)$, respectively. As a consequence, we can apply the same arguments in Section \ref{sec:exact} and \ref{sec:stability} and obtain the following theorems that are analogous to Theorem \ref{thm:p_c}, \ref{thm:stability_weak} and \ref{thm:stability_strong}. The associated numerical results are provided in Section \ref{sec:incomplete_exp}.
\begin{theorem}
\label{thm:p_c_missing}
Assume that the index set of measured rotation ratios $\mathcal{E}$ is a realization of a random graph drawn from the Erd\H{o}s-R\'enyi model $\mathcal{G}(n,p_1), p_1\geq 2\log(n)/n$, and the rotation ratio measurements $R_{ij}$  in $\mathcal{E}$ are generated according to (\ref{eq:noise}). Then there exists a critical probability $p^{*}_c(d,p_1)$ such that when $p > p^{*}_c(d,p_1)$, the Gram matrix $G$ is exactly recovered by the solution to the optimization problem (\ref{eq:min_G}) w.h.p. (as $n\to \infty$). Moreover, an upper bound $p_c(d,p_1)$ for $p^{*}_c(d,p_1)$ is
\begin{equation}
p_c(d,p_1)=1-\left(\frac{-c_{1}\left(d\right)+\sqrt{c_{1}\left(d\right)^2+8p_1\left(c\left(d\right)+2/\sqrt{d}\right)/\sqrt{d}}}{2\sqrt{p_1}\left(c\left(d\right)+2/\sqrt{d}\right)}\right)^{2}\label{eq:p_c_missing},
\end{equation}
where  $c\left(d\right)$ and $c_1(d)$ are constants defined in (\ref{eq:def_cd}) and (\ref{eq:c_1}). In particular, when $p_1 = 1$, $p_c(d,1)=p_c(d)$ in (\ref{eq:p_c}).
\end{theorem}

\begin{theorem}
\label{thm:stability_weak_missing}
\emph{(Weak stability)}
Assume that the index set of measurements $\mathcal{E}$ is generalized as Theorem \ref{thm:p_c_missing}, and the rotation ratio measurements $R_{ij}$  in $\mathcal{E}$ are generated according to (\ref{eq:noise_purb}) such that the condition (\ref{eq:condition_purb}) holds for a fixed small $\epsilon>0$. Then there exists a critical probability $p^*_c(d, p_1)$ such that when $p > p^*_c(d,p_1)$, the solution $\hat G$ to the optimization problem (\ref{eq:min_G}) is close to the true Gram matrix $G$ in the sense that
\[
\left\Vert G-\hat{G}\right\Vert^2\leq \mathcal{O}\left(n^2\right)\epsilon^2
\]
 w.h.p. (as $n\to \infty$). Moreover, an upper bound $p_c(d,p_1)$ for $p^*_c(d,p_1)$ is given by (\ref{eq:p_c_missing}).
\end{theorem}

\begin{theorem}
\label{thm:stability_strong_missing}
\emph{(Strong stability)}
Assume that the index set of measurements $\mathcal{E}$ is generalized as Theorem \ref{thm:p_c_missing}, and the rotation ratio measurements $R_{ij}$  in $\mathcal{E}$ are generated according to (\ref{eq:noise_purb}) such that the condition (\ref{eq:condition_purb}) holds for an arbitrary small $\epsilon>0$. Then there exists a critical probability $p^*_c(d, p_1)$ such that when $p > p^*_c(d,p_1)$, the solution $\hat G$ to the optimization problem (\ref{eq:min_G}) is close to the true Gram matrix $G$ in the sense that
\[
\left\Vert G-\hat{G}\right\Vert^2\leq \mathcal{O}\left(n^2\right)\epsilon^2
\]
 w.h.p. (as $n\to \infty$). Moreover, an upper bound $p_c(d,p_1)$ for $p^*_c(d,p_1)$ is given by (\ref{eq:p_c_missing}).
\end{theorem}

\section{Alternating Direction Augmented Lagrangian method (ADM)}
\label{sec:ADM}

Here we briefly describe the ADM \cite{ADM} to solve the non-smooth minimization problem  (\ref{eq:min_G}). ADM is a multiple-splitting algorithm that minimizes the dual augmented Lagrangian function sequentially regarding the Lagrange multipliers, then the dual slack variables, and finally the primal variables in each step. In addition, in the minimization over a certain variable, the other variables are kept fixed.
The optimization problem (\ref{eq:min_G}) can be written
as
\begin{equation}
\min_{X_{ij},G\succcurlyeq0}\sum_{i<j}\left\Vert X_{ij}\right\Vert ,\text{ s.t. }\mathcal{A}\left(G\right)={\bf b},\, X_{ij}=R_{ij}-G_{ij},\label{eq:primal}
\end{equation}
where the operator $\mathcal{A}:\mathbb{R}^{nd\times nd}\rightarrow\mathbb{R}^{nd^{2}}$
is defined as
\[
\mathcal{A}\left(G\right)=\left(G_{ii}^{pq}\right)_{i=1,\ldots,n,\, p,q=1,\ldots d}
\]
and the row vector ${\bf b}\in\mathbb{R}^{nd^{2}}$ is
\[
{\bf b}={\bf 1}_{n}\otimes\left(\mathcal{X}_{\left(p=q\right)}\left(p,q\right)\right)_{p,q=1,\ldots d}.
\]
Since
\begin{align}
 & \max_{\theta_{ij},y,W\succcurlyeq0}\,\min_{X_{ij},G}\sum_{i<j}\left(\left\Vert X_{ij}\right\Vert -\left\langle \theta_{ij},\, X_{ij}-R_{ij}+G_{ij}\right\rangle \right)-\left\langle {\bf y},\,\mathcal{A}\left(G\right)-b\right\rangle -\left\langle G,W\right\rangle \nonumber \\
\Leftrightarrow & \max_{\theta_{ij},y,W\succcurlyeq0}\,\min_{X_{ij},G}-\left\langle Q\left(\theta\right)+W+\mathcal{A}^{*}\left({\bf y}\right),G\right\rangle +{\bf y}{\bf b}^{T}+\sum_{i<j}\left(\left\Vert X_{ij}\right\Vert -\left\langle \theta_{ij},\, X_{ij}\right\rangle +\left\langle \theta_{ij},\, R_{ij}\right\rangle \right)\label{eq:max_min}
\end{align}
where
\[
Q\left(\theta\right)=\frac{1}{2}\left(\begin{array}{cccc}
0 & \theta_{12} & \cdots & \theta_{1m}\\
\theta_{12}^{T} & 0 & \cdots & \theta_{2m}\\
\vdots & \vdots & \ddots & \vdots\\
\theta_{1m}^{T} & \theta_{2m}^{T} & \cdots & 0
\end{array}\right).
\]
We want to first minimize the function over $X_{ij}$ and $G$ in (\ref{eq:max_min}).
The rearrangement of terms in (\ref{eq:max_min}) enable us to minimize
$-\left\langle Q\left(\theta\right)+W+\mathcal{A}^{*}\left({\bf y}\right),G\right\rangle $
over $G$ and minimize $\left\Vert X_{ij}\right\Vert -\left\langle \theta_{ij},\, X_{ij}\right\rangle $
over $X_{ij}$, $i<j$ separately. To minimize $-\left\langle Q\left(\theta\right)+W+\mathcal{A}^{*}\left({\bf y}\right),G\right\rangle $
over $G$, the optimum value will be $-\infty$ if $Q\left(\theta\right)+W+\mathcal{A}^{*}\left({\bf y}\right)\neq0$.
Therefore due to the dual feasibility $Q\left(\theta\right)+W+\mathcal{A}^{*}\left({\bf y}\right)=0$
and the optimum value is zero. Since
\begin{align}
 & \left\Vert X_{ij}\right\Vert -\left\langle \theta_{ij},\, X_{ij}\right\rangle \nonumber \\
= & \left\Vert X_{ij}\right\Vert -\left\Vert X_{ij}\right\Vert \left\Vert \theta_{ij}\right\Vert \left\langle \theta_{ij}/\left\Vert \theta_{ij}\right\Vert ,\, X_{ij}/\left\Vert X_{ij}\right\Vert \right\rangle \nonumber \\
= & \left\Vert X_{ij}\right\Vert \left(1-\left\Vert \theta_{ij}\right\Vert \left\langle \theta_{ij}/\left\Vert \theta_{ij}\right\Vert ,\, X_{ij}/\left\Vert X_{ij}\right\Vert \right\rangle \right)\label{eq:X_line_1}\\
\geq & \left\Vert X_{ij}\right\Vert \left(1-\left\Vert \theta_{ij}\right\Vert \right),\label{eq:X_line_2}
\end{align}
to minimize $\left\Vert X_{ij}\right\Vert -\left\langle \theta_{ij},\, X_{ij}\right\rangle $
over $X_{ij}$, if $\left\Vert \theta_{ij}\right\Vert >1$, then let
$X_{ij}=\alpha\theta_{ij}$, $\alpha>0$ and then from (\ref{eq:X_line_1})
$\left\Vert X_{ij}\right\Vert -\left\langle \theta_{ij},\, X_{ij}\right\rangle =\alpha\left\Vert \theta_{ij}\right\Vert \left(1-\left\Vert \theta_{ij}\right\Vert \right)$
goes to $-\infty$ if $\alpha$ goes to $+\infty$. Hence $\left\Vert \theta_{ij}\right\Vert \leq1$
and from (\ref{eq:X_line_2}) we get the optimum value is zero where
$X_{ij}=0$. Therefore the dual problem is
\begin{equation}
\min_{\theta_{ij},y,W\succcurlyeq0}-{\bf y}{\bf b}^{T}-\sum_{i<j}\left\langle \theta_{ij},\, R_{ij}\right\rangle \text{ s.t. \ensuremath{\left\Vert \theta_{ij}\right\Vert \leq}1,}\, Q\left(\theta\right)+W+\mathcal{A}^{*}\left({\bf y}\right)=0.\label{eq:dual}
\end{equation}
The augmented Lagrangian function of the dual problem (\ref{eq:dual})
is
\begin{eqnarray}
\mathcal{L}\left({\bf y},\theta,W,G\right) & = & -{\bf y}{\bf b}^{T}-\sum_{i<j}\left\langle \theta_{ij},\, R_{ij}\right\rangle +\left\langle Q\left(\theta\right)+W+\mathcal{A}^{*}\left({\bf y}\right),G\right\rangle \label{eq:L}\\
 &  & +\frac{\mu}{2}\left\Vert Q\left(\theta\right)+W+\mathcal{A}^{*}\left({\bf y}\right)\right\Vert _{F}^{2},\text{ s.t. \ensuremath{\left\Vert \theta_{ij}\right\Vert \leq}1},\nonumber
\end{eqnarray}
where $\mu>0$ is a penalty parameter. Then we can devise an alternating
direction method (ADM) that minimizes (\ref{eq:L}) with respect to
${\bf y},\theta,W,\text{ and }G$ in an alternating fashion, that is,
given some initial guess ${\bf y}^{0},\theta^{0},W^{0},\text{ and }G^{0}$,
the simplest ADM method solves the following three subproblems sequentially
in each iteration:
\begin{eqnarray}
{\bf y}^{k+1} & = & \arg\min_{\mathbf{y}}\mathcal{L}\left(\mathbf{y},\theta^{k},W^{k},G^{k}\right)\label{eq:y}\\
\theta_{ij}^{k+1} & = & \arg\min_{\theta_{ij}}\mathcal{L}\left(\mathbf{y}^{k+1},\theta,W^{k},G^{k}\right),\text{ s.t. \ensuremath{\left\Vert \theta_{ij}\right\Vert \leq}1}\label{eq:theta}\\
W^{k+1} & = & \arg\min_{W\succcurlyeq0}\mathcal{L}\left(\mathbf{y}^{k+1},\theta^{k+1},W,G^{k}\right),\label{eq:W}
\end{eqnarray}
and updates the Lagrange multiplier $G$ by
\begin{equation}
G^{k+1}=G^{k}+\gamma\mu\left(Q\left(\theta^{k+1}\right)+W^{k+1}+\mathcal{A}^{*}\left(\mathbf{y}^{k+1}\right)\right),\label{eq:G_ADM}
\end{equation}
where $\gamma\in\left(0,\frac{1+\sqrt{5}}{2}\right)$ is an appropriately
chosen step length.

To solve (\ref{eq:y}), set $\nabla_{\mathbf{y}}\mathcal{L}=0$ and
using $\mathcal{AA}^{*}=I$, we obtain
\[
\mathbf{y}^{k+1}=-\mathcal{A}\left(Q\left(\theta^{k}\right)+W^{k}\right)-\frac{1}{\mu}\left(\mathcal{A}\left(G^{k}\right)-\mathbf{b}\right).
\]

By rearrangement of terms of $\mathcal{L}$, it is easy to see problem
(\ref{eq:theta}) is equivalent to
\[
\min_{\theta_{ij}}-\left\langle \theta_{ij},R_{ij}\right\rangle +\frac{\mu}{2}\left\Vert \theta_{ij}-\Phi_{ij}\right\Vert _{F}^{2},\, s.t.\left\Vert \theta_{ij}\right\Vert \leq1
\]
where $\Phi=W^{k}+\mathcal{A}^{*}\left(\mathbf{y}^{k+1}\right)+\frac{1}{\mu}G^{k}.$
And it can be further simplified as
\[
\min_{\theta_{ij}}\left\langle \theta_{ij},\mu\Phi_{ij}-R_{ij}\right\rangle +\frac{\mu}{2}\left\Vert \theta_{ij}\right\Vert ^{2},\, s.t.\left\Vert \theta_{ij}\right\Vert \leq1
\]
whose solution is
\[
\theta_{ij}=\begin{cases}
\frac{1}{\mu}R_{ij}-\Phi_{ij}&\text{if } \left\Vert\frac{1}{\mu}R_{ij}-\Phi_{ij}\right\Vert\leq1,\\
\frac{R_{ij}-\mu\Phi_{ij}}{\left\Vert R_{ij}-\mu\Phi_{ij}\right\Vert }& \text{ otherwise.},
\end{cases}
\]

Problem (\ref{eq:W}) is equivalent to
\[
\min\left\Vert W-H^{k}\right\Vert _{F}^{2},\text{ s.t. }W\succcurlyeq0,
\]
where $H^{k}=-Q\left(\theta^{k+1}\right)-\mathcal{A}^{*}\left(\mathbf{y}^{k+1}\right)-\frac{1}{\mu}G^{k}.$
Hence we obtain the solution $W^{k+1}=V_{+}\Sigma_{+}V_{+}^{T},$
where
\begin{equation}
\label{eq:eig}
V\Sigma V=\left(V_{+}V_{-}\right)\left(\begin{array}{cc}
\Sigma_{+} & 0\\
0 & \Sigma_{-}
\end{array}\right)\left(\begin{array}{c}
V_{+}^{T}\\
V_{-}^{T}
\end{array}\right)
\end{equation}
is the spectral decomposition of the matrix $H^{k}$, and $\Sigma_{+}$
and $\Sigma_{-}$ are the positive and negative eigenvalues of $H^{k}.$

Following (\ref{eq:G_ADM}), we have
\[
G^{k+1}=\left(1-\gamma\right)G^{k}+\gamma\mu\left(W^{k+1}-H^{k}\right).
\]

The convergence analysis and the practical issues related to how to
take advantage of low-rank assumption of $G$ in the eigenvalue decomposition
performed at each iteration, strategies for adjusting the penalty
parameter $\mu$, the use of a step size $\gamma$ for updating the
primal variable $X$ and termination rules using the in-feasibility
measures are discussed in details in \cite{ADM}.
According to the convergence rate analysis of ADM in \cite{ADM_convergence_rate}, we need $\mathcal{O}(1/\delta)$ iterations to reach a $\delta$ accuracy.
At each iteration, the most time-consuming step of ADM is the
computation of the eigenvalue decomposition in (\ref{eq:eig}) . Fortunately, for the synchronization problem, the primal solution $G$ is a low rank matrix (i.e. rank($G$) = $d$).
Moreover, since the optimal solution pair $({\bf y}, \theta, W, G)$ satisfies the complementary condition $WG = 0$,
the matrices $W$ and $G$ share the same set of eigenvectors and the positive eigenvalues of $G$ corresponds to zero eigenvalues of $W$. Therefore, at $k$th iteration we only need to compute $V_{-}$, the part corresponding to the negative eigenvalues of $W^k$. Thus to take advantage of the low rank structure of $G$, we use the Arnoldi iterations \cite{laug} to compute first few negative eigenvectors of $W^{k}$. However, for the noisy case, the optimal solution $G$ may have rank greater than $d$, and also during the iterations the rank of the solution $G^k$ may increase. Correspondingly, during the iterations $W^k$ may have more than $d$ negative eigenvalues. Therefore it is impossible to decide ahead of time how many negative eigenvectors of $W^k$ are required. A heuristic that could work well in practice is to compute only eigenvectors whose eigenvalues are smaller than some small negative threshold epsilon, with the hope that the number of such eigenvectors would be $o(n)$, yet not effecting the convergence of the algorithm. The Arnoldi iterations require $\mathcal{O}(n^3)$ operations if $\mathcal{O}(n)$  eigenvalues need to be computed. However,  when first few negative eigenvalues of $W_k$ are required, the time cost by the Arnoldi iterations will be much reduced.

%The time complexity of  ADM for LUD is $\mathcal{O}(n^3/\delta)$, where $\delta$ is the accuracy of ADM. Since in each iteration, the most expensive process is the Arnoldi iterations \cite{laug} to compute the eigen-decomposition of $W^{k+1}$ whose cost is $\mathcal{O}(n^3)$, and according to the convergence rate analysis of ADM in \cite{ADM_convergence_rate}, we need $\mathcal{O}(1/\delta)$ iterations to reach a $\delta$ accuracy. Note that the time cost $\mathcal{O}(n^3)$ of an eigen-decomposition can be  reduced if only a few largest eigenvalues/eigenvectors are needed to compute, which we can take advantage of since the solution $G$ is low-rank (see more details in section 3.1 in \cite{ADM}).

\section{Numerical experiments}
\label{sec:exp}

All numerical experiments were performed on a machine with 2 Intel(R)
Xeon(R) CPUs X5570, each with 4 cores, running at 2.93 GHz. We simulated $100$, $500$ and $1000$ rotations in the groups $SO\left(2\right)$ and  $SO\left(3\right)$ respectively. The noise is added to the rotation ratio measurements according to the ER random graph model $\mathcal{G}\left(n,p\right)$  (\ref{eq:noise}) and the model (\ref{eq:noise_purb}) with small perturbations on ``good'' edge set in subsection \ref{sec:exp1} and subsection \ref{sec:exp2}, respectively, where $n$ is total number of rotations, and $p$ is the proportion of good rotation ratio measurements.

We define the relative error of the estimated Gram matrix $\hat G$ as
\begin{equation}
\text{RE}\left(\hat{G},G\right)=\left\Vert \hat{G}-G\right\Vert/\left\Vert G\right\Vert,\label{eq:re_G}
\end{equation}
and the mean squared error (MSE) of the estimated rotation matrices
$\hat{R}_{1},\ldots,\hat{R}_{n}$ as
\begin{equation}
\text{MSE}=\frac{1}{n}\sum_{i=1}^{n}\left\Vert R_{i}-\hat{O}\hat{R}_{i}\right\Vert ^{2},\label{eq:MSE}
\end{equation}
where $\hat{O}$ is the optimal solution to the registration problem
between the two sets of rotations $\left\{ R_{1},\ldots,R_{n}\right\} $
and $\left\{ \hat{R}_{1},\ldots,\hat{R}_{n}\right\} $ in the sense
of minimizing the MSE. As shown in \cite{cryoem_eig_sdp}, there is a simple procedure to obtain both $\hat{O}$ and the MSE from the singular value decomposition of the matrix $\frac{1}{n}\sum_{i=1}^{n}\hat{R}_{i}R_{i}^{T}$. For each experiment with fixed $d$, $n$, $p$ and $\kappa$, we run 10 trials and record the mean of REs and MSEs.

We compare LUD to EIG, SDP \cite{ang_sync} (using the SDP solver SDPLR \cite{Burer01anonlinear}). EIG and SDP are two algorithms to solve the least squares problem in (\ref{LS}) with equal weights using spectral relaxation and semidefinite relaxation, respectively. LUD does not have advantage in the running time. In our experiments, the running time of LUD using ADM is about 10 to 20 times slower than that of SDP, and it is hundreds times slower than EIG. We will focus on the comparison of the accuracy of the rotation recovery using the three algorithms.

\subsection{Experiments with full measurements}
\subsubsection{E1: Exact Recovery by LUD}
\label{sec:exp1}
In this experiment, we use LUD to recover rotations in $SO\left(2\right)$ and $SO\left(3\right)$ with different values of $n$ and $p$ in the noise model (\ref{eq:noise}).  Table \ref{tab:EGram} shows that when $n$ is large enough, the critical probability where the Gram matrix $G$ can be exactly recovered is very close to $p_{c}\left(2\right)\approx0.4570,\: p_{c}\left(3\right)\approx0.4912$.

The comparison of the accuracy of the estimated rotations by EIG, SDP and LUD is shown in Tables \ref{tab:MSE_SO2}-\ref{tab:MSE_SO3} and Figure \ref{fig:MSE-1}, that demonstrate LUD outperforms EIG and SDP in terms of accuracy.

\begin{table}[H]
\centering{}\subfloat[\label{tab:EGram_SO2}$SO\left(2\right)$]{%
\begin{tabular}{|c|c|c|c|c|c|c|}
\hline
\backslashbox{$n$}{$p$} & 0.9 & 0.8 & 0.7 & 0.6 & 0.5 & 0.4\tabularnewline
\hline
$100$ & 0.0002 & 0.0004 & 0.0007 & 0.0007 & 0.0613 & 0.3243\tabularnewline
\hline
$500$ & 0.0002 & 0.0002 & 0.0004 & 0.0005 & 0.0007 & 0.1002\tabularnewline
\hline
$1000$ & 0.0001 & 0.0002 & 0.0003 & 0.0006 & 0.0007 & 0.0511\tabularnewline
\hline
\end{tabular}}
\hfill{}
\centering{}
\subfloat[\label{tab:EGram_SO3}$SO\left(3\right)$]{%
\begin{tabular}{|c|c|c|c|c|c|c|}
\hline
\backslashbox{$n$}{$p$} & 0.9 & 0.8 & 0.7 & 0.6 & 0.5 & 0.4\tabularnewline
\hline
$100$ & 0.0001 & 0.0001 & 0.0002 & 0.0029 & 0.1559 & 0.4772\tabularnewline
\hline
$500$ & 0.0001 & 0.0001 & 0.0001 & 0.0002 & 0.0011 & 0.3007\tabularnewline
\hline
$1000$ &0.0001  &0.0001  & 0.0001 & 0.0002 & 0.0007 & 0.2146\tabularnewline
\hline
\end{tabular}}

\caption{\label{tab:EGram}The Relative Error (\ref{eq:re_G}) of the Gram matrix $\hat G$ obtained by LUD in $\bf{E1}$. The critical probability where the Gram matrix $G$ can be exactly recovered is upper bounded by $p_{c}\left(2\right)\approx0.4570,\: p_{c}\left(3\right)\approx0.4912$ when $n$ is large enough.}
\end{table}

\begin{table}[H]
\centering
\begin{tabular}{c}
\subfloat[\label{tab:SO2_100}$n=100$]{%
\begin{tabular}{|c|c|c|c|c|c|c|c|}
\hline
$p$  & 0.7 & 0.6 & 0.5 & 0.4 & 0.3 & 0.2 \tabularnewline
\hline
EIG  & 0.0064 & 0.0115 & 0.0222 & 0.0419 & 0.0998 & 0.4285 \tabularnewline
\hline
SDP & 0.0065& 0.0116 & 0.0225 & 0.0427 & 0.1014 & 0.3971 \tabularnewline
\hline
LUD  & \bf1.7e-07 &\bf 4.7e-08 &\bf 8.4e-05 &\bf 0.0043 &\bf 0.0374 &\bf 0.3296 \tabularnewline
\hline
\end{tabular}}
\\
\\
\subfloat[\label{tab:SO2_500}$n=500$]{%
\begin{tabular}{|c|c|c|c|c|c|c|c|}
\hline
$p$  & 0.7 & 0.6 & 0.5 & 0.4 & 0.3 & 0.2 \tabularnewline
\hline
EIG  & 0.0012 & 0.0023 & 0.0040 & 0.0077 & 0.0163 & 0.0445 \tabularnewline
\hline
SDP  & 0.0012 & 0.0023 & 0.0041 & 0.0078 & 0.0164 & 0.0440 \tabularnewline
\hline
LUD & \bf6.4e-10 & \bf 5.5e-09 & \bf 9.6e-09 & \bf 6.3e-05 & \bf 0.0025 &\bf 0.0211 \tabularnewline
\hline
\end{tabular}}
\\
\\
\subfloat[\label{tab:SO2_1000}$n=1000$]{%
\begin{tabular}{|c|c|c|c|c|c|c|c|}
\hline
$p$  & 0.7 & 0.6 & 0.5 & 0.4 & 0.3 & 0.2 \tabularnewline
\hline
EIG  & 0.0006 & 0.0011 & 0.0020 & 0.0037 & 0.0080 & 0.0207 \tabularnewline
\hline
SDP  & 0.0006 & 0.0011 & 0.0020 & 0.0037 & 0.0081 & 0.0207 \tabularnewline
\hline
LUD  & \bf 3.0e-10 & \bf 1.5e-09 &\bf 7.3e-09 &\bf 7.5e-06 &\bf 0.0010 &\bf 0.0084 \tabularnewline
\hline
\end{tabular}}
\end{tabular}
\caption{\label{tab:MSE_SO2}MSE (\ref{eq:MSE}) of the estimated rotations in $SO\left(2\right)$ using EIG, SDP and LUD in $\bf{E1}$.  The critical probability where the rotations can be exactly recovered is upper bounded by $p_{c}\left(2\right)\approx0.4570$ when $n$ is large enough.}
\end{table}

\begin{table}[H]
\centering
\begin{tabular}{c}
\subfloat[\label{tab:SO3_100}$n=100$]{%
\begin{tabular}{|c|c|c|c|c|c|c|c|}
\hline
$p$ & 0.7 & 0.6 & 0.5 & 0.4 & 0.3 & 0.2 \tabularnewline
\hline
EIG  & 0.0063 & 0.0120 & 0.0224 & 0.0435 & 0.0920 & 0.3716 \tabularnewline
\hline
SDP & 0.0064 & 0.0122 & 0.0233 & 0.0452 & 0.0968 & 0.4040 \tabularnewline
\hline
LUD  & \bf 1.0e-09 &\bf 6.4e-07 &\bf 4.1e-04 & \bf 0.0094 &\bf 0.0461 &\bf 0.2700 \tabularnewline
\hline
\end{tabular}}
\\
\\
\subfloat[\label{tab:SO3_500}$n=500$]{%
\begin{tabular}{|c|c|c|c|c|c|c|c|}
\hline
$p$ & 0.7 & 0.6 & 0.5 & 0.4 & 0.3 & 0.2 \tabularnewline
\hline
EIG  & 0.0012 & 0.0023 & 0.0041 & 0.0080 & 0.0164 & 0.0442 \tabularnewline
\hline
SDP  & 0.0012 & 0.0023 & 0.0041 & 0.0080 & 0.0165 & 0.0447 \tabularnewline
\hline
LUD  &\bf 4.7e-11 & \bf 1.8e-10 &\bf 2.1e-09 &\bf 0.0006 &\bf 0.0061 &\bf 0.0295 \tabularnewline
\hline
\end{tabular}}
\\
\\
\subfloat[\label{tab:SO3_1000}$n=1000$]{%
\begin{tabular}{|c|c|c|c|c|c|c|c|}
\hline
$p$  & 0.7 & 0.6 & 0.5 & 0.4 & 0.3 & 0.2 \tabularnewline
\hline
EIG   & 0.0006 & 0.0011 & 0.0020 & 0.0037 & 0.0079 & 0.0208 \tabularnewline
\hline
SDP  & 0.0006  & 0.0011 & 0.0020 & 0.0038 & 0.0079 & 0.0209 \tabularnewline
\hline
LUD  & \bf 2.5e-11 & \bf 2.4e-10 & \bf 8.0e-10 & \bf 0.0001 & \bf 0.0026 & \bf 0.0131  \tabularnewline
\hline
\end{tabular}}
\end{tabular}
\caption{\label{tab:MSE_SO3}MSE (\ref{eq:MSE}) of the estimated rotations in $SO\left(3\right)$ using EIG, SDP and LUD in $\bf{E1}$. The critical probability where the rotations can be exactly recovered is upper bounded by $p_{c}\left(3\right)\approx0.4912$ when $n$ is large enough.}
\end{table}

\begin{figure}[H]
\centering
\subfloat[\label{fig:MSE_SO2}$SO(2)$]{
\centering
\includegraphics[width=0.6\paperwidth]{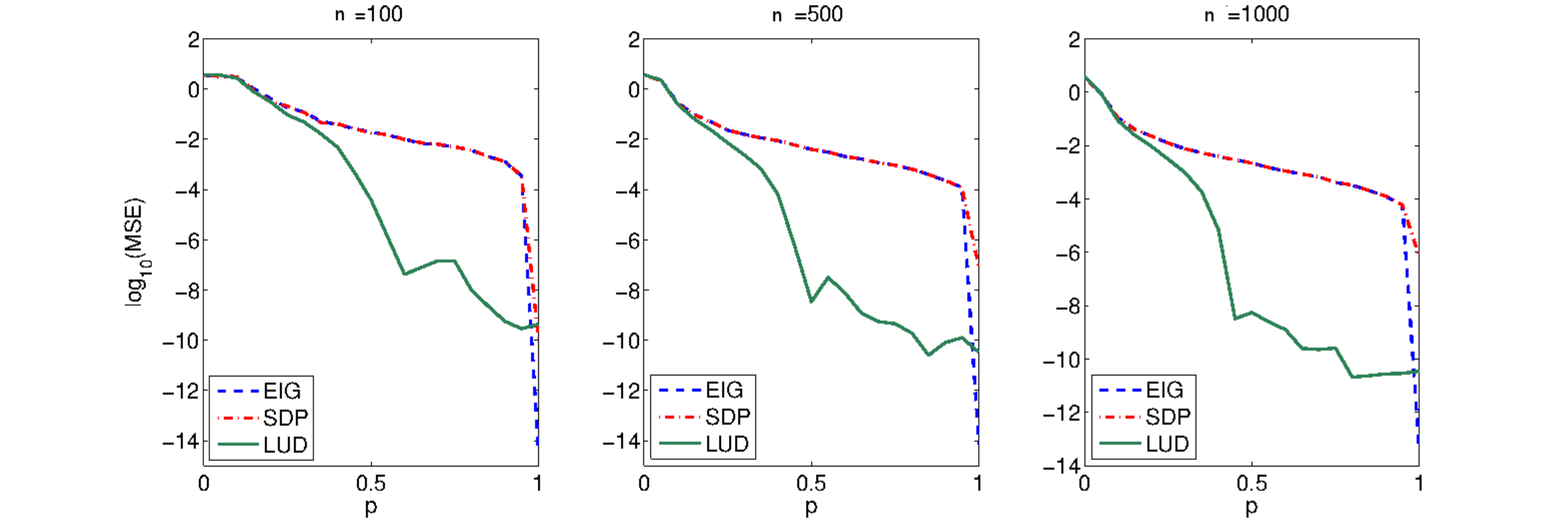}%
}
\hfill
\subfloat[\label{fig:MSE_SO3}$SO(3)$]{
\centering
\includegraphics[width=0.6\paperwidth]{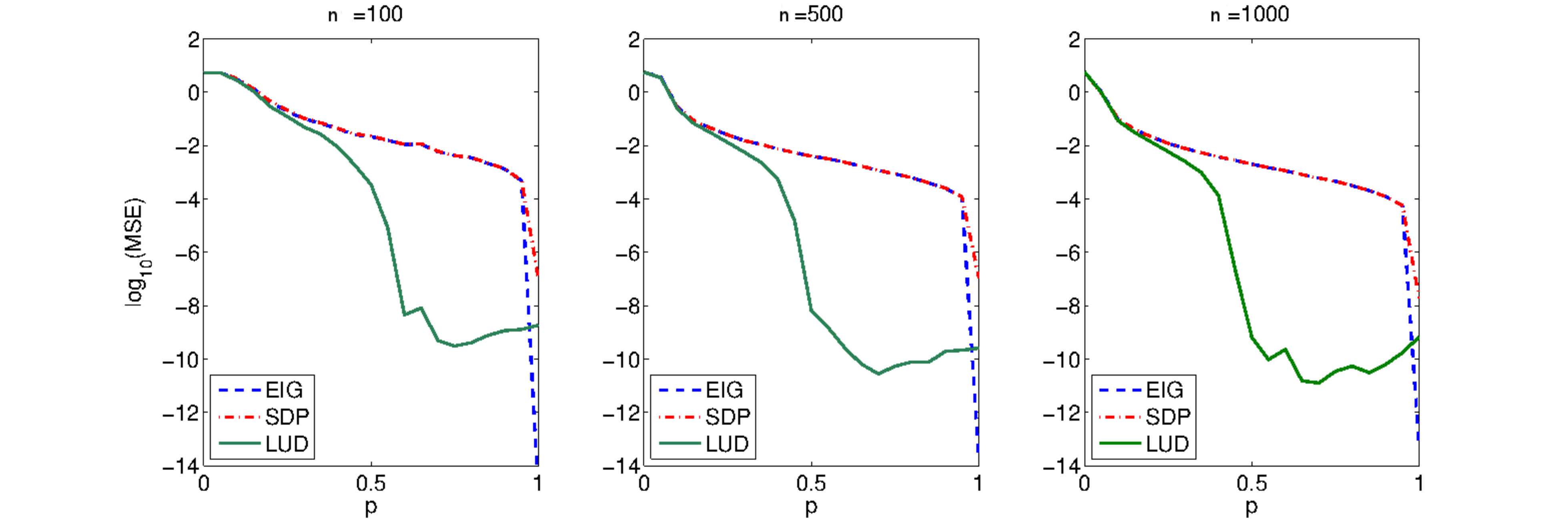}%
}
\caption[MSE (\ref{eq:MSE}) of the estimated rotations in $SO\left(3\right)$ and $SO\left(2\right)$ using EIG, SDP and LUD  for different values of $n$ and $p$ in $\bf{E1}$. Exact Recovery is achieved by LUD.]{MSE (\ref{eq:MSE}) of the estimated rotations in $SO\left(3\right)$ and $SO\left(2\right)$ using EIG, SDP and LUD  for different values of $n$ and $p$ in $\bf{E1}$. Exact Recovery (i.e. $\log_{10}(MSE) < -7$) is achieved by LUD.}\label{fig:MSE-1}
\end{figure}

\subsubsection{E2: Stability of LUD}
\label{subsec:stability_ex}
In this experiment, the three algorithms are used to recover rotations in $SO\left(2\right)$ with different values of $n$, $p$ in the noise model (\ref{eq:noise_purb}).
In (\ref{eq:noise_purb}), the perturbed rotations $\bar{R}_{ij}$ for the ``good'' edge set are sampled from a von Mises-Fisher distribution \cite{chiuso2008} with mean rotation $R_i^TR_j$ and a concentration parameter $\kappa>0$. The probability density function of the von Mises-Fisher distribution for $\bar{R}_{ij}$ is given by:
\[
f\left(\bar{R}_{ij};R_i^TR_j,\kappa\right)=c\left(\kappa\right)\exp\left(\kappa\text{Tr}\left(R_j^TR_i\bar{R}_{ij}\right)\right),
\]
where $c\left(\kappa\right)$ is a normalization constant. The parameters $R_i^TR_j$ and $1/\kappa$ are analogous to $\mu$ (mean) and $\sigma^2$  (variance) of the normal distribution:
\begin{enumerate}
\item $R_i^TR_j$ is a measure of location (the distribution is clustered around $R_i^TR_j$).
\item $\kappa$ is a measure of concentration (a reciprocal measure of dispersion, so $1/\kappa$ is analogous to $\sigma^2$). If $\kappa$ is zero, the distribution is uniform, and for small $\kappa$, it is close to uniform. If $\kappa$ is large, the distribution becomes very concentrated about the rotation $R_i^TR_j$. In fact, as $\kappa$ increases, the distribution approaches a normal distribution in $\bar{R}_{ij}$ with mean $R_i^TR_j$ and variance $1/\kappa$.
\end{enumerate}
For arbitrary fixed small $\epsilon>0$, we can choose the concentration parameter $\kappa$ large enough so that the condition (\ref{eq:condition_purb}) is satisfied. In fact, using Weyl integration formula (\ref{eq:weyl}), it can be shown that when $\kappa\rightarrow\infty$,
\begin{eqnarray*}
\epsilon=\mathbb{E}(\left\Vert \bar{R}_{ij}-R_i^TR_j\right\Vert) = \sqrt{\frac{2}{\pi \kappa}}+\mathcal{O}(\kappa^{-3/2}),\\
\text{Var}(\left\Vert \bar{R}_{ij}-R_i^TR_j\right\Vert) = (1-\frac{2}{\pi})\frac{1}{\kappa} + \mathcal{O}(\kappa^{-2}).
\end{eqnarray*}
Figures \ref{fig:MSE_SO2_n100_kappa}, \ref{fig:MSE_SO2_n500_kappa} and \ref{fig:different_kappa} show that for different $n$, $p$ and concentration parameter $\kappa$, LUD is more accurate than EIG and SDP. In Figure \ref{fig:different_kappa}, The MSEs of LUD, EIG and SDP are compared with the Cram\'er-Rao bound (CRB) for synchronization. In \cite{Boumal2012}, Boumal et. al. established formulas for the Fisher information matrix and associated CRBs that are structured by the pseudoniverse of the  Laplacian of the measurement graph.

In addition we observe that when the concentration parameter $\kappa$ is as large as $100$, which means the perturbations on the ``good'' edges are small, the critical probability $p_c$ for phase transition can be clearly identified around $0.5$. As $\kappa$ decreases, the phase transition becomes less obvious.

\begin{figure}[H]
\centering
\subfloat[\label{fig:MSE_SO2_n100_kappa}$n=100$]{
\centering
\includegraphics[width=0.5\paperwidth]{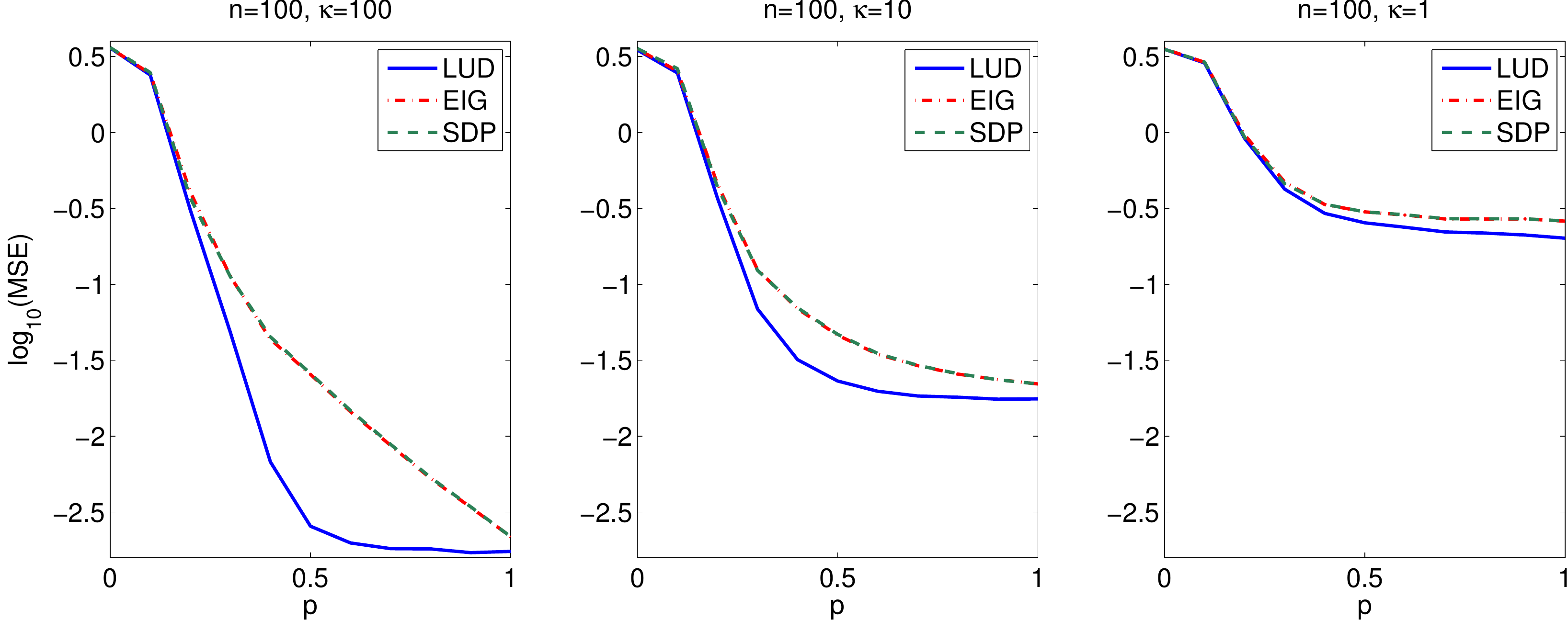}%
}
\hfill
\subfloat[\label{fig:MSE_SO2_n500_kappa}$n=500$]{
\centering
\includegraphics[width=0.5\paperwidth]{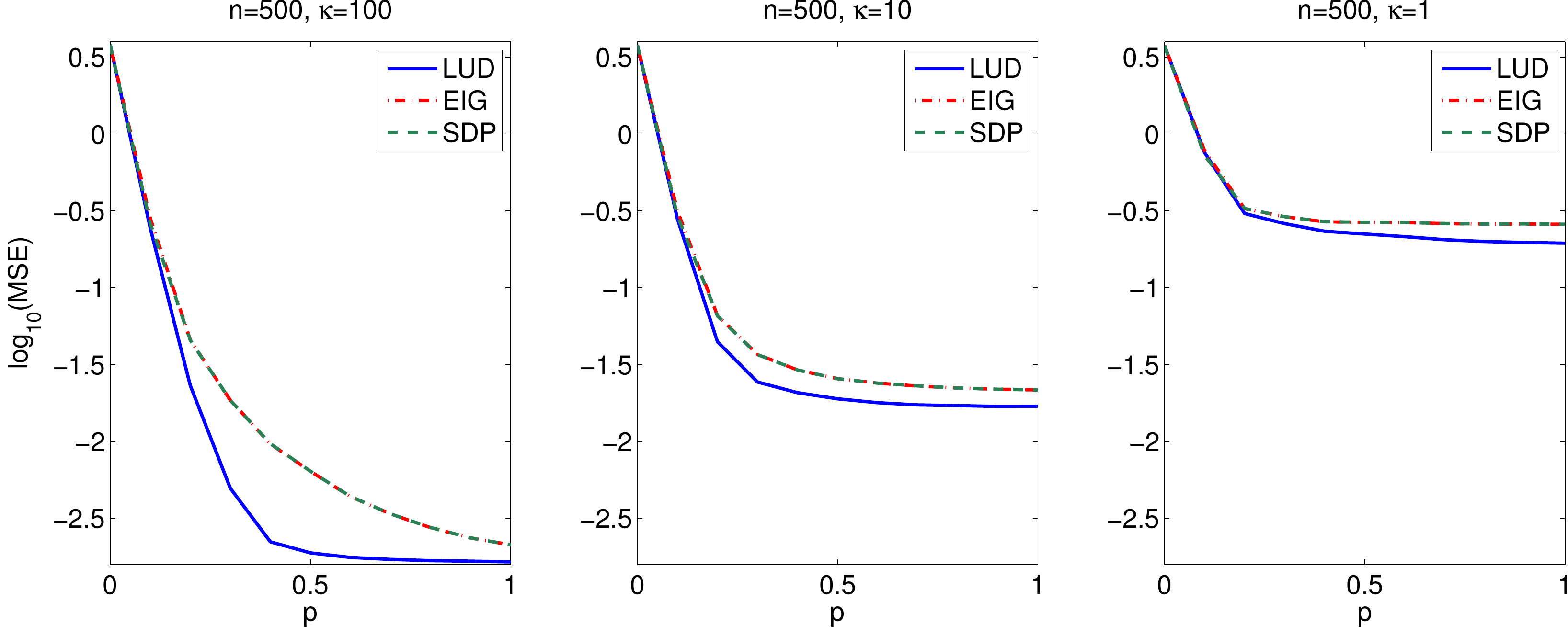}%
}
\caption{MSE (\ref{eq:MSE}) of the estimated rotations in $SO\left(2\right)$ using EIG, SDP and LUD  for $100$ and $500$ rotations and different values of $p$ and concentration parameter $\kappa$ in $\bf{E2}$. LUD is stable when there are small perturbations on the ``good'' edge set.}\label{fig:MSE-2}
\end{figure}

\begin{figure}[H]
\center
\includegraphics[width=0.5\paperwidth]{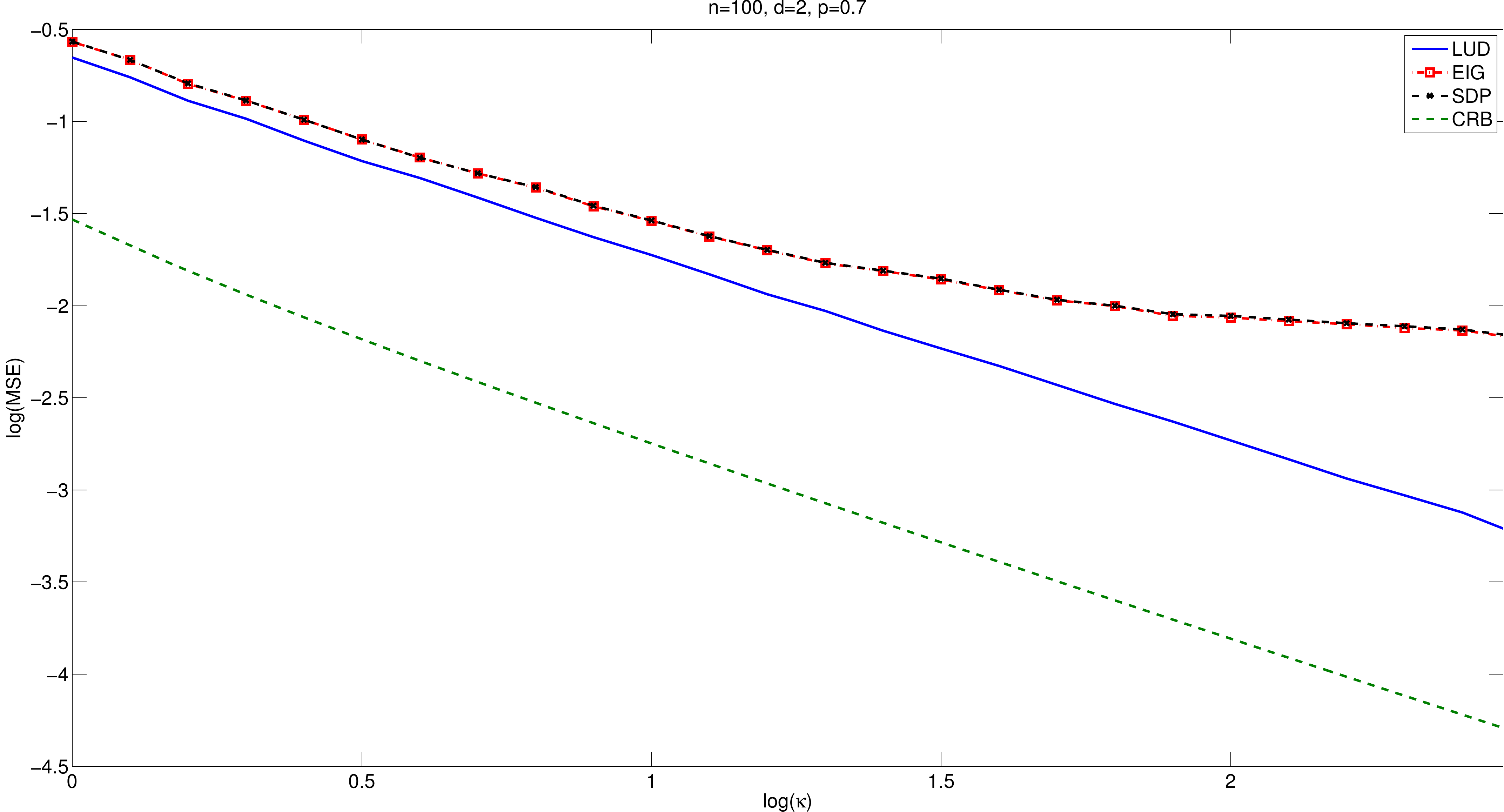}%
\caption{\label{fig:different_kappa}MSE of $n=100$ estimated rotations in $SO(2)$ as a function of $\kappa$ when $p=0.7$ using LUD, EIG and SDP. The recovery by LUD is stable to small perturbations on the ``good" edges as indicated by the linear relationship of $\log(\text{MSE})$ and $\log(\kappa)$. The green dotted line represents the Cram\'er-Rao bound for synchronization \cite{Boumal2012}.}
\end{figure}
\label{sec:exp2}

\subsection{Experiments with incomplete measurements}
\label{sec:incomplete_exp}
In the experiments E3 and E4 shown in Figure \ref{fig:incomplete} and \ref{fig:incomplete_stab}, the measurements $R_{ij}$ are generated as in experiments E1 and E2, respectively, with the exception that instead of using the complete graph of measurements as in E1 and E2,
 the index set of measurements $\mathcal{E}$ is a realization of a random graph drawn from the Erd\H{o}s-R\'enyi model $\mathcal{G}(n,p_1)$, where $p_1$ is the proportion of measured rotation ratios. The results demonstrate the exact recovery and stability of LUD with incomplete measurements that are described in Section \ref{sec:missing_entries}.
%\subsubsection{E3: Exact Recovery by LUD with incomplete measurements}
%\label{sec:exp3}
\begin{figure}[H]
\begin{center}
\includegraphics[width=0.6\paperwidth]{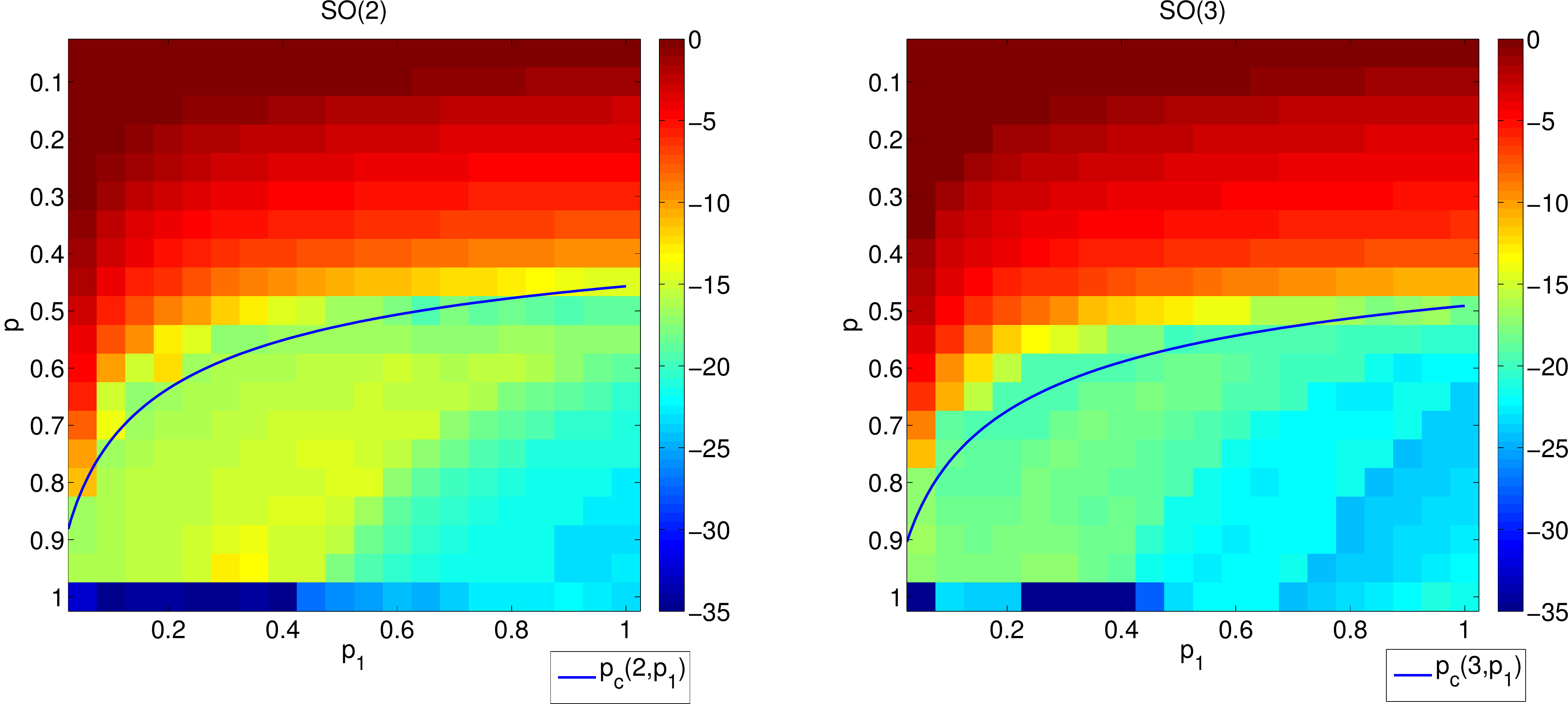}
\end{center}
\caption{Phase transitions of LUD on incomplete graphs ({\bf E3}). The color intensity of each pixel represents $\log $(MSE), depending on the edge probability $p_1$ (x-axis), and the ``good'' edge probability $p$ (y-axis). The blue curves are the upper bounds of the critical probability $p_c(d,p_1)$ in Theorem \ref{thm:p_c_missing}. Both experiments used $n=500$ rotations.}\label{fig:incomplete}
\end{figure}

%\subsubsection{E4: Stability of LUD with incomplete measurements}
%\label{sec:exp4}
\begin{figure}[H]
\begin{center}
\includegraphics[width=0.6\paperwidth]{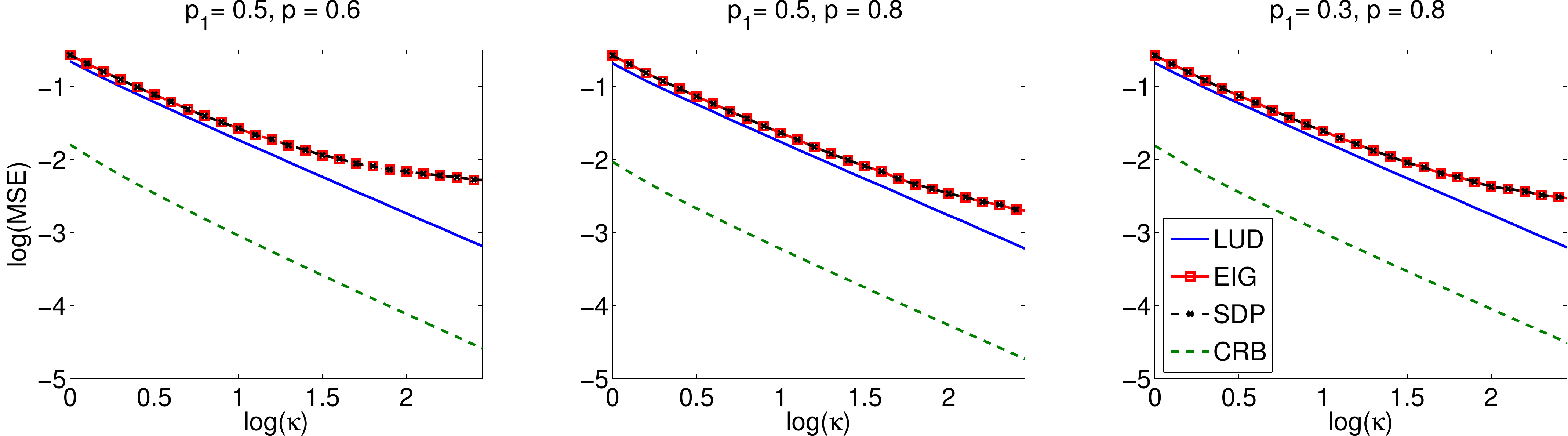}
\end{center}
\caption{MSE of $n=500$ estimated rotations in $SO(2)$ as a function of $\kappa$ with different values of the edge probability $p_1$ and the ``good'' edge probability $p$ using LUD, EIG and SDP ({\bf E4}). The recovery by LUD is stable to small perturbations on the ``good" edges as indicated by the linear relationship of $\log(\text{MSE})$ and $\log(\kappa)$. The green dotted line represents the Cram\'er-Rao bound for synchronization \cite{Boumal2012}.}\label{fig:incomplete_stab}
\end{figure}

\subsection{E5: Real data experiments}
We tried LUD for solving the global alignment problem of 3D scans from the Lucy dataset\footnote{Available from The Stanford 3-D Scanning Repository at http://www-graphics.stanford.edu/data/3Dscanrep/} (see Figure \ref{fig:lucy}). We are using a down-sampled version of the dataset containing 368 scans with a total number of 3.5 million triangles.
Running the automatic Iterative Closest Point (ICP) algorithm \cite{ICP} starting from initial estimates returned 2006 pairwise transformations.
For this model, we only have the best reconstruction found so far at our disposal but no actual ground truth. Nevertheless, we use this reconstruction to evaluate the error of the estimated rotations.

We apply the two algorithms LUD, EIG on the Lucy dataset since we observed SDP did not perform so well on this dataset. Although the MSEs are quite similar (0.4044 for LUD and 0.3938 for EIG), we observe that the unsquared residuals $\left\Vert\hat{R}_i-R_i\right\Vert$ ($i=1,2,\ldots,368$), where $\hat{R}_i$ is the estimated rotation, are more concentrated around zero for LUD (Figure \ref{fig:hist}). Figure \ref{fig:error_ratio} suggests that the ``bad'' edges $(i,j)$ (edges with truly large measurement errors in the left subfigure of Figure \ref{fig:error_ratio}) can be eliminated using the results of LUD more robustly, compared to that of EIG. We set the cutoff value to be $0.1$ in Figure \ref{fig:error_ratio} for the estimated measurement errors obtained by LUD and EIG. Then 1527  and 1040 edges are retained from 2006 edges by LUD and EIG respectively, and the largest connected component of the sparsified graph (after eliminating the seemingly ``bad" edges) has size 312 and 299 respectively. The 3D scans with the estimated rotations in the largest component are used in the reconstruction. The reconstruction obtained by LUD is better than that by EIG (Figure \ref{fig:lucy_reconstruction}).

\begin{figure}[H]
\begin{center}
\includegraphics[width=0.15\textwidth]{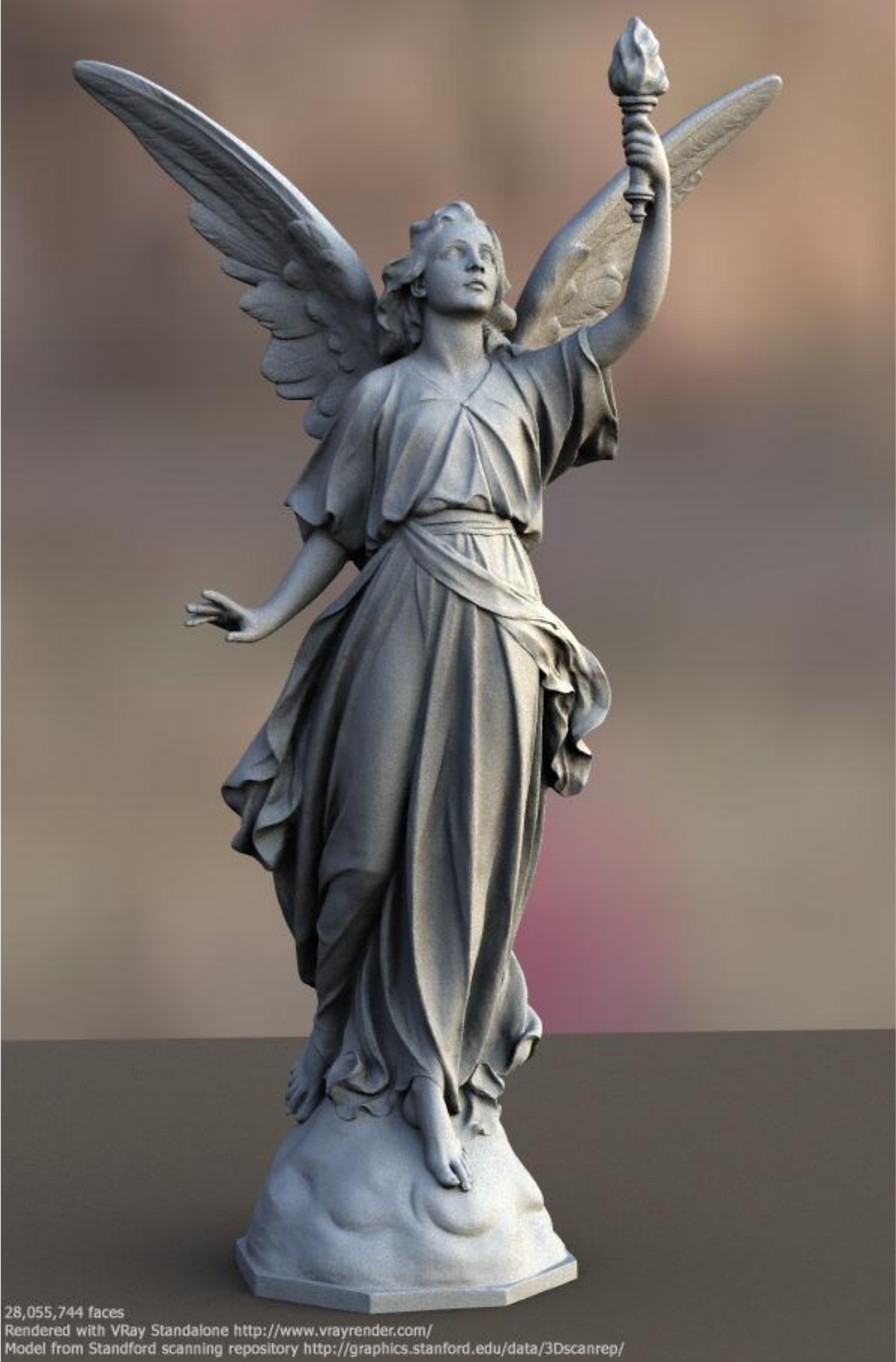}
\end{center}
\caption{The Lucy statue.}\label{fig:lucy}
\end{figure}

\begin{figure}[h]
\begin{center}
\includegraphics[width=0.45\paperwidth]{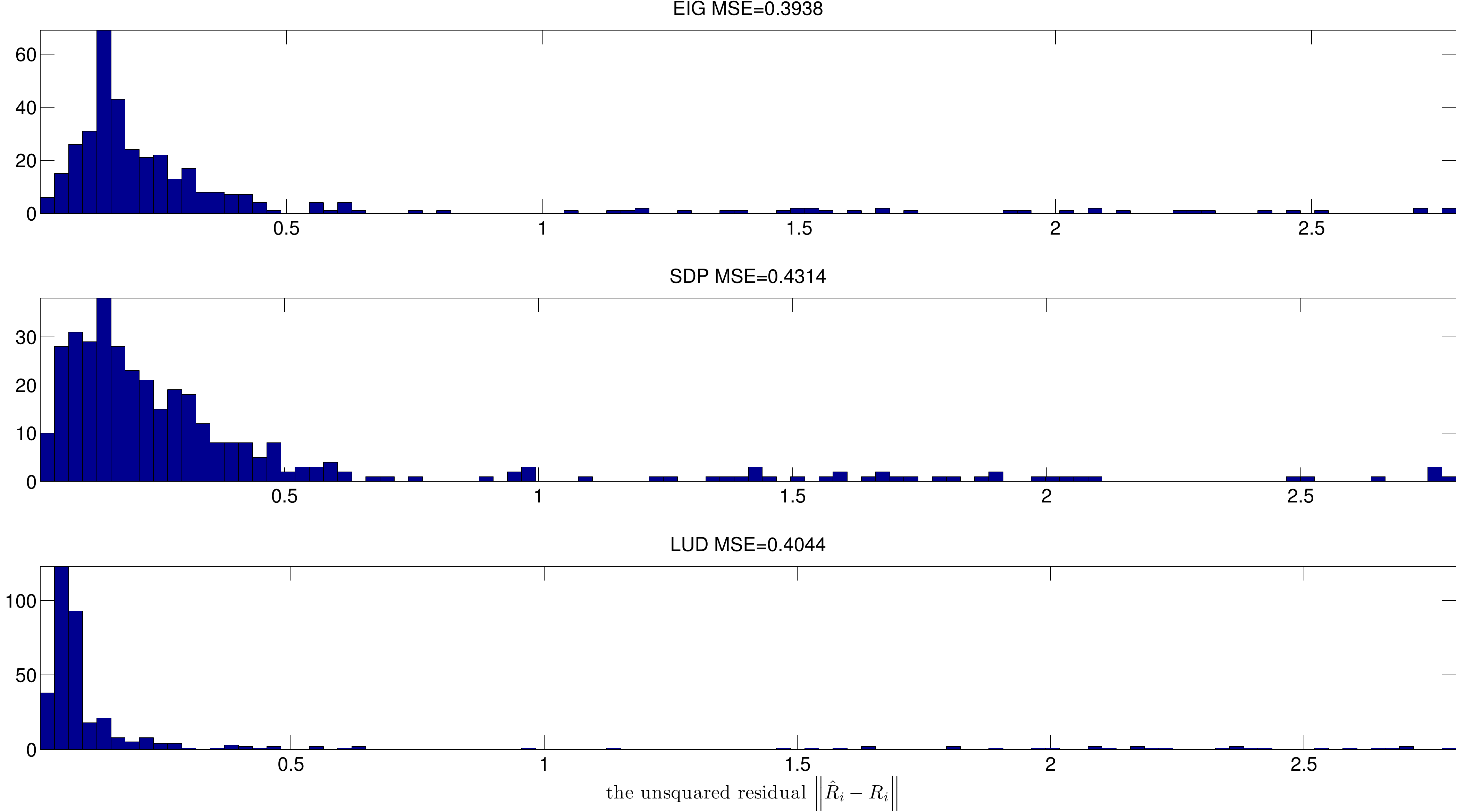}
\end{center}
\caption{Histogram of the unsquared residuals of EIG, SDP and LUD for the Lucy dataset. The errors by LUD are more concentrated near zero.}\label{fig:hist}
\end{figure}

\begin{figure}[H]
\begin{center}
\includegraphics[width=0.45\paperwidth]{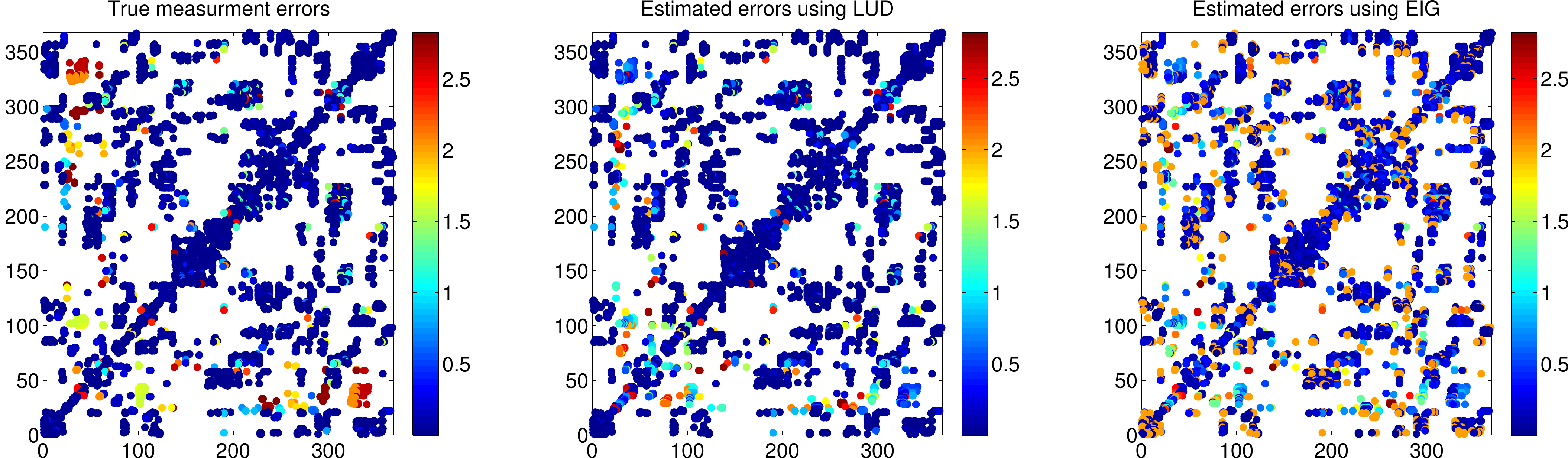}
\end{center}
\caption[True and estimated rotation ratio measurement errors for the Lucy dataset.]{True and estimated rotation ratio measurement errors for the Lucy dataset. The measurements are sparse since only 2006 rotation ratios were measured for 368 3-D scans. Left: the value at the $(i,j)$th element is $\left\Vert R_{ij}-R_i^TR_j\right\Vert$ if it is non-empty, where $i,j=1,2,\ldots,368$. Middle and Right: similar to Left except that the value at the $(i,j)$th element is $\left\Vert R_{ij}- \hat{R}_i^T\hat{R}_j\right\Vert$ if it is non-empty.}\label{fig:error_ratio}
\end{figure}

\begin{figure}[H]
\centering
\includegraphics[width=0.6\paperwidth]{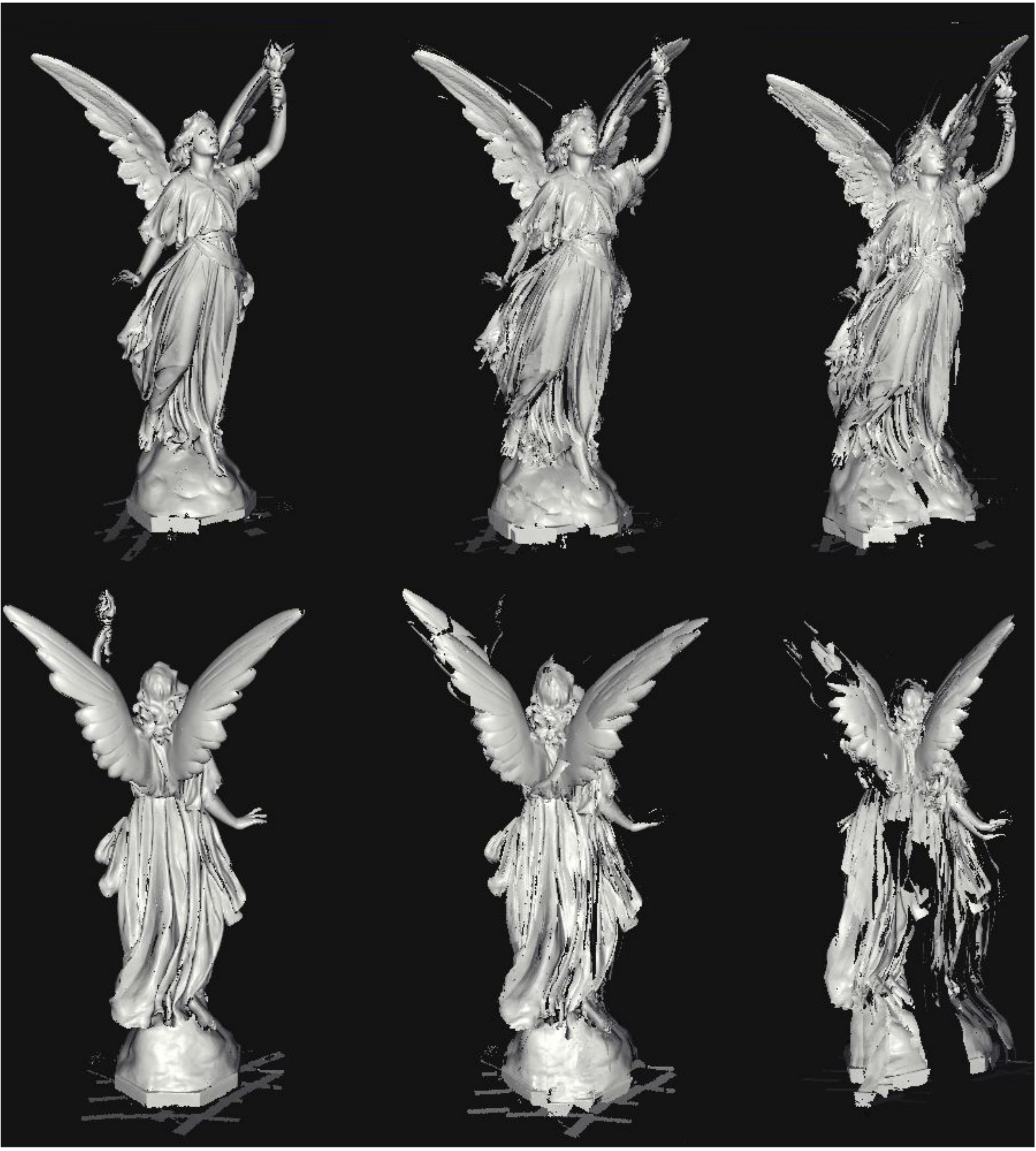}
\caption[Front and back views of the reconstructions from the Lucy dataset using LUD and EIG.]{Front and back views of the reconstructions from the Lucy dataset. The left one is the best found reconstruction. The middle and right ones are obtained by LUD and EIG.}\label{fig:lucy_reconstruction}
\end{figure}

\section{Summary}
\label{sec:summary}

In this paper we proposed to estimate the rotations using LUD. LUD minimizes a robust self consistency error, which
is the sum of unsquared residuals instead of the sum of squared residuals. LUD is then semidefinite relaxed and solved by ADM. We compare LUD method to EIG and SDP methods, both of which are based on least squares approach, and demonstrate that the results obtained by LUD are the most accurate. When the noise in the rotation ratio measurements comes from ER random graph model $G\left(n,p\right)$, we compute an upper bound $p_c$ of the phase transition point such that the rotations can be exactly recovered when $p > p_c$. Moreover, the solution of LUD is stable when small perturbations are added to ``good" rotation ratio measurements. We also showed exact recovery and stability for LUD when the measurements of the rotation ratios are incomplete and the measured rotation ratios come from ER random graph model $G\left(n,p_1\right)$.

The exact recovery result for the noise model (\ref{eq:noise}) is actually not that surprising. In order to determine if the rotation measurement for a given edge $(i,j)$ is correct or noisy, we can consider all cycles of length three (triangles) that include that edge. There are $n-2$ such triangles, and we can search for a consistent triangle by multiplying the three rotation ratios and checking if the product is the identity rotation. If there is such a consistent triangle, then the measurement for the edge $(i,j)$ is almost surely correct. The expected number of such consistent triangles is $(n-2)p^2$, so the critical probability for such an algorithm is $p_c = \mathcal{O}(1/\sqrt{n})$, which is already significantly better than our exact recovery condition for LUD for which the critical probability does not depend on $n$. In fact, by considering all cycles of length $k \geq 3$ consisting of a given fixed edge (there are $\mathcal{O}(n^{k-2})$ such cycles), the expected number of consistent cycles is $\mathcal{O}(n^{k-2}p^{k-1})$, therefore the critical probability is $O(1/n^{1-\varepsilon})$ for any $\varepsilon >0$. Since exact recovery requires the graph of good edges is connected, and the ER graph is connected almost surely when $p=\frac{2\log n}{n}$, the critical probability cannot be smaller than $p=\frac{2\log n}{n}$. The computational complexity of cycle-based algorithms increases exponentially with the cycle length, but already for short cycles (e.g., $k=3$) they are preferable to LUD in terms of the critical probability. So what is the advantage of LUD compared to cycle-based algorithms? The answer to this question lies in our stability result. While LUD is stable to small perturbations on the good edges, cycle-based algorithms are unstable, and are therefore not as useful in applications where small measurement noise is present.

An iterative algorithm for robust estimation of rotations has been recently proposed in \cite{Hartley2011} for applications in computer vision. That algorithm aims to minimize the sum of geodesic distances between the measured rotation ratios and those derived from the estimated rotations. In each iteration, the algorithm sequentially updates the rotations by the median of their neighboring rotations using the Weiszfeld algorithm. We tested this algorithm numerically and find it to perform well, in the sense that its critical probability for exact recovery for the noise model (\ref{eq:noise}) is typically smaller than that of LUD. However, the objective function that algorithm aims to minimize is not even locally convex. It therefore requires a good initial guess which can be provided by either EIG, SDP or LUD. In a sense, that algorithm may be regarded as a non-convex refinement algorithm for the estimated rotations. There is currently no theoretical analysis of that algorithm due to the non-convexity and the non-smoothness of its underlying objective function.

In the future, we plan to extend the LUD framework in at least two ways. First, we plan to investigate exact and stable recovery conditions for more general (and possibly weighted) measurement graphs other than the complete and random ER graphs considered here. We speculate that the second eigenvalue of the graph Laplacian would play a role in conditions for exact and stable recovery, similar to the role it plays in the bounds for synchronization obtained in \cite{Bandeira2012}. Also in \cite{Demanet2013}, for synchronization-like problems over SO(2), the error bounds are given in terms of the second eigenvalue of the graph Laplacian and the graph connection Laplacian of the measurements, which hold for general graphs. The solutions are obtained using a SDR, which is very similar to our LUD formulation. The only difference is that the sum of unsquared divinations is not used as a cost function, instead, it is used as a constraint, which yields a feasibility problem. Second, the problem of synchronization of rotations is closely related to the orthogonal Procrustes problem of rotating $n$ matrices toward a best least-squares fit \cite{TenBerge1977}. Closed form solution is available only for $n=2$, and a certain SDP relaxation for $n \geq 3$ was analyzed in \cite{Nemirovski2007}. The LUD framework presented here can be extended in a straightforward manner for the orthogonal Procrustes problem of rotating $n$ matrices toward a best least-unsquared deviations fit, which could be referred to as the robust orthogonal Procrustes problem. Similar to the measurement graph assumed in our theoretical analysis of LUD for robust synchronization, also in the Procrustes problem the measurement graph is the complete graph (all pairs of matrices are compared).

\section*{Acknowledgements}
This work was supported by the National Science Foundation [DMS-0914892]; the Air Force Office of Scientific Research [FA9550-09-1-0551]; the National Institute of General Medical Sciences
[R01GM090200]; the Simons Foundation [LTR DTD 06-05-2012]; and the Alfred P. Sloan Foundation. The authors thank Zaiwen Wen for many useful discussions regarding ADM. They also thank Szymon Rusinkiewicz and Tzvetelina Tzeneva for invaluable information and assistance with the Lucy dataset.

\appendix

\section{The value of $c\left(d\right)$}
\label{sec:cd}
Now let us study the constant $c(d)$ such that $c(d)I_{d}=\mathbb{E}\left(\frac{I_{d}-R}{\left\Vert I_{d}-R\right\Vert }\right)$,
where $R$ is uniformly sampled from the rotation group $SO(d)$.
Due to the symmetry of the uniform distribution over the rotation
group $SO(d)$, the off-diagonal entries of $\mathbb{E}\left(\frac{I_{d}-R}{\left\Vert I_{d}-R\right\Vert }\right)$
are all zeros, and the diagonal entries of $\mathbb{E}\left(\frac{I_{d}-R}{\left\Vert I_{d}-R\right\Vert }\right)$
are all the same as the constant $c(d)$. Using the fact that
\begin{eqnarray*}
\left\Vert I_{d}-R\right\Vert  & = & \sqrt{\text{Tr}\left(\left(I_{d}-R\right)\left(I_{d}-R\right)^{T}\right)}\\
 & = & \sqrt{\text{Tr}\left(2I_{d}-R-R^{T}\right)}\\
 & = & \sqrt{2\text{Tr}\left(I_{d}-R\right)},
\end{eqnarray*}
we have
\[
c(d)  =  \frac{1}{d}\mathbb{E}\left(\frac{\text{Tr}\left(I_{d}-R\right)}{\left\Vert I_{d}-R\right\Vert }\right)
= \frac{1}{d}\mathbb{E}\left(\frac{\text{Tr}\left(I_{d}-R\right)}{\sqrt{2\text{Tr}\left(I_{d}-R\right)}}\right)
 =  \frac{1}{\sqrt{2}d}\mathbb{E}\left(\sqrt{\text{Tr}\left(I_{d}-R\right)}\right),
\]
where the probability measure for expectation is the Haar measure.
The function $f\left(R\right)=\sqrt{\text{Tr}\left(I_{d}-R\right)}$
is a class function, that is, $f\left(R\right)$ is invariant under
conjugation, meaning that for all $O,R\in SO(d)$ we have $f\left(ORO^{T}\right)=f\left(R\right)$.
Therefore $\mathbb{E}\left(f\left(R\right)\right)$ can be computed
using the Weyl integration formula specialized to $SO(d)$ (Exercise 18.1\textendash{}2 in \cite{bump2004lie}) as below:
\begin{eqnarray}
\int_{SO\left(2m+1\right)}f\left(R\right)\text{d}R & = & \frac{1}{2^{m}m!}\int_{[-\pi,\pi]^{m}}f\left(\left(\begin{array}{cccccc}
\cos\theta_{1} & -\sin\theta_{1}\\
\sin\theta_{1} & \cos\theta_{1}\\
 &  & \ddots\\
 &  &  & \cos\theta_{m} & -\sin\theta_{m}\\
 &  &  & \sin\theta_{m} & \cos\theta_{m}\\
 &  &  &  &  & 1
\end{array}\right)\right)\nonumber\\
 &  &\times \prod_{i<j}\left(\left|e^{\imath\theta_{i}}-e^{\imath\theta_{j}}\right|^{2}\left|e^{\imath\theta_{i}}-e^{-\imath\theta_{j}}\right|^{2}\right)\prod_{i}\left|e^{\imath\theta_{i}}-1\right|^{2}\text{d}\theta_{1}\ldots\text{d}\theta_{m},\nonumber\\
\\
\int_{SO\left(2m\right)}f\left(R\right)\text{d}R & = & \frac{1}{2^{m-1}m!}\int_{[-\pi,\pi]^{m}}f\left(\left(\begin{array}{ccccc}
\cos\theta_{1} & -\sin\theta_{1}\\
\sin\theta_{1} & \cos\theta_{1}\\
 &  & \ddots\\
 &  &  & \cos\theta_{m} & -\sin\theta_{m}\\
 &  &  & \sin\theta_{m} & \cos\theta_{m}
\end{array}\right)\right)\nonumber\\
 &  &\times \prod_{i<j}\left(\left|e^{\imath\theta_{i}}-e^{\imath\theta_{j}}\right|^{2}\left|e^{\imath\theta_{i}}-e^{-\imath\theta_{j}}\right|^{2}\right)\text{d}\theta_{1}\ldots\text{d}\theta_{m}.\label{eq:weyl}
\end{eqnarray}
In particular, for $d$ = $2$ or $3$ ($m=1)$, using the Weyl integration
formula we obtain
\begin{eqnarray*}
c(2) & = & \frac{1}{2\sqrt{2}}\mathbb{E}\left(\sqrt{\text{Tr}\left(I_{2}-R\right)}\right)\\
 & = & \frac{1}{2\sqrt{2}}\cdot\frac{1}{2\pi}\int_{-\pi}^{\pi}\sqrt{\text{Tr}\left(I_{2}-\left(\begin{array}{cc}
\cos\theta & -\sin\theta\\
\sin\theta & \cos\theta
\end{array}\right)\right)}\text{d}\theta\\
 & = & \frac{1}{4\pi}\int_{-\pi}^{\pi}\sqrt{1-\cos\theta}\text{d}\theta\\
 & = & \frac{1}{2\pi}\int_{0}^{\pi}\sqrt{1-\cos\theta}\text{d}\theta\text{ (let }x=\cos\theta\text{)}\\
 & = & \frac{1}{2\pi}\int_{-1}^{1}\frac{1}{\sqrt{1+x}}\text{d}x\\
 & = & \sqrt{2}/\pi,
\end{eqnarray*}
and
\begin{eqnarray*}
c(3) & = & \frac{1}{3\sqrt{2}}\mathbb{E}\left(\sqrt{\text{Tr}\left(I_{3}-R\right)}\right)\\
 & = & \frac{1}{3\sqrt{2}}\cdot\frac{1}{2\pi}\int_{-\pi}^{\pi}\sqrt{\text{Tr}\left(I_{3}-\left(\begin{array}{ccc}
\cos\theta & -\sin\theta&0\\
\sin\theta & \cos\theta&0\\
0 &0  & 1
\end{array}\right)\right)}\left(1-\cos\theta\right)\text{d}\theta\\
 & = & \frac{1}{3\pi}\int_{-\pi}^{\pi}\sqrt{1-\cos\theta}\left(1-\cos\theta\right)\text{d}\theta\\
 & = & \frac{2\sqrt{2}}{3\pi}\int_{0}^{\pi}\sin\frac{\theta}{2}\left(1-\cos^{2}\frac{\theta}{2}\right)\text{d}\theta\text{ (let }x=\cos\frac{\theta}{2}\text{)}\\
 & = & \frac{4\sqrt{2}}{3\pi}\int_{0}^{1}\left(1-x^{2}\right)\text{d}x\\
 & = & \frac{8\sqrt{2}}{9\pi}.
\end{eqnarray*}
For $d\geq4$, the computation of $c\left(d\right)$ involves complicated
multiple integrals. Now we study the lower and upper bound and the
limiting value of $c\left(d\right)$ for large $d$.

\begin{lemma}
 \label{lem:bounds_cd}
$\frac{1}{2\sqrt{2\left\lfloor d/2\right\rfloor }}\leq c\left(d\right)\leq\frac{1}{\sqrt{2d}},$
where $\left\lfloor \cdot\right\rfloor $ denotes the floor of a number.
\end{lemma}
\begin{proof}
 Using the fact that the square root function is concave, we
obtain
\begin{eqnarray*}
c\left(d\right) & = & \frac{1}{\sqrt{2}d}\mathbb{E}\left(\sqrt{\text{Tr}\left(I_{d}-R\right)}\right)\\
 & \leq & \frac{1}{\sqrt{2}d}\sqrt{\mathbb{E}\left(\text{Tr}\left(I_{d}-R\right)\right)}\\
 & = & \frac{\sqrt{d}}{\sqrt{2}d}=\frac{1}{\sqrt{2d}},
\end{eqnarray*}
where the second equality uses the fact that $\text{Tr}\left(R\right)=0$
due to the symmetry of the Haar measure.

Since $d-4\left\lfloor d/2\right\rfloor \leq\text{Tr}\left(R\right)\leq d$,
we have
\begin{equation}
0\leq\text{Tr}\left(I_{d}-R\right)\leq4\left\lfloor d/2\right\rfloor .\label{eq:range_tr}
\end{equation}
In the range (\ref{eq:range_tr}), $\sqrt{\text{Tr}\left(I_{d}-R\right)}\geq\frac{1}{2\sqrt{\left\lfloor d/2\right\rfloor }}\text{Tr}\left(I_{d}-R\right)$
due to the concavity of the square root function. Therefore we obtain
\begin{eqnarray*}
c\left(d\right) & = & \frac{1}{\sqrt{2}d}\mathbb{E}\left(\sqrt{\text{Tr}\left(I_{d}-R\right)}\right)\\
 & \geq & \frac{1}{2d\sqrt{\left\lfloor d/2\right\rfloor }}\mathbb{E}\left(\text{Tr}\left(I_{d}-R\right)\right)\\
 & = & \frac{d}{2d\sqrt{\left\lfloor d/2\right\rfloor }}=\frac{1}{2\sqrt{\left\lfloor d/2\right\rfloor }}.
\end{eqnarray*}
\end{proof}

Remark. The upper bound of $c\left(d\right)$ is very close to $c\left(d\right)$
for $d=2,$ $3$, and $4$:
\begin{eqnarray*}
\frac{1}{2\sqrt{2}}\left(\approx0.3536\right)\leq & c(2)=\frac{\sqrt{2}}{\pi}\left(\approx0.4502\right) & \leq\frac{1}{2},\\
\frac{1}{2\sqrt{2}}\left(\approx0.3536\right)\leq & c(3)=\frac{8\sqrt{2}}{9\pi}\left(\approx0.4001\right) & \leq\left(\frac{1}{\sqrt{6}}\approx0.4082\right),\\
0.25=\frac{1}{4}\leq & c(4)\left(\approx0.3505\right) & \leq\frac{1}{2\sqrt{2}}\left(\approx0.3536\right).
\end{eqnarray*}
In fact we can prove the following lemma.

\begin{lemma}
\label{lem:limiting_cd}
$\lim_{d\rightarrow\infty}\sqrt{2d}c\left(d\right)=1$.
\end{lemma}
\begin{proof}
 We have
\begin{eqnarray}
1\geq\sqrt{2d} & c\left(d\right)= & \frac{1}{\sqrt{d}}\mathbb{E}\left(\sqrt{\text{Tr}\left(I_{d}-R\right)}\right)\nonumber \\
 & \geq & \frac{1}{\sqrt{d}}\mathbb{P}\left(\left|\text{Tr}\left(R\right)\right|\leq d^{\frac{1}{4}}\right)\sqrt{d-d^{\frac{1}{4}}},\label{eq:cd_limit_1}
\end{eqnarray}
where the first inequality follows Lemma \ref{lem:bounds_cd}. Diaconis and Mallows \cite{trace_rm} showed that the moments of the trace of $R$ equal the moments
of a standard normal variable for all sufficiently large n. In particular,
when $d\rightarrow\infty$, the limit of $\text{Tr}\left(R\right)$
has mean $0$ and variance $1$. Therefore using Chebyshev inequality
we get that
\begin{equation}
\mathbb{P}\left(\left|\text{Tr}\left(R\right)\right|\leq d^{\frac{1}{4}}\right)\rightarrow1\text{ as }d\rightarrow\infty.\label{eq:cd_limit_2}
\end{equation}
Hence let $d\rightarrow\infty$ in (\ref{eq:cd_limit_1}), we obtain
$\lim_{d\rightarrow\infty}\sqrt{2d}c\left(d\right)=1$.
\end{proof}

\section{Proof of Lemma \ref{lem:D2parts}}
\label{sec:rm}
Here we aim to prove that the limiting spectral density of the normalized
random matrix $\frac{1}{\sqrt{n-1}}D$ is Wigner's semi-circle, where
$D$ is defined in (\ref{eq:D}). The following definitions and results
will be used.
\begin{itemize}
\item The Cauchy (Stieltjes) transform of a probability measure $\mu$ and
moments $\left\{ m_{k}\right\} _{k=1}^{\infty}$ is given by
\[
G_{\mu}\left(z\right)=\int_{\mathbb{R}}\frac{\text{d}\mu\left(x\right)}{z-x}=\frac{1}{z}+\sum_{k=1}^{\infty}\frac{m_{k}}{z^{k+1}},\text{ }z\in\mathbb{C},\text{ }{\bf Im}\left(z\right)>0.
\]

\item The density $g\left(x\right)$ can be recovered from $G_{\mu}$ by
the Stieltjes inversion formula:
\[
g\left(x\right)=-\frac{1}{\pi}\lim_{\epsilon\rightarrow0}{\bf Im}\left(G_{\mu}\left(x+\imath\epsilon\right)\right).
\]

\item The Wigner semicircle distribution $\mu_{0,\sigma^{2}}$centered at
$0$ with variance $\sigma^{2}$ is the distribution with density
\[
g\left(x\right)=\begin{cases}
\frac{\sqrt{4\sigma^{2}-x^{2}}}{2\pi\sigma^{2}} & \text{if }x\in\left[-2\sigma,2\sigma\right],\\
0 & \text{otherwise.}
\end{cases}
\]

\item The Cauchy transform of the semicircle distribution $\mu_{0,\sigma^{2}}$is
given by
\[
G_{\mu_{0},\sigma^{2}}\left(z\right)=\frac{z-\sqrt{z^{2}-4\sigma^{2}}}{2\sigma^{2}}.
\]

\end{itemize}
We now state a theorem by Girko, which extends the semicircle law to
the case of random block matrices, and show how, in particular, this
follows for the random matrix$\frac{1}{\sqrt{n}}D$.
\begin{theorem}
\label{semi_circle}
(Girko, 1993 in \cite{Girko}) Suppose the $nd\times nd$ matrix $M$ is composed of
$d\times d$ independent random blocks $M_{ij}$, $i,j=1,2,\ldots,n$
which satisfy
\begin{equation}
M_{ij}=M_{ji}^{T},\:\mathbb{E}\left(M_{ij}\right)=O_{d},\:\mathbb{E}\left\Vert M_{ij}\right\Vert ^{2}<\infty.\label{eq:cond_1}
\end{equation}
Suppose also that
\begin{equation}
\sup_{n}\max_{i=1,2,\ldots,n}\sum_{j=1}^{n}\mathbb{E}\left\Vert M_{ij}\right\Vert ^{2}<\infty\label{eq:cond_2}
\end{equation}
and that Lindeberg's condition holds: for every $\epsilon>0$
\begin{equation}
\lim_{n\rightarrow\infty}\max_{i=1,2,\ldots,n}\sum_{j=1}^{n}\mathbb{E}\left(\left\Vert M_{ij}\right\Vert ^{2}\mathcal{X}_{\left(\left\Vert M_{ij}\right\Vert >\epsilon\right)}\right)=0,\label{eq:cond_3}
\end{equation}
where $\mathcal{X}$ denotes an indicator function. Let $\lambda_{1}\geq\ldots\geq\lambda_{nd}$
be the eigenvalues of $M$ and let $\mu_{n}\left(x\right)=\frac{1}{nd}\sum_{l=1}^{nd}\mathcal{X}_{\left(\lambda_{l}<x\right)}$
be the eigenvalue measure. Then for almost any $x$,
\[
\left|\mu_{n}\left(x\right)-F_{n}\left(x\right)\right|\rightarrow0\text{ as }n\rightarrow\infty\text{ a.s.},
\]
where $F_{n}\left(x\right)$ are the distribution functions whose
Cauchy/Stieltjes transforms are given by $\frac{1}{nd}\sum_{i=1}^{n}\text{Tr}\left(C_{i}\left(z\right)\right)$,
where the matrices $C_{i}\left(z\right)$satisfy
\begin{equation}
C_{i}\left(z\right)=\left(\mathbb{E}\left(zI_{d}-\mathbb{E}\sum_{j=1}^{n}M_{ij}C_{j}\left(z\right)M_{ij}^{T}\right)\right)^{-1},\text{ }i=1,2,\ldots,n.\label{eq:C_i}
\end{equation}
The solution $C_{i}\left(z\right)$, $i=1,2,\ldots,n$ to (\ref{eq:C_i})
exists and is unique in the class of analytic $d\times d$ matrix
functions $C\left(z\right)$ for which ${\bf Im}\left(z\right){\bf Im}\left(C\left(z\right)\right)<0$,
${\bf Im}\left(z\right)\neq0$.
\end{theorem}

Now we apply the result of the theorem on the matrix $W=\frac{1}{\sqrt{n-1}}D$
whose $d\times d$ blocks satisfy
\[
W_{ij}=\begin{cases}
\frac{1}{\sqrt{n-1}}\left(\frac{I_{d}-R_{ij}}{\left\Vert I_{d}-R_{ij}\right\Vert }-c\left(d\right)I_{d}\right) & \text{w.p. }1-p,\\
0 & \text{w.p. }p,
\end{cases}
\]
for $i\neq j$, where $R_{ij}=R_{ji}^{T}.$ Notice that $W$ is a random symmetric matrix ($W_{ij}=W_{ji}^{T}$)
with i.i.d off-diagonal blocks. From the definition of $c\left(d\right)$
it follows that $\mathbb{E}\left(W_{ij}\right)=0$. Also, $\left\Vert W_{ij}\right\Vert $
is finitely bounded, that is,
\begin{eqnarray*}
\left\Vert W_{ij}\right\Vert ^{2} & = & \frac{1}{n-1}\left(1+c\left(d\right)^{2}d-\frac{c\left(d\right)}{\sqrt{2}}\sqrt{\text{Tr}\left(I_{d}-R_{ij}\right)}\right)\\
 & \leq & \frac{1}{n-1}\left(1+c\left(d\right)^{2}d\right).
\end{eqnarray*}
Therefore, all the conditions (condition (\ref{eq:cond_1})-(\ref{eq:cond_3}))
of the theorem are satisfied by the matrix $W$. Before we continue,
we need the following

\begin{lemma}
\label{lem:lem_appendix_b}
We have $\mathbb{E}\left(\frac{I_{d}-R}{\left\Vert I_{d}-R\right\Vert ^{2}}\right)=\frac{1}{2d}I_{d}$,
and $\mathbb{E}\left(\frac{I_{d}-R}{\left\Vert I_{d}-R\right\Vert }\right)=c\left(d\right)I_{d}$.
\end{lemma}

\begin{proof}
Using the symmetry of the Haar measure on $SO(d)$, we can
assume that $\mathbb{E}\left(\frac{I_{d}-R}{\left\Vert I_{d}-R\right\Vert ^{2}}\right)=aI_{d}$,
and $\mathbb{E}\left(\frac{I_{d}-R}{\left\Vert I_{d}-R\right\Vert }\right)=bI_{d}$,
where $a$ and $b$ are constants depending on $d$. Therefore we
have
\begin{eqnarray*}
a & = & \frac{1}{d}\mathbb{E}\left(\text{Tr}\left(\frac{I_{d}-R}{\left\Vert I_{d}-R\right\Vert ^{2}}\right)\right)\\
 & = & \frac{1}{d}\mathbb{E}\left(\frac{\text{Tr}\left(I_{d}-R\right)}{\text{Tr}\left(\left(I_{d}-R\right)\left(I_{d}-R\right)^{T}\right)}\right)\\
 & = & \frac{1}{d}\mathbb{E}\left(\frac{\text{Tr}\left(I_{d}-R\right)}{2\text{Tr}\left(I_{d}-R\right)}\right)\\
 & = & \frac{1}{2d},
\end{eqnarray*}
and
\begin{eqnarray*}
b & = & \frac{1}{d}\mathbb{E}\left(\text{Tr}\left(\frac{I_{d}-R}{\left\Vert I_{d}-R\right\Vert }\right)\right)=c\left(d\right).
\end{eqnarray*}
\end{proof}

We claim that $C_{i}\left(z\right)=h\left(z\right)I_{d}$, where $h\left(z\right)$
is a function of $z$. In fact,
\begin{align*}
 & \mathbb{E}\sum_{j=1}^{n}W_{ij}C_{j}\left(z\right)W_{ij}^{T}\\
= & h\left(z\right)\mathbb{E}\sum_{j=1}^{n}W_{ij}W_{ij}^{T}\\
= & \left(1-p\right)h\left(z\right)\mathbb{E}\left(\left(\frac{I_{d}-R}{\left\Vert I_{d}-R\right\Vert }-c\left(d\right)I_{d}\right)\left(\frac{I_{d}-R}{\left\Vert I_{d}-R\right\Vert }-c\left(d\right)I_{d}\right)^{T}\right)\\
= & \left(1-p\right)h\left(z\right)\mathbb{E}\left(\frac{2I_{d}-R-R^{T}}{\left\Vert I_{d}-R\right\Vert ^{2}}-2c\left(d\right)\frac{I_{d}-R}{\left\Vert I_{d}-R\right\Vert }+c\left(d\right)^{2}I_{d}\right)\\
= & \left(1-p\right)h\left(z\right)\left(\frac{1}{d}I_{d}-2c\left(d\right)^{2}I_{d}+c\left(d\right)^{2}I_{d}\right)\\
= & \left(1-p\right)h\left(z\right)\left(\frac{1}{d}-c\left(d\right)^{2}\right)I_{d},
\end{align*}
where the fourth equality uses Lemma \ref{lem:lem_appendix_b}. From (\ref{eq:C_i}) we obtain
\[
h\left(z\right)I_{d}=\left(zI_{d}-\left(1-p\right)h\left(z\right)\left(\frac{1}{d}-c\left(d\right)^{2}\right)I_{d}\right)^{-1},
\]
that is,
\[
\left(1-p\right)\left(\frac{1}{d}-c\left(d\right)^{2}\right)h\left(z\right)^{2}-zh\left(z\right)+1=1,
\]
which reduces to
\[
h\left(z\right)=\frac{z\pm\sqrt{z^{2}-4\left(1-p\right)\left(\frac{1}{d}-c\left(d\right)^{2}\right)}}{2\left(1-p\right)\left(\frac{1}{d}-c\left(d\right)^{2}\right)}.
\]
The uniqueness of the solution now implies that the Cauchy transform
of $F_{n}\left(x\right)$ is
\[
\frac{1}{nd}\sum_{i=1}^{n}\text{Tr}\left(C_{i}\left(z\right)\right)=h\left(z\right)=\frac{z-\sqrt{z^{2}-4\left(1-p\right)\left(\frac{1}{d}-c\left(d\right)^{2}\right)}}{2\left(1-p\right)\left(\frac{1}{d}-c\left(d\right)^{2}\right)},
\]
where we select the $-$ branch according to the properties of the
Cauchy transform. This is exactly the Cauchy transform of the semicircle
law with support $\left[-2\sqrt{\left(1-p\right)\left(\frac{1}{d}-c\left(d\right)^{2}\right)},2\sqrt{\left(1-p\right)\left(\frac{1}{d}-c\left(d\right)^{2}\right)}\right]$.
Hence, for large $n$ the eigenvalues of $D$ are distributed according
to the semicircle law with support
\[
\left[-2\sqrt{\left(1-p\right)\left(\frac{1}{d}-c\left(d\right)^{2}\right)\left(n-1\right)},2\sqrt{\left(1-p\right)\left(\frac{1}{d}-c\left(d\right)^{2}\right)\left(n-1\right)}\right].
\]
Thus the
limiting spectral density of $\frac{1}{\sqrt{n-1}}D$ is Wigner's
semi-circle. Since any symmetric
matrix can be decomposed into a superposition of a positive semidefinite
matrix and a negative definite matrix (this follows immediately from
the spectral decomposition of the matrix), the matrix $D$ can be
decomposed as
\[
D=D_{+}+D_{-},
\]
where $D_{+}\succcurlyeq0$ and $D_{-}\prec0$. Clearly,
\[
\left\Vert D\right\Vert ^{2}=\left\Vert D_{+}\right\Vert ^{2}+\left\Vert D_{-}\right\Vert ^{2},
\]
and since the limiting spectral density of $\frac{1}{\sqrt{n-1}}D$
is Wigner's semi-circle that is symmetric around $0$, we have
\[
\left\Vert D_{+}\right\Vert ^{2}\approx\left\Vert D_{-}\right\Vert ^{2}\approx\frac{1}{2}\left\Vert D\right\Vert ^{2}.
\]
From the law of large numbers we get that $\left\Vert D\right\Vert ^{2}$
is concentrated at
\begin{eqnarray*}
\left\Vert D\right\Vert ^{2} & \approx & \left(1-p\right)n\left(n-1\right)\mathbb{E}\left(\left\Vert \frac{I_{d}-R_{ij}}{\left\Vert I_{d}-R_{ij}\right\Vert }-c\left(d\right)I_{d}\right\Vert ^{2}\right)\\
 & = & \left(1-p\right)n\left(n-1\right)\left(1+c\left(d\right)^{2}d-2c\left(d\right)\mathbb{E}\left(\text{Tr}\left(\frac{I_{d}-R}{\left\Vert I_{d}-R\right\Vert }\right)\right)\right)\\
 & = & \left(1-p\right)n\left(n-1\right)\left(1-c\left(d\right)^{2}d\right).
\end{eqnarray*}
Hence,
\begin{equation*}
\left\Vert D_{+}\right\Vert ^{2}\approx\left\Vert D_{-}\right\Vert ^{2}\approx\frac{1}{2}\left(1-p\right)n\left(n-1\right)\left(1-c\left(d\right)^{2}d\right).
\end{equation*}
\section{Proof of Lemma \ref{lem:f2_off_bound}}
\label{sec:contribution_Q}
We have already shown that the gain from $Q$ is at most $\mathcal{O}_{P}\left(\sqrt{n\log n}\right)\sum_{i=1}^{n}\|Q_{ii}^{1}\|$.
Now we will show that the gain from $Q$ is less than the
sum of $\mathcal{O}_{P}\left(\sqrt{n\log n}\right)\text{Tr}\left(T\right)$
and the loss in the off-diagonal entries of $\Delta_{ij}$ for $\left(i,j\right)\in \mathcal{E}_{c}$, specifically speaking,
\begin{align}
&\mathcal{O}_{P}\left(\sqrt{n\log n}\right)\sum_{i=1}^{n}\|Q_{ii}^{1}\|\nonumber\\
\leq&\mathcal{O}_{P}\left(\sqrt{n\log n}\right)\text{Tr}\left(T\right)+\frac{d^{2}}{pn}\mathcal{O}_{P}\left(\sqrt{n\log n}\right)\sum_{\left(i,j\right)\in \mathcal{E}_{c}}\left\Vert \Delta_{ij}^{\text{off}}\right\Vert ,\label{eq:contribution_Q}
\end{align}
which is equivalent to prove Lemma \ref{lem:f2_off_bound}.

A matrix is skew-symmetric if $X^{T}=-X$. Every matrix $X$ can be
written as the sum of a symmetric matrix and a skew-symmetric matrix:
$X=\frac{X+X^{T}}{2}+\frac{X-X^{T}}{2}$. We will call the first term
the symmetric part and the second term the skew-symmetric part. Recall
that $Q^{1}$ is a matrix with identical columns ($Q_{ji}^{1}=Q_{ii}^{1}$).
We denote the symmetric part and the skew-symmetric part of $Q_{ji}^{1}$
as
\begin{equation}
Q_{i}^{1,s}=\frac{Q_{ii}^{1}+Q_{ii}^{2}}{2}=-\frac{T_{ii}+P_{ii}}{2}\label{eq:Q_sym}
\end{equation}
 and
\begin{equation}
Q_{i}^{1,ss}=\frac{Q_{ii}^{1}-Q_{ii}^{2}}{2}\label{eq:Q_skew}
\end{equation}
 respectively. Then we have
\begin{equation}
\|Q_{ii}^{1}\|=\left\Vert Q_{i}^{1,s}+Q_{i}^{1,ss}\right\Vert \leq\left\Vert Q_{i}^{1,s}\right\Vert +\left\Vert Q_{i}^{1,ss}\right\Vert .\label{eq:Q_two_parts}
\end{equation}
Later we will see the skew-symmetric part $Q_{i}^{1,ss}$ is exactly
what's causing troubles. Before dealing with the trouble, let's first
handle the symmetric part $Q_{i}^{1,s}$.

For the symmetric part $Q_{i}^{1,s}$ we have
\begin{eqnarray}
\sum_{i=1}^{n}\left\Vert Q_{i}^{1,s}\right\Vert  & = & \frac{1}{2}\sum_{i=1}^{n}\left\Vert T_{ii}+P_{ii}\right\Vert \nonumber \\
 & \leq & \frac{1}{2}\left(\sum_{i=1}^{n}\left\Vert T_{ii}\right\Vert +\sum_{i=1}^{n}\left\Vert P_{ii}\right\Vert \right)\nonumber \\
 & = & \frac{1}{2}\left(\sum_{i=1}^{n}\left\Vert T_{ii}\right\Vert +\left\Vert \sum_{i=1}^{n}T_{ii}\right\Vert \right)\nonumber \\
 & \leq & \sum_{i=1}^{n}\left\Vert T_{ii}\right\Vert \leq\sum_{i=1}^{n}\text{Tr}\left(T_{ii}\right)=\text{Tr}\left(T\right),\label{eq:Q_sym_bound}
\end{eqnarray}
where the second equality uses Lemma \ref{lem:sum_trace_p} and the third inequality uses
the fact that $T_{ii}\succcurlyeq0$ and that for any matrix $X\succcurlyeq0$,
$\left\Vert X\right\Vert \leq\text{Tr}\left(X\right)$.

Let us consider now the skew-symmetric part $Q_{i}^{1,ss}$. For any
index $p,q\in\{1,2,...,d\}$, consider $\ensuremath{\sum_{i=1}^{n}|Q_{i}^{1,ss}\left(p,q\right)|}$,
where $Q_{i}^{1,ss}\left(p,q\right)$ is the $\left(p,q\right)$th entry of the
matrix $Q_{i}^{1,ss}$. Without loss of generality and for the purpose
of simplicity, let us assume that $\sum_i\left|Q_{i}^{1,ss}\left(p,q\right)\right|$
is largest when $p=1$ and $q=2$ (it won't be largest
on diagonal, because diagonal entries have to be $0$). Now we know
\begin{equation}
\sum_{i=1}^{n}\|Q_{i}^{1,ss}\|\le d^{2}\sum_{i=1}^{n}|Q_{i}^{1,ss}\left(1,2\right)|.\label{eq:Q_i_ss}
\end{equation}
This is just because $Q_{i}^{1,ss}$ has $d^{2}$ entries. Thus later we will
just care about $Q_{i}^{1,ss}\left(1,2\right)$ and $Q_{i}^{1,ss}\left(2,1\right)$.
We define the matrices $Q^{1,s}$, $Q^{1,ss}$, $Q^{2,s}$, and $Q^{2,ss}$
as
\begin{align}
Q_{ij}^{1,s} & =Q_{j}^{1,s},\: Q_{ij}^{1,ss}=Q_{j}^{1,ss},\: Q_{ij}^{2,s}=Q_{i}^{1,s},\: Q_{ij}^{2,ss}=-Q_{i}^{1,ss}\text{ for all }i\text{ and }j,\label{eq:def_Q_s_ss}
\end{align}
and the matrix $Q^{s}=Q^{1,s}+Q^{2,s}$ and $Q^{ss}=Q^{1,ss}+Q^{2,ss}$.
Thus it is easy to see $Q=Q^{s}+Q^{ss}$. Restrict the matrices $P$,
$Q$, $T$ and so on to the good entries and obtain $P_{c}$, $Q_{c}$,
$T_{c}$ and so on (which means, for example, when $\left(i,j\right)\notin \mathcal{E}_{c}$,
set the $\left(i,j\right)$'s block $P_{c}\left(i,j\right)=0$; when
$\left(i,j\right)\in \mathcal{E}_{c}$, keep $P_{c}\left(i,j\right)=P\left(i,j\right)$.
Then we can prove the following lemmas.
\begin{lemma}
\label{lem:Q_1}
Let the vectors ${\bf s}_{1},{\bf s}_{2}\in\mathbb{R}^{nd}$ be defined
as in (\ref{eq:s}), that is,
\begin{eqnarray*}
{\bf s}_{1}&=&\left(1,0,0,\ldots0,1,0,0,\ldots,0,\ldots\ldots,1,0,0,\ldots,0\right)^{T},\\
 {\bf s}_{2}&=&\left(0,1,0,\ldots0,0,1,0,\ldots,0,\ldots\ldots,0,1,0,\ldots,0\right)^{T}.
\end{eqnarray*}
And define the vectors ${\bf a}_{1},{\bf a}_{2}\in\mathbb{R}^{nd}$
as following:
\begin{equation}
{\bf a}_{l}\left(m\right)=\begin{cases}

1 & \text{if }m=d\cdot i+l\text{ and }Q_{i}^{1,ss}\left(p,q\right)>0,\\
-1 & \text{if }m=d\cdot i+l\text{ and }Q_{i}^{1,ss}\left(p,q\right)<0,\\
0 & \text{otherwise,}
\end{cases}\qquad\text{for }l=1,2.\label{eq:def_a}
\end{equation}
Then we have the following inequality
\begin{equation}
{\bf s}_{1}^{T}\left(P_{c}+Q_{c}\right){\bf a}_{2}-{\bf s}_{2}^{T}\left(P_{c}+Q_{c}\right){\bf a}_{1}\geq\left(\Omega(pn)-\mathcal{O}_{P}\left(\sqrt{n\log n}\right)\right)\sum_{i=1}^{n}\left|Q_{i}^{1,ss}\left(1,2\right)\right|.\label{eq:Q_12_bound}
\end{equation}

\end{lemma}
\begin{proof}
Firstly, note that for any matrix $X\in\mathbb{R}^{nd\times nd}$,
${\bf s}_{1}^{T}X{\bf a}_{2}-{\bf s}_{2}^{T}X{\bf a}_{1}$ is the
sum of the differences between $X_{ij}\left(1,2\right)$ and $X_{ij}\left(2,1\right)$.
Thus it is easy to verify that the symmetric matrix $P_{c}$ and the
symmetric parts $Q_{i}^{1,s}$ and $Q_{i}^{2,s}$ in $Q_{c}^{s}$
contribute $0$ to the LHS of (\ref{eq:Q_12_bound}), that
is,
\begin{equation}
{\bf s}_{1}^{T}\left(P_{c}+Q_{c}\right){\bf a}_{2}-{\bf s}_{2}^{T}\left(P_{c}+Q_{c}\right){\bf a}_{1}={\bf s}_{1}^{T}Q_{c}^{ss}{\bf a}_{2}-{\bf s}_{2}^{T}Q_{c}^{ss}{\bf a}_{1}=2{\bf s}_{1}^{T}Q_{c}^{ss}{\bf a}_{2},\label{eq:sym_part}
\end{equation}
where the second equality uses the fact that $-{\bf s}_{2}^{T}Q_{c}^{ss}{\bf a}_{1}={\bf s}_{1}^{T}Q_{c}^{ss}{\bf a}_{2}$
due to the skew-symmetry of the submatrices of $Q_{c}^{ss}$.
\begin{figure}[H]

\hspace{1cm}\subfloat[\label{fig:M1}$M_a$]{

\includegraphics[width=0.3\columnwidth]{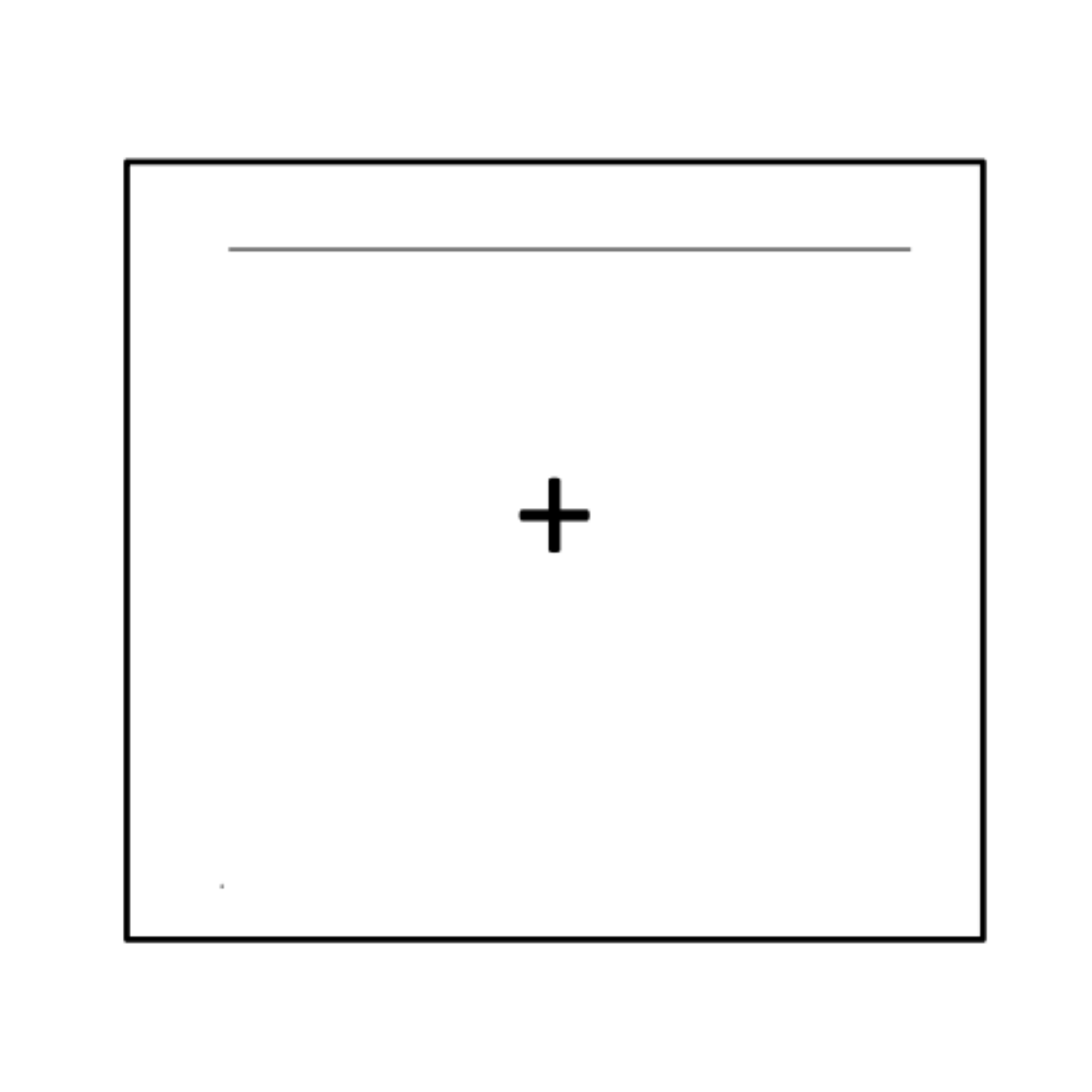}}
\hspace{2cm}\subfloat[\label{fig:M2}$M_b$]{

\includegraphics[width=0.3\columnwidth]{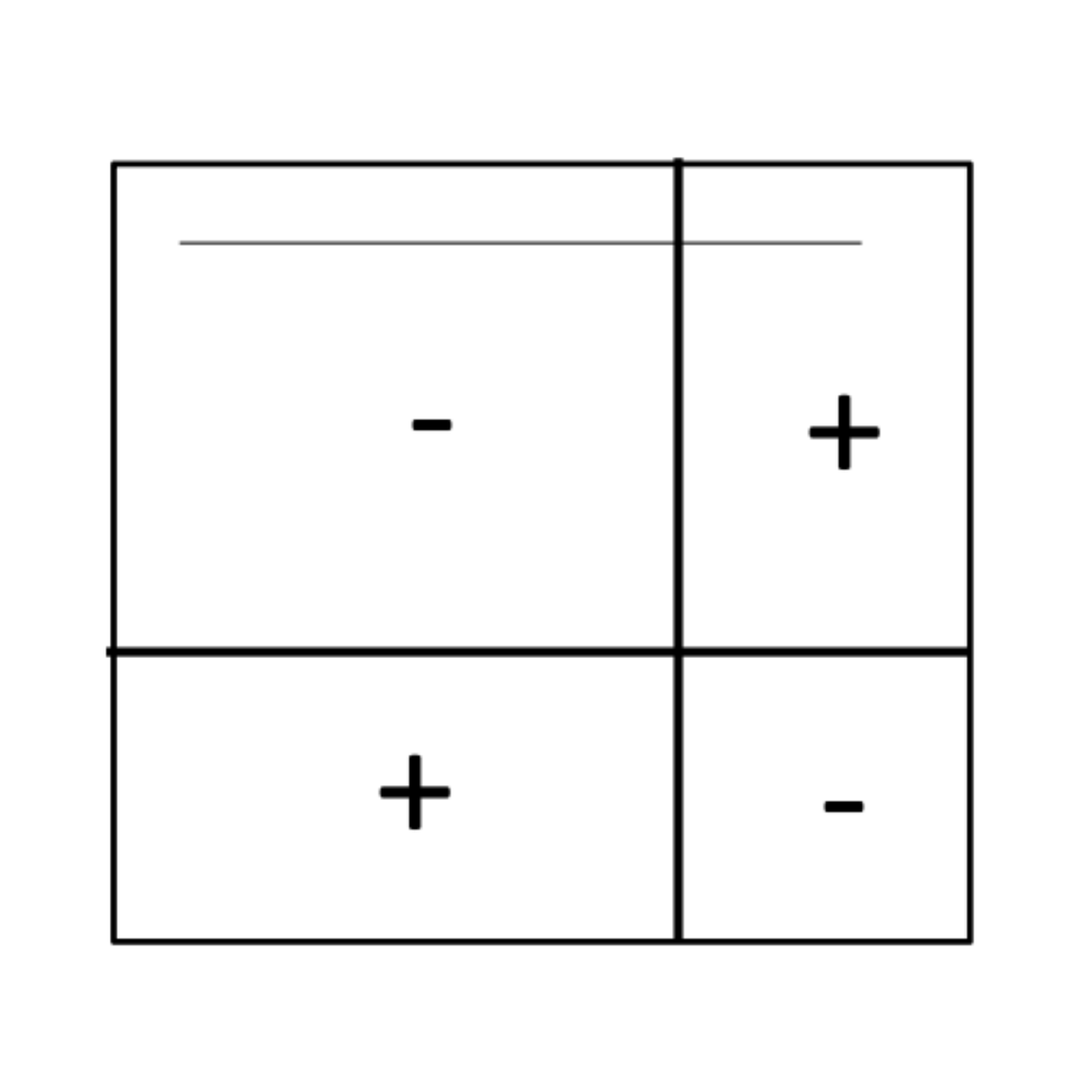}}

%\hspace{3.5cm}\subfloat[\label{fig:M3}Contribution from $M_b.$]{

%\includegraphics[width=0.4\columnwidth]{matrix_3}}

\caption[Two $n\times n$ matrices.]{Two $n\times n$ matrices. The matrix $M_a$ is defined as $M_a(i,j)=\left|Q_{i}^{1,ss}\left(1,2\right)\right|$ if $(i,j)\in \mathcal{E}_c$, otherwise $M_a(i,j)=0$. The matrix $M_b$ is defined as $M_b(i,j)=-Q_{i}^{1,ss}\left(1,2\right)\text{sign}\left(Q_{j}^{1,ss}\left(1,2\right)\right)$ if $(i,j)\in \mathcal{E}_c$, otherwise $M_b(i,j)=0$.  With the assumption that
there is an integer $i_{0}$ such that $Q_{i}^{1,ss}\geq0$ when
$i\leq i_{0}$ and $Q_{i}^{1,ss}<0$ when $i>i_{0}$, $M_a$ and $M_b$ have non-negative regions with sign ``+'' and non-positive regions with sign ``-'' as shown in sub-figure (a) and (b). The RHS of (\ref{eq:s1_Q1_a2}) and (\ref{eq:s1_Q2_a2}) can be
represented as sums of the entries $\left(i,j\right)\in \mathcal{E}_{c}$ of
$M_a$ and $M_b$ respectively. When taking a summation of the RHS of (\ref{eq:s1_Q1_a2}) and (\ref{eq:s1_Q2_a2}), the entries in the ``-'' regions of $M_b$ and those corresponding entries in $M_a$ cancel each other out. Thus the contribution to the sum comes from the entries at $(i,j)\in \mathcal{E}_c$ in ``+'' regions of $M_b$ and the corresponding parts in $M_a$.}\label{fig:ThreeMatrices}
\end{figure}
Now we just need to focus on ${\bf s}_{1}^{T}Q_{c}^{ss}{\bf a}_{2}={\bf s}_{1}^{T}\left(Q_{c}^{1,ss}+Q_{c}^{2,ss}\right){\bf a}_{2}$.
For each of the submatrices of $Q_{c}^{1,ss}$ and $Q_{c}^{2,ss}$,
we only look at their entry at index $\left(1,2\right)$. Due to the
definition in (\ref{eq:def_Q_s_ss}) and (\ref{eq:def_a}) we have
\begin{equation}
{\bf s}_{1}^{T}Q_{c}^{1,ss}{\bf a}_{2}=\sum_{\left(i,j\right)\in \mathcal{E}_{c}}Q_{ij}^{1,ss}\left(1,2\right)\text{sign}\left(Q_{j}^{1,ss}\left(1,2\right)\right)=\sum_{\left(i,j\right)\in \mathcal{E}_{c}}\left|Q_{i}^{1,ss}\left(1,2\right)\right|\label{eq:s1_Q1_a2}
\end{equation}
and
\begin{eqnarray}
{\bf s}_{1}^{T}Q_{c}^{2,ss}{\bf a}_{2}&=&\sum_{\left(i,j\right)\in \mathcal{E}_{c}}Q_{ij}^{2,ss}\left(1,2\right)\text{sign}\left(Q_{j}^{1,ss}\left(1,2\right)\right)\nonumber\\
&=&\sum_{\left(i,j\right)\in \mathcal{E}_{c}}-Q_{i}^{1,ss}\left(1,2\right)\text{sign}\left(Q_{j}^{1,ss}\left(1,2\right)\right).\label{eq:s1_Q2_a2}
\end{eqnarray}
Without loss of generality and for the purpose of simplicity, we assume
there is an integer $i_{0}$ such that $Q_{i}^{1,ss}\geq0$ when
$i\leq i_{0}$ and $Q_{i}^{1,ss}<0$ when $i>i_{0}$. Then the RHS of (\ref{eq:s1_Q1_a2}) and (\ref{eq:s1_Q2_a2}) can be
represented as sums of the entries $\left(i,j\right)\in \mathcal{E}_{c}$ of
the $n\times n$ matrices in Figure \ref{fig:M1}  and Figure \ref{fig:M2} respectively.
Apparently in some parts, even if there is a edge in $\mathcal{E}_{c}$ things
will cancel (these are the parts that are marked negative in Figure
\ref{fig:M2}). But in the parts that are marked positive in Figure \ref{fig:M2}, if there is an edge
in $\mathcal{E}_{c}$, things will not cancel and will indeed contribute to
${\bf s}_{1}^{T}Q_{c}^{ss}{\bf a}_{2}$. In fact, the contribution of ``+'' regions is at least
\begin{equation}
\left(\Omega(pn)-\mathcal{O}_{P}\left(\sqrt{n\log n}\right)\right)\sum_{i=1}^{n}\left|Q_{i}^{1,ss}\left(1,2\right)\right| (\leq 2{\bf s}_{1}^{T}Q_{c}^{ss}{\bf a}_{2})
\label{eq:contribution}
\end{equation}

where we use Lemma \ref{lem:contribution_bound} and the fact that $$\sum_{Q_{i}^{ss}\left(1,2\right)>0}Q_{i}^{ss}\left(1,2\right)=-\sum_{Q_{i}^{ss}\left(1,2\right)<0}Q_{i}^{ss}\left(1,2\right)=\sum_{i=1}^{n}|Q_{i}^{ss}\left(1,2\right)|/2.$$
Combine (\ref{eq:sym_part}) and (\ref{eq:contribution})
together and we obtain (\ref{eq:Q_12_bound}).
\end{proof}

To verify that the contribution of ``+" regions is lower bounded by (\ref{eq:contribution}), we first assume that in Figure \ref{fig:ThreeMatrices}, every good edge in the ``+'' area of (b) contributes the same amount (that corresponds to $Q^{1,ss}_i(1,2)$ has the same absolute value for all $i$), then there are two cases:

1. The rectangle has large area, in this case we will show ALL large rectangles have a large number of good edges.

2. The rectangle has small area, in this case the ``+'' area must be very wide (in order to have small area, it must be a $q\times (n-q)$ rectangle where $q$ is very small), in fact its width $(n-q)$ will be much wider than $n/2$, and will even be wider than $(1-p)n$ (where $p$ is the probability of a good edge), so there must be many good edges inside the rectangle.

\begin{lemma}
\label{lem:Q_same}
Let $p$ be the probability of a good edge. When $p > log^C n$ for some fixed constant $C$, with high probability for any set $S_1, S_2 \subset \{1,2,...,n\}$, where $|S_1| = q$ and $|S_2| \ge n - C_2 q$ ($C_2 = 16$ should be good enough), the number of good edges in $S_1\times S_2$ is at least $C_3 p |S_1| |S_2|$ ($C_3$ is some universal constant, and can be $1/4$).
\end{lemma}

\begin{proof}
With high probability every vertex has at least $pn/2$ good edges, so when $C_2q < pn/4$ the Lemma holds trivially.

When $C_2q \ge pn/4$, use Chernoff bound. Fix the size of $S_1$ and $S_2$, for any particular $S_1, S_2$, the probability that the number of good edges is too small is bounded by

$$e^{- \Omega((1-C_3)p|S_1||S_2|/ \sqrt{p|S_1||S_2|})^2} = e^{-\Omega(pqn)}$$

%(Explanation: the expectation is $p|S_1||S_2|$, standard deviation is at most $\sqrt{p|S_1||S_2|}$, since this should behave like a Gaussian, the probability that it deviates $t$ times standard deviation is $e^{-c t^2}$ for some constant $c$).

On the other hand, the number of such pairs $S_1,S_2$ is at most $n^q \cdot n^{C_2q} = e^{O(q\log n)}$, by union bound the probability that there exists a bad pair is very small (because $q\log n \ll pqn$).
\end{proof}

Now we are essentially done if all the positive entries of $Q^{1,ss}_i(1,2)$ are of similar size. We are not really done because of the following bad case: One of $Q^{1,ss}_i(1,2)$ is very large, the others are small. Now although there are many good edges in the ``+'' region, their real contribution depends on how many of the good edges belong to the vertex with large $Q^{1,ss}_i(1,2)$. Thus we need to group positive entries of $Q^{1,ss}_i(1,2)$ according to their value, and apply Lemma \ref{lem:Q_same} to different groups.

\begin{lemma}
\label{lem:contribution_bound}
\begin{equation}
{\bf s}_{1}^{T}Q_{c}^{ss}{\bf a}_{2} = {\bf s}_{1}^{T}Q_{c}^{1,ss}{\bf a}_{2} +{\bf s}_{1}^{T}Q_{c}^{2,ss}{\bf a}_{2} \geq (\Omega(pn) - \mathcal{O}_P(\sqrt{n\log n})) \sum_{Q^{1,ss}_i(1,2) < 0} - Q^{1,ss}_i(1,2)
\label{eq:contribution_bound}
\end{equation}
\end{lemma}

\begin{proof}
Assume without loss of generality that the number of negative entries in $Q^{1,ss}_i(1,2)$ is less than $n/2$ (otherwise take $-Q$, the loss will be the same).

Let $N$ be the set of nonnegative entries of $Q^{1,ss}_i(1,2)$. Let $m$ be the minimum entry in $Q^{1,ss}_i(1,2)$ (which is negative). Let $S_j$ ($j\ge 0$) be the set of entries that are between $[m\cdot 2^{-j}, m\cdot 2^{-j-1})$.

We shall consider the rectangles $S_j \times (\{1,2,...,n\} - S_{j-1} - S_{j}- S_{j+1})$, we would like to show that the number of good edges in all these rectangles are at least $C_3p |S_j| \cdot n/2$ (the second set is larger than $n/2$ because it contains $N$), and each of these good edges will contribute at least $- m\cdot 2^{-j-2}$ (remember $m$ is negative). In fact, these edges must be from $S_j$ to something in $S_{j-2}$ or $S_{j+2}$ or even farther index, the difference between any entry in $S_j$ and any entry in those sets are at least $2^{-j-2}$, thus they will not cancel completely. Therefore the total contribution will be at least

$$
\sum_{j\ge 0} C_3p |S_{j}| \cdot n/2 \cdot (- m)\cdot 2^{-j-2} \ge (\Omega(pn) -\mathcal{O}_P(\sqrt{n\log n})) \sum_{Q^{1,ss}_i(1,2) < 0} - Q^{1,ss}_i(1,2).$$

We can show there are many good edges by Lemma \ref{lem:Q_same}, the only thing we need to guarantee is the second set $\{1,2,...,n\}-S_{j-1}-S_{j}-S_{j+1}$ is large enough.

When $\{1,2,...,n\}-S_{j-1}-S_{j}-S_{j+1}$ is small, there are two possibilities here:

1. $|S_{j-1}|$ might be larger than $C_2|S_j|/2$, but in this case the sum of entries in $S_{j-1}$ is constant times larger than $S_{j}$, thus we can ignore $S_{j}$.

2. $|S_{j+1}|$ might be larger than $C_2|S_{j}|/2$, but we have chosen $C_2$ to be large enough, so in this case the sum of entries in $S_{j+1}$ is also constant times larger than $S_{j}$, thus we can ignore $S_{j}$.

To clarify points 1 and 2, let $h(j)$ be the absolute value of sum of entries in $S_{j}$, i.e.,
\begin{equation}
h(j) =\sum_{Q^{1,ss}_i(1,2) < 0, i \in S_j} - Q^{1,ss}_i(1,2).
\end{equation}
We will show the contributions from the rectangles are at least
$$(\Omega(pn) -\mathcal{O}_P(\sqrt{n\log n})\sum_j h(j).$$
Points 1 and 2 say if $h(j)$ is at most $h(j+1)/2$
or $h(j-1)/2$ then we ignore $h(j)$. This is fine because if we look at a $h(j)$ that is not ignored (there must be such sets, otherwise the sum of $h(j)$ will be infinity), it can only be responsible for $h(j-1), h(j-2), ...$ and $h(j+1), h(j+2), ...$. And we know $h(j-t) \le 2^{-t} h(j)$ and $h(j+t) \le 2^{-t} h(j)$ (because they are all ignored). The sum of all these $h(j-t)$'s and $h(j+t)$'s are bounded by constant times $h(j)$.
\end{proof}
%If $h(j)$'s are $1,2,4,8,4,2,1$, then $8$ might be responsible for all the other numbers, but even the sum of these numbers are small.

\begin{lemma}
\label{Q_2}
\begin{equation}
\sum_{i=1}^{n}\|Q_{i}^{1,ss}\|\leq\frac{d^{2}}{\Omega(pn)-\mathcal{O}_{P}\left(\sqrt{n\log n}\right)}\left(\sum_{\left(i,j\right)\in \mathcal{E}_{c}}\left\Vert \Delta_{ij}^{\text{off}}\right\Vert +\frac{n}{2}\text{Tr}\left(T\right)\right).\label{eq:Q_i_1_ss_bound}
\end{equation}

\end{lemma}
\begin{proof}
We know that
\begin{align}
 & {\bf s}_{1}^{T}\left(P_{c}+Q_{c}+T_{c}\right){\bf a}_{2}-{\bf s}_{2}^{T}\left(P_{c}+Q_{c}+T_{c}\right){\bf a}_{1}\nonumber \\
= & {\bf s}_{1}^{T}\Delta_{c}{\bf a}_{2}-{\bf s}_{2}^{T}\Delta_{c}{\bf a}_{1}\nonumber \\
\leq & \sum_{\left(i,j\right)\in \mathcal{E}_{c}}\left|\Delta_{ij}\left(1,2\right)\right|+\sum_{\left(i,j\right)\in \mathcal{E}_{c}}\left|\Delta_{ij}\left(2,1\right)\right|\nonumber \\
\leq & 2\sum_{\left(i,j\right)\in \mathcal{E}_{c}}\left\Vert \Delta_{ij}^{\text{off}}\right\Vert ,\label{eq:delta_off}
\end{align}
where $\mbox{\ensuremath{\Delta_{ij}^{\text{off}}}}$ is obtained
by restricting $\Delta_{ij}$ on the off-diagonal entries. In addition,
we have
\[
-{\bf s}_{1}^{T}T_{c}{\bf a}_{2}\leq\sum_{i,j}\left|T_{ij}\left(1,2\right)\right|\leq\frac{1}{2}\sum_{i,j}\left(T_{ii}\left(1,1\right)+T_{jj}\left(2,2\right)\right)=\frac{n}{2}\text{Tr}\left(T\right),
\]
where the second inequality uses the fact that $T\succcurlyeq0$,
and similarly we have ${\bf s}_{2}^{T}T_{c}{\bf a}_{1}\leq\frac{n}{2}\text{Tr}\left(T\right)$,
thus we obtain
\begin{equation}
-{\bf s}_{1}^{T}T_{c}{\bf a}_{2}+{\bf s}_{2}^{T}T_{c}{\bf a}_{1}\leq n\text{Tr}\left(T\right).\label{eq:T_part}
\end{equation}
Combining (\ref{eq:Q_i_ss}), (\ref{eq:Q_12_bound}), (\ref{eq:delta_off})
and (\ref{eq:T_part}) together we get (\ref{eq:Q_i_1_ss_bound}).
\end{proof}

Using (\ref{eq:Q_two_parts}), (\ref{eq:Q_sym_bound}) and (\ref{eq:Q_i_1_ss_bound}), we reach the conclusion in (\ref{eq:contribution_Q}).

\section{Proof of strong stability of LUD (Theorem \ref{thm:stability_strong})}
\label{sec:strong_stability}
To prove strong stability of LUD, we need a tighter lower bound for
$f_{2}$. First we define the loss from the diagonal entries $f_{2}^{\text{d}}$
and that from the off-diagonal entries $f_{2}^{\text{off}}$ as
\begin{eqnarray*}
f_{2}^{\text{d}} & = & \sum_{\left(i,j\right)\in\mathcal{E}_{c}}\left(\left\Vert \left(I_{d}-R_{ij}+\Delta_{ij}\right)^{\text{d}}\right\Vert -\left\Vert \left(I_{d}-R_{ij}\right)^{\text{d}}\right\Vert \right)\\
f_{2}^{\text{off}} & = & \sum_{\left(i,j\right)\in\mathcal{E}_{c}}\left(\left\Vert \left(I_{d}-R_{ij}+\Delta_{ij}\right)^{\text{off}}\right\Vert -\left\Vert \left(I_{d}-R_{ij}\right)^{\text{off}}\right\Vert \right).
\end{eqnarray*}
Then we have
\begin{eqnarray}
f_{2}^{\text{d}} & \geq & \sum_{\left(i,j\right)\in\mathcal{E}_{c}}\left(\left\Vert \Delta_{ij}^{\text{d}}\right\Vert -2\left\Vert \left(I_{d}-R_{ij}\right)^{\text{d}}\right\Vert \right)\nonumber \\
 & \geq & \left(\frac{pn}{\sqrt{d}}-\mathcal{O}_{P}\left(\sqrt{n\log n}\right)\right)\text{Tr}\left(T\right)-2pn^{2}\epsilon^{2},\label{eq:f2_d_rhs}
\end{eqnarray}
since $\left\Vert \left(I_{d}-R_{ij}\right)^{\text{d}}\right\Vert \approx\epsilon^{2}$.
If $\text{Tr}\left(T\right)=\mathcal{O}\left(n\epsilon^{2}\right)$,
then using Lemma \ref{lem:Q_norm2} we can show that $\left\Vert Q\right\Vert ^{2}=\mathcal{O}\left(n^{2}\epsilon^{2}\right)$.
And using (\ref{eq:Delta2_decompostion}) we are done. Thus we will
continue with the case when $\text{Tr}\left(T\right)\gg n\epsilon^{2}$,
where the term $2pn^{2}\epsilon^{2}$ in (\ref{eq:f2_d_rhs}) is negligible.
Now let us consider $f_{2}^{\text{off}}$. We further decompose $f_{2}^{\text{off}}$
to two parts as following

\begin{equation}
f_{2}^{\text{off}}=\sum_{\left(i,j\right)\in\mathcal{E}_{c}}\left\langle \frac{\left(I_{d}-R_{ij}\right)^{\text{off}}}{\left\Vert \left(I_{d}-R_{ij}\right)^{\text{off}}\right\Vert },\Delta_{ij}^{\text{off}}\right\rangle +\sum_{\left(i,j\right)\in\mathcal{E}_{c}}l\left(\Delta_{ij}\right),\label{eq:f_2_off}
\end{equation}
where
\[
l\left(\Delta_{ij}\right):=\left\Vert \left(I_{d}-R_{ij}+\Delta_{ij}\right)^{\text{off}}\right\Vert -\left\Vert \left(I_{d}-R_{ij}\right)^{\text{off}}\right\Vert -\left\langle \frac{\left(I_{d}-R_{ij}\right)^{\text{off}}}{\left\Vert \left(I_{d}-R_{ij}\right)^{\text{off}}\right\Vert },\Delta_{ij}^{\text{off}}\right\rangle \geq0.
\]
Apply the same analysis as that in Section \ref{sec:gain} to $\sum_{\left(i,j\right)\in\mathcal{E}_{c}}\left\langle \frac{\left(I_{d}-R_{ij}\right)^{\text{off}}}{\left\Vert \left(I_{d}-R_{ij}\right)^{\text{off}}\right\Vert },\Delta_{ij}^{\text{off}}\right\rangle $
and notice that $\mathbb{E}\left(\frac{\left(I_{d}-R_{ij}\right)^{\text{off}}}{\left\Vert \left(I_{d}-R_{ij}\right)^{\text{off}}\right\Vert }\right)=0$,
we obtain
\begin{eqnarray}
&&\sum_{\left(i,j\right)\in\mathcal{E}_{c}}\left\langle \frac{\left(I_{d}-R_{ij}\right)^{\text{off}}}{\left\Vert \left(I_{d}-R_{ij}\right)^{\text{off}}\right\Vert },\Delta_{ij}^{\text{off}}\right\rangle \nonumber\\
&\geq&-\left(\frac{1}{\sqrt{2}}\sqrt{1-p}n+\mathcal{O}_{P}\left(\sqrt{\log n}\right)\right)\text{Tr}\left(T\right)-\mathcal{O}_{P}\left(\sqrt{n\log n}\right)\sum_{i}\left\Vert Q_{ii}^{1}\right\Vert .\label{eq:f_2_off_part1}
\end{eqnarray}
For the part $\sum_{\left(i,j\right)\in\mathcal{E}_{c}}l\left(\Delta_{ij}\right)$,
we have
\begin{eqnarray}
\sum_{\left(i,j\right)\in\mathcal{E}_{c}}l\left(\Delta_{ij}\right) & = & \sum_{\left(i,j\right)\in\mathcal{E}_{c}}l\left(P_{ij}+Q_{ij}^{s}+Q_{ij}^{ss}+T_{ij}\right)\nonumber \\
 & \geq & \sum_{\left(i,j\right)\in\mathcal{E}_{c}}l\left(Q_{ij}^{ss}\right)-2\sum_{\left(i,j\right)\in\mathcal{E}_{c}}\left(\left\Vert P_{ij}^{\text{off}}\right\Vert +\left\Vert T_{ij}^{\text{off}}\right\Vert +\left\Vert Q_{ij}^{s}\right\Vert \right)\nonumber \\
 & \geq & \sum_{\left(i,j\right)\in\mathcal{E}_{c}}l\left(Q_{ij}^{ss}\right)-c_{8}n\text{Tr}\left(T\right).\label{eq:f2_off_part2}
\end{eqnarray}
Now we consider
\begin{eqnarray*}
&&\sum_{\left(i,j\right)\in\mathcal{E}_{c}}l\left(Q_{ij}^{ss}\right) \\
& = & \sum_{\left(i,j\right)\in\mathcal{E}_{c}}\left(\left\Vert \left(I_{d}-R_{ij}+Q_{ij}^{ss}\right)^{\text{off}}\right\Vert -\left\Vert \left(I_{d}-R_{ij}\right)^{\text{off}}\right\Vert -\left\langle \frac{\left(I_{d}-R_{ij}\right)^{\text{off}}}{\left\Vert \left(I_{d}-R_{ij}\right)^{\text{off}}\right\Vert },Q_{ij}^{ss}\right\rangle \right),
\end{eqnarray*}
where $l\left(Q_{ij}^{ss}\right)\geq0$ and
\begin{equation}
Q_{ij}^{ss}=Q_{i}^{1,ss}+Q_{j}^{2,ss}=Q_{i}^{1,ss}-Q_{j}^{1,ss}.\label{eq:Q_ss_ij}
\end{equation}
Define two sets as $$\mathcal{S}_{1}=\left\{ i=1,2,\ldots,n|\left\Vert Q_{i}^{1,ss}\right\Vert \gg\epsilon\right\} $$
and $$\mathcal{S}_{2}=\left\{ i=1,2,\ldots,n|\left\Vert Q_{i}^{1,ss}\right\Vert =\mathcal{O}\left(\epsilon\right)\right\}, $$
then from (\ref{eq:Q_ss_ij}) we obtain
\[
Q_{ij}^{ss}\approx Q_{i}^{1,ss},\text{ if }i\in\mathcal{S}_{1}\text{ and }j\in\mathcal{S}_{2}.
\]
The set $\mathcal{S}_{1}$ is assumed to be not empty, otherwise it
is easy to see $\left\Vert Q\right\Vert ^{2}=\mathcal{O}_{P}\left(n^{2}\epsilon^{2}\right)$
and thus $\left\Vert \Delta\right\Vert ^{2}=\mathcal{O}_{P}\left(n^{2}\epsilon^{2}\right)$.
In addition, $\#\mathcal{S}_{1}\ll n$ since $\sum_{i}\left\Vert Q_{i}^{1,ss}\right\Vert =\mathcal{O}\left(n\epsilon\right)$.

For every $i$, $\#\left\{ j|\left(i,j\right)\in\mathcal{E}_{c},\left\langle \frac{\left(I_{d}-R_{ij}\right)^{\text{off}}}{\left\Vert \left(I_{d}-R_{ij}\right)^{\text{off}}\right\Vert },Q_{i}^{1,ss}\right\rangle <0\right\} \geq c_{9}pn$,
thus for every $i\in\mathcal{S}_{1}$we have
\[
\#\left\{ j|\left(i,j\right)\in\mathcal{E}_{c},j\in\mathcal{S}_{2},\left\langle \frac{\left(I_{d}-R_{ij}\right)^{\text{off}}}{\left\Vert \left(I_{d}-R_{ij}\right)^{\text{off}}\right\Vert },Q_{i}^{1,ss}\right\rangle <0\right\} \geq c_{9}pn-\#\mathcal{S}_{1}\geq c_{10}pn.
\]
Thus
\begin{eqnarray}
\sum_{\left(i,j\right)\in\mathcal{E}_{c}}l\left(Q_{ij}^{ss}\right) & \geq & 2\sum_{\left(i,j\right)\in\mathcal{E}_{c}\cap\left(\mathcal{S}_{1}
\times\mathcal{S}_{2}\right)}l\left(Q_{ij}^{ss}\right)\nonumber \\
 & \geq & 2\sum_{i\in\mathcal{S}_{1}}\sum_{j\in\left\{ j|\left(i,j\right)\in\mathcal{E}_{c},j\in\mathcal{S}_{2},\left\langle \frac{\left(I_{d}-R_{ij}\right)^{\text{off}}}{\left\Vert \left(I_{d}-R_{ij}\right)^{\text{off}}\right\Vert },Q_{i}^{1,ss}\right\rangle <0\right\} }l\left(Q_{ij}^{ss}\right)\nonumber \\
 & \geq & 2\sum_{i\in\mathcal{S}_{1}}c_{11}pn\left\Vert Q_{i}^{1,ss}\right\Vert =c_{12}n\sum_{i\in\mathcal{S}_{1}}\left\Vert Q_{i}^{1,ss}\right\Vert .\label{eq:f2_off_part2_ctn}
\end{eqnarray}
Combine (\ref{eq:f_2_off}), (\ref{eq:f_2_off_part1}), (\ref{eq:f2_off_part2})
and (\ref{eq:f2_off_part2_ctn}) together and we obtain
\begin{equation}
f_{2}^{\text{off}}\geq c_{12}n\sum_{i\in\mathcal{S}_{1}}\left\Vert Q_{i}^{1,ss}\right\Vert -\mathcal{O}_{P}\left(\sqrt{n\log n}\right)\sum_{i}\left\Vert Q_{ii}^{1}\right\Vert -c_{13}n\text{Tr}\left(T\right).\label{eq:f2_off_bound_1}
\end{equation}
Apply Lemma \ref{lem:d_and_off} to (\ref{eq:f2_d_rhs}) and (\ref{eq:f2_off_bound_1}), and set $\alpha=\sqrt{1-\epsilon_{0}}$
and $\beta=\sqrt{\epsilon_{0}}$ , where $\epsilon_{0}\ll p-p_{c}$,
then we obtain
\begin{eqnarray*}
f_{2} & \geq & \sqrt{1-\epsilon_{0}}f_{2}^{\text{d}}+\sqrt{\epsilon_{0}}f_{2}^{\text{off}}\\
 & \geq & \sqrt{1-\epsilon_{0}}\left(\frac{pn}{\sqrt{d}}-\mathcal{O}_{P}\left(\sqrt{n\log n}\right)\right)\text{Tr}\left(T\right)\\
&&+\sqrt{\epsilon_{0}}\left(c_{12}n\sum_{i\in\mathcal{S}_{1}}\left\Vert Q_{i}^{1,ss}\right\Vert -\mathcal{O}_{P}\left(\sqrt{n\log n}\right)\sum_{i}\left\Vert Q_{ii}^{1}\right\Vert -c_{13}n\text{Tr}\left(T\right)\right).
\end{eqnarray*}
If $f_{1}+f_{2}\geq0$, then $G+\Delta$ is not the minimizer. And
$f_{1}+f_{2}\geq0$ leads to the condition that
\[
n\sum_{i\in\mathcal{S}_{1}}\left\Vert Q_{i}^{1,ss}\right\Vert \geq\mathcal{O}_{P}\left(\sqrt{n}\right)\sum_{i}\left\Vert Q_{ii}^{1}\right\Vert .
\]
Thus if $G+\Delta$ is the minimizer, then $\Delta$ must satisfy
the condition that
\[
n\sum_{i\in\mathcal{S}_{1}}\left\Vert Q_{i}^{1,ss}\right\Vert \leq\mathcal{O}_{P}\left(\sqrt{n}\right)\sum_{i}\left\Vert Q_{ii}^{1}\right\Vert \leq\mathcal{O}_{P}\left(n\sqrt{n}\right)\epsilon,
\]
that is,
\[
\sum_{i\in\mathcal{S}_{1}}\left\Vert Q_{i}^{1,ss}\right\Vert \leq\mathcal{O}_{P}\left(\sqrt{n}\right)\epsilon
\]
therefore
\begin{eqnarray*}
\sum_{i}\left\Vert Q_{i}^{1,ss}\right\Vert ^{2} & = & \sum_{i\in\mathcal{S}_{1}}\left\Vert Q_{i}^{1,ss}\right\Vert ^{2}+\sum_{i\in\mathcal{S}_{2}}\left\Vert Q_{i}^{1,ss}\right\Vert ^{2}\\
 & \leq & \left(\sum_{i\in\mathcal{S}_{1}}\left\Vert Q_{i}^{1,ss}\right\Vert \right)^{2}+\sum_{i\in\mathcal{S}_{2}}\left\Vert Q_{i}^{1,ss}\right\Vert ^{2}\\
 & \leq & \mathcal{O}_{P}\left(n\right)\epsilon^{2}+n\mathcal{O}_{P}\left(\epsilon^{2}\right)=\mathcal{O}_{P}\left(n\right)\epsilon^{2}
\end{eqnarray*}
This, together with the decomposition (\ref{eq:Delta2_decompostion}) and arguments similar to (\ref{eq:Q_sym})-(\ref{eq:Q_sym_bound}), finishes the proof.
%\clearpage

\bibliographystyle{plain}
\bibliography{sync}
\end{document}